\documentclass[sigconf]{acmart}

\usepackage{algorithmic}
\usepackage[ruled, vlined, commentsnumbered, linesnumbered]{algorithm2e}
\usepackage{mathdots}
\usepackage{amsfonts}
\usepackage{graphicx}
\usepackage{CJKutf8}
\usepackage[normalsize]{subfigure}
\usepackage{wrapfig}
\usepackage{mathdots}
\usepackage{appendix}
\usepackage{multirow}

\newcommand{\DEL}[1]{\iffalse #1 \fi}
\newcommand{\revision}{\color{red}}
\newcommand{\revisiondone}{\color{black}}

\newcommand{\squishlist}{
\begin{list}{$\bullet$}
  { \setlength{\itemsep}{0pt}
     \setlength{\parsep}{0pt}
     \setlength{\topsep}{0pt}
     \setlength{\partopsep}{0pt}
     \setlength{\leftmargin}{0em}
     \setlength{\labelwidth}{0em}
     \setlength{\labelsep}{0.2em} } }

\newtheorem{assumption}{\bf Assumption}
\newtheorem{theorem}{\bf{Theorem}}
\newtheorem{definition}{\bf{Definition}}

\newtheorem{lemma}{\bf{Lemma}}
\newtheorem{corollary}{\bf Corollary}

\newtheorem{reproposition}{\bf{Proposition}}
\newtheorem{proposition}{\bf{Proposition}}

\SetKwRepeat{Do}{do}{while}
\newcommand{\squishlisttwo}{
\begin{list}{$\bullet$}
  { \setlength{\itemsep}{0pt}
     \setlength{\parsep}{0pt}
    \setlength{\topsep}{0pt}
    \setlength{\partopsep}{0pt}
    \setlength{\leftmargin}{2em}
    \setlength{\labelwidth}{1.5em}
    \setlength{\labelsep}{0.5em} } }

\begin{document}

\title[Context-Aware Metric Differential Privacy for Vehicle Trajectory Data]{Context-Aware Metric Differential Privacy for Vehicle Trajectory Data}


\author{Gaoyi Chen, Yan Huang, and Chenxi Qiu}
\affiliation{%
  \institution{University of North Texas}
  \city{Denton}
  \state{Texas}
  \country{USA}}


\renewcommand{\shortauthors}{Trovato et al.}

\begin{abstract}
\emph{Metric Differential Privacy (mDP)} generalizes differential privacy by allowing privacy guarantees to be expressed with respect to an arbitrary distance metric over secrets. While mDP has been adopted in geo-location protection, most existing mechanisms perturb each location record in isolation and do not model how contextual information (e.g., recent mobility history) affects the utility of the released data. This mismatch is particularly pronounced for vehicle mobility traces, where service quality often depends on temporally correlated locations.

In this paper, we propose \emph{Context-aware mDP (C-mDP)}, a framework for vehicle location privacy that incorporates contextual dependencies into both the utility model and the privacy notion. C-mDP treats the protected secret as a context-augmented record and enforces metric indistinguishability over this augmented domain. We formulate optimal C-mDP mechanism design as a \emph{linear program (LP)} that minimizes expected utility loss subject to C-mDP constraints. To improve scalability, we exploit \emph{conditional-independence} structure between the current location and contextual variables to derive a reduced formulation with substantially fewer decision variables and constraints. We evaluate C-mDP on real-world vehicle mobility datasets and compare it with standard mDP baselines. The results show that C-mDP consistently achieves higher utility under the same privacy budget while satisfying the required metric privacy guarantees.
\end{abstract}


\maketitle

\vspace{-0.0in}
\section{Introduction}
\label{sec:intro}
\vspace{-0.00in}

Within the spectrum of data privacy protection mechanisms, \emph{data perturbation} has emerged as a widely adopted approach for protecting users' sensitive information. The central idea is to deliberately inject noise into data, so that personal information remains unintelligible to unauthorized parties even in the event of server-side breaches. Among perturbation-based methods, \emph{Differential Privacy (DP)}~\cite{Dwork-ALP2006} has become the standard paradigm due to its rigorous, provable privacy guarantees. Differential privacy (DP) requires a mechanism to produce “indistinguishable” outputs on any two \emph{neighboring databases} that differ in at most one record (i.e., have Hamming distance at most one). This framework has been generalized to \emph{metric Differential Privacy (mDP)}, which accommodates arbitrary distance metrics over diverse data domains~\cite{Chatzikokolakis-PETS2013}. Unlike DP, mDP defines neighboring data records via an underlying metric space and scales the indistinguishability guarantees according to the (non-binary) distances between records. mDP has been successfully applied to the release of sensitive geo-location data~\cite{Andres-CCS2013}, 
using distance measures such as Manhattan, Euclidean, and Haversine distances. \looseness=-1

Compared to standard (Hamming-based) DP, optimizing mDP introduces additional challenges due to heterogeneous privacy requirements between neighboring records and the direction- and magnitude-dependent sensitivity of utility loss to perturbations. To minimize the utility loss induced by perturbations, several works on mDP have adopted a \emph{Linear Programming (LP)} framework~\cite{Bordenabe-CCS2014, Qiu-TMC2020, Qiu-IJCAI2024, Pappachan-EDBT2023, Qiu-EDBT2024}. These approaches primarily focus on discrete domains, where the utility loss associated with each perturbation outcome can be quantified explicitly, and the LP directly optimizes the perturbation probability distribution to minimize the expected utility loss. 
However, \emph{these mDP methods primarily concentrate on optimizing perturbations for individual records without accounting for how contextual information affects utility loss.} This limits the practical applicability of mDP in real-world scenarios, where the impact of data perturbations on data utility 
is influenced by the contextual aspects of the data. \looseness = -1

\begin{figure}[t]
\centering
\hspace{0.00in}
\begin{minipage}{0.48\textwidth}
  \subfigure{
\includegraphics[width=0.90\textwidth]{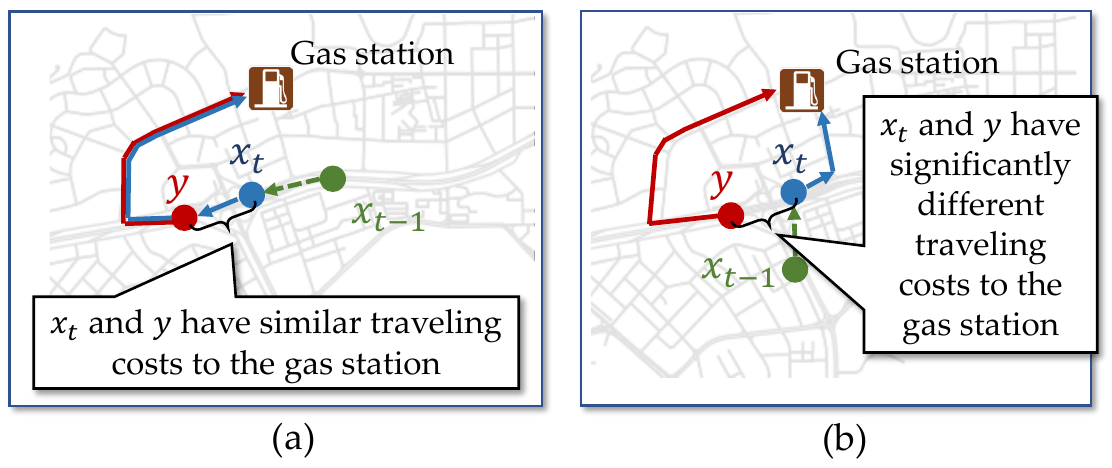}}
\vspace{-0.10in}
\end{minipage}
\caption{Example: Different context information could cause different impacts in geo-location perturbation. \\
$x_t$ and $x_{t-1}$ are the current location and the location in the last time slot, respectively. $y$ is the perturbed location of $x_t$. }
\label{fig:example1}
\vspace{-0.15in}
\end{figure}

\vspace{0.05in}
\noindent \textbf{Motivating example}. Consider a \emph{location-based service (LBS)} that recommends the nearest \emph{point of interest (POI)}, such as a gas station, to a moving vehicle. To provide recommendations, the vehicle shares its location with the LBS server, which then estimates the traveling distance from the vehicle to candidate POIs. The utility of a reported (perturbed) location depends on how accurately the server can estimate the actual travel distance from the vehicle to each POI using the reported location.

As illustrated in Fig. \ref{fig:example1}(a) and (b), the vehicle’s context, including its current location $x_{t}$ and preceding location $x_{t-1}$, plays an important role in determining utility. In Fig. \ref{fig:example1}(a), the vehicle is heading in a south-west direction, while in Fig. \ref{fig:example1}(b), it is moving north-east. If the current location $x_t$ is perturbed to the same point $y$, which is positioned south-west of $x_t$, the resulting utility losses differ significantly between the two contexts: In Fig. \ref{fig:example1}(a), the estimated travel distance from $y$ to the gas station is similar to the actual travel distance from $x_t$ to the gas station, leading to a low utility loss. Conversely, in Fig. \ref{fig:example1}(b), the estimated travel distance from $y$ to the gas station deviates significantly from the actual distance, resulting in a higher utility loss.

This example shows that the utility loss incurred when perturbing the current location $x_t$ to $y$ is context-dependent: it is much smaller in Fig.~\ref{fig:example1}(a) than in Fig.~\ref{fig:example1}(b). Motivated by this observation, we aim to minimize perturbation-induced utility loss \emph{by selecting perturbed locations according to their context-dependent utility}. For instance, in scenarios like Fig.~\ref{fig:example1}(a), location $y$ should be assigned higher probability because it leads to lower utility loss than it does in scenarios like Fig.~\ref{fig:example1}(b). \looseness = -1

  
\DEL{
\vspace{0.02in}
\noindent \textbf{Example B (Text perturbation)}. 
Consider the different impacts of changing the word from ``\textbf{explore}'' to ``\textbf{travel}'' in the following two sentences:
\vspace{-0.00in}
\begin{itemize}
\item \emph{She loves to \textbf{explore} to new places of interest every year.}
\vspace{-0.00in}
\item \emph{The explorer decided to \textbf{explore} different territories in search of undiscovered species.} 
\vspace{-0.00in}
\end{itemize}
In the first sentence, changing ``explore'' to ``travel'' doesn't make much difference; it still conveys the idea of visiting new places. However, in the second sentence, changing ``explore'' to ``travel'' significantly alters the meaning, as the emphasis is on actively investigating and discovering in the original sentence, while ``travel'' suggests a more passive action of going from one place to another without necessarily seeking out new territories.}

\vspace{-0.04in}
\subsection{Our Work}
\vspace{-0.00in}
\textbf{Contribution 1: \underline{C}ontext-aware \underline{mDP} (C-mDP)}. 
To fill the identified research gap, this paper introduces a new data perturbation framework called \emph{C-mDP}. In contrast to conventional mDP, where the selection of perturbed data relies solely on the secret target record, C-mDP incorporates contextual information when determining the perturbation data distribution. 
By incorporating this contextual data, the framework explicitly accounts for its impact on utility loss caused by perturbation, an aspect that, to the best of our knowledge, has not been explored in other LP-based approaches. At the same time, C-mDP ensures that information disclosure remains bounded by a predefined threshold, as formally demonstrated in \textbf{Proposition \ref{prop:ContextDP}}.

\smallskip
\noindent\textbf{Contribution 2: Efficient computation of C-mDP}. Like \cite{Fawaz-CCS2014}, we frame the task of optimizing C-mDP as an LP problem, of which the objective is to minimize the expected data utility loss while ensuring the indistinguishability criterion of neighboring records. 

Considering that directly incorporating context data into the perturbation derivation significantly increases the complexity of LP, 
we reduce C-mDP by including only the ``\emph{Markov blanket}'' (Definition \ref{def:MarkovBlanket}) of the target secret records when optimizing the perturbation distribution of secret records. We theoretically prove that the reduced C-mDP can still satisfy the desired privacy criterion (\textbf{Proposition \ref{prop:CDcorrect}}) while achieving the minimum data utility loss (\textbf{Proposition \ref{propo:CDutility}}). We also design a \emph{Markov Blanket Identification (MBI)} framework to identify the Markov blanket of secret data via \emph{conditional independence testing} \cite{Scetbon-ICML2022} and predict it using a \emph{deep neural network}. \looseness = -1


\smallskip
\noindent\textbf{Contribution 3: Empirical study of C-mDP in geo-location data privacy protection.} 
As an example, we apply C-mDP to protect vehicle location privacy in location-based services (LBSs), where contextual data includes vehicles' historical locations. Experimental results using two real-world taxicab datasets from "Rome, Italy" \cite{roma-taxi-20140717} and "Porto, Portugal" \cite{taxiPorto} show that (1) the impact of context information on utility loss in LBS varies across time, regions, speed ranges, and cities, and (2) incorporating context information into existing LP-based perturbation reduces average data utility loss by at least 15.6\% and 4.9\% on the Rome and Porto datasets, respectively, compared to the benchmarks \cite{Qiu-TMC2020, ImolaUAI2022, McSherry-FOCS2007}.

The rest of the paper is organized as follows: Section \ref{sec:preliminaries} gives the preliminaries of mDP. Sections \ref{sec:CmDP} and \ref{sec:MBI} introduce the problem formulation and the algorithm design of C-mDP. Section \ref{sec:exp} evaluates the performance of C-mDP. Section \ref{sec:related} presents the related work. Finally,  Section \ref{sec:conclude} concludes the paper.

\vspace{-0.00in}
\section{Technical Preliminaries}
\vspace{-0.00in}
\label{sec:preliminaries}
In this section, we first introduce the preliminaries of mDP (Section \ref{subsec:mDP}), the threat model and the countermeasure (Section \ref{subsec:threatmodel}), and the LP-based computation framework (Section \ref{subsec:LP}).  Table \ref{Tb:Notationmodel} in Appendix \ref{app:sec:notations} lists the main notations used throughout this paper.

\vspace{-0.00in}
\subsection{mDP}
\label{subsec:mDP}
\vspace{-0.03in}
Generally, a data perturbation mechanism can be represented as a \emph{probabilistic function} $Q$: $\mathcal{X} \rightarrow \mathcal{Y}$, where $\mathcal{X}$ and $\mathcal{Y}$ are the \emph{secret dataset} and the \emph{perturbed dataset}, respectively. We define the measure $d: \mathcal{X}^2 \rightarrow \mathbb{R}$ to quantify the distance between records in $\mathcal{X}$ and denote the distance between any two secret records $x_i, x_j \in \mathcal{X}$ by $d_{x_i, x_j}$. We call two records $x_i, x_j \in \mathcal{X}$ \emph{neighbors} if their distance $d_{x_i, x_j} \leq \eta$, where $\eta >0$ is a pre-determined threshold. We use the random variable $X$ to represent the secret data (or secret record), and $Q(X)$ to represent its perturbed data. 
\vspace{-0.03in}
\begin{definition}
\label{def:mDP}
(mDP~\cite{Andres-CCS2013}) For any pair of \emph{neighboring records} $x_i, x_j \in \mathcal{X}$ ($d_{x_i, x_j} \leq \eta$), $\epsilon$-mDP ensures that the probability distributions of their perturbed data $Q(x_i)$ and $Q(x_j)$ are sufficiently close so that it is hard for an attacker to distinguish $x_i$ and $x_j$ even if $Q(x_i)$ and $Q(x_j)$ are breached to the attacker, which can be represented mathematically by, 
\vspace{-0.00in}
\begin{equation}
\label{eq:mDP}
\frac{\Pr\left[Q(X) = y| X = x_i\right]}{\Pr\left[Q(X) = y| X = x_j\right]} \leq e^{\epsilon d_{x_i, x_j}}, ~\forall y \in \mathcal{Y}.  
\vspace{-0.00in}
\end{equation}
Here, $\epsilon > 0$ is called the \emph{privacy budget}, reflecting how much information about the secret data $X$ is allowed to be disclosed from its perturbed representation $Q(X)$, i.e., lower $\epsilon$ implies a higher privacy level. \looseness = -1
\end{definition}

\vspace{-0.05in}
\subsection{Threat Model and Countermeasure}
\label{subsec:threatmodel}
\vspace{-0.02in}
Like the existing works \cite{Andres-CCS2013}, we assume that both perturbed data $Q(X)$ and perturbation function $Q$ are known by attackers and users apply the function $Q$ to perturb their secret data. An attacker can use Bayes' formula \cite{Yu-NDSS2017} to derive the \emph{posterior} of the secret data $X$, i.e., $\Pr\left[X = x | Q(X) = y\right]$, $\forall x \in \mathcal{X}$, of which the accuracy can be quantified by the posterior leakage  (Definition \ref{def:PL}): \looseness = -1
\begin{definition}
\label{def:PL}
(Posterior leakage~\cite{Kifer-PODS2012} (PL)) Given the perturbation function $Q$, the PL of any pair of neighboring records $x_i, x_j \in \mathcal{X}$ s.t. $d_{x_i, x_j} \leq \eta$, 
\vspace{-0.00in}
\normalsize
\begin{eqnarray}
\label{eq:PL}
\mathrm{PL}\left((x_i, x_j); Q\right) = \small \sup_{y} |\ln \underbrace{\left(\frac{\Pr\left[X = x_i | Q(X) = y\right]}{\Pr\left[X = x_j| Q(X) = y\right]}\right.}_{\mbox{posterior ratio}}\left\slash\underbrace{\left.\frac{\Pr\left[X=x_i\right]}{\Pr\left[X=x_j\right]}\right)}_{\mbox{prior ratio}}\right.| 
\end{eqnarray}
\normalsize
\end{definition}
Intuitively, the prior ratio and the posterior ratio in Eq. (\ref{eq:PL}) reflect the record $X$'s probabilities of being $x_i$ and $x_j$ \textbf{before and after the observation of the perturbed data $Q(X) =y$}. If 
$\mathrm{PL}\left((x_i, x_j); Q\right)$ has a lower value, it implies that the attacker can obtain less additional information of $X$ by observing $Q(X)$, therefore achieving a higher privacy level. As a countermeasure against the Bayesian inference attacks, the perturbation function $Q$ is designed to enforce the posterior leakage between any $x_i$ and $x_j$ to be bounded by a threshold: \looseness = -1
\begin{equation}
\label{eq:PLcriterion}
\mathrm{PL}\left((x_i, x_j); Q\right) \leq \epsilon d_{x_i,x_j}, \forall x_i, x_j\in \mathcal{X},  
\end{equation}
meaning that the perturbed data $Q(X)$ only discloses limited additional information to attackers. As proved by \cite{Andres-CCS2013}, meeting mDP as defined in Eq. (\ref{eq:mDP}) is \textbf{sufficient} to achieve the PL bound in Eq. (\ref{eq:PLcriterion}).

\vspace{-0.00in}
\subsection{Optimization of mDP using LP}
\label{subsec:LP}
\vspace{-0.00in}
Like \cite{ImolaUAI2022, Qiu-TMC2020}, we consider the case that both $\mathcal{X}$ and $\mathcal{Y}$ are finite. As such, the perturbation function $Q$ can be represented as the \emph{perturbation matrix} $\mathbf{Q} = \left\{q_{x_i,y_k}\right\}_{(x_i,y_k) \in \mathcal{X}\times \mathcal{Y}}$, where each entry $q_{x_i,y_k}$ denotes the probability of selecting $y_k \in \mathcal{Y}$ as the perturbed data given the real record $x_i \in \mathcal{X}$. 
In this case, the mDP constraints formulated in Eq. (\ref{eq:mDP}) can be written as the following linear constraints: For each pair of neighboring records $x_i, x_j \in \mathcal{X}$ s.t. $d_{x_i, x_j} \leq \eta$, 
\vspace{-0.00in}
\begin{equation}
\label{eq:mDPdiscrete}
\frac{q_{x_i,y_k}}{q_{x_j,y_k}} \leq e^{\epsilon d_{x_i,x_j}}, ~\forall y_k \in \mathcal{Y}. 
\vspace{-0.00in}
\end{equation}
Additionally, the sum probability of perturbed record $y_k \in \mathcal{Y}$ for each real record $x_i$ should be equal to 1 (\emph{probability unit measure}), i.e., 
\vspace{-0.10in}
\begin{equation}
\label{eq:um}
\sum_{y_k \in \mathcal{Y}}q_{x_i,y_k} = 1,~ \forall x_i \in \mathcal{X}.
\end{equation}
We use $c_{x_i,y_k}$ to represent the \emph{data utility loss} of the downstream decision-making caused by the perturbed record $y_k$ when the real record is $x_i$. 

In practice, the assessment of each data utility loss, $c_{x_i,y_k}$, depends on the specific manner in which the data is utilized in downstream decision-making processes. In the performance evaluation detailed in Section~\ref{sec:exp}, we focus on \emph{location-based services (LBS)} where vehicles are required to physically travel to fulfill tasks in spatial crowdsourcing scenarios~\cite{Qiu-TMC2020} (e.g., picking up passengers). 
If a user's reported location is inaccurate, the server will use the inaccurate location to estimate the travel cost from this user to the destination, causing estimation errors; in this case, the server might 
assign a vehicle that is too far away from the passenger. 
{\revisiondone In such cases, the utility loss $c_{x_i,y_k}$ can be quantified using the discrepancy between the estimated and actual travel costs. 
By minimizing the utility loss, the system improves the likelihood of recommending the vehicle closest to the spatial task. 
 \looseness =-1
}

The \emph{loss function} of $\mathbf{Q}$, measuring the expected utility loss caused by the perturbation matrix $\mathbf{Q}$, can be defined as 
\vspace{-0.02in}
\begin{equation}
\label{eq:losscontextfree}
\mathcal{L}(\mathbf{Q}) = \sum_{x_i\in \mathcal{X}, y_k\in \mathcal{Y}}  P(X = x_i) \cdot c_{x_i,y_k} \cdot q_{x_i,y_k},
\end{equation}
which is a linear function of $\mathbf{Q}$. The goal of optimizing the perturbation matrix $\mathbf{Q}$ is to minimize $\mathcal{L}(\mathbf{Q})$ while satisfying both the mDP (Eq. (\ref{eq:mDPdiscrete})) and the probability unit measure (Eq. (\ref{eq:um})), which can be formulated as the following LP problem: 
\begin{eqnarray}
\label{eq:LPobjective}
\min && \mathcal{L}(\mathbf{Q}) \\ \label{eq:LPconstraint2}
\mathrm{s.t.} && \mbox{Eq. (\ref{eq:mDPdiscrete})(\ref{eq:um}) are satisfied}, 
\end{eqnarray}
where each entry $q_{x_i,y_k}$ in $\mathbf{Q}$ satisfies $0 \leq q_{x_i,y_k} \leq 1$.

\vspace{-0.02in}
\section{Context-Aware mDP}
\vspace{-0.00in}
\label{sec:CmDP}

We adopt an LP-based approach to optimize the selection probabilities of perturbed outputs, with the objective of minimizing the expected utility loss. This formulation requires explicitly quantifying the utility loss associated with each possible perturbed data point. As illustrated in Fig.~\ref{fig:example1}, the impact of perturbation on data utility may vary significantly across different contexts, highlighting the necessity of incorporating contextual information when evaluating utility loss and determining perturbation probabilities. 

Motivated by this observation, we propose a new data perturbation framework, termed \emph{\underline{C}ontext-aware \underline{mDP} (C-mDP)}. We first define the C-mDP problem in Section \ref{subsec:formulation} and then present a reduction in problem complexity in Section \ref{subsec:complexityreduce}. 

\vspace{-0.00in}
\subsection{Problem Formulation}
\label{subsec:formulation}
\vspace{-0.00in}

In C-mDP, the perturbation distribution for a secret record $X$ is allowed to depend not only on $X$ itself but also on contextual information $V_X$ associated with $X$. Consider Fig.~\ref{fig:example1}: when perturbing a vehicle’s current location $x_t$, the probability of reporting each perturbed location can vary with the prior location $x_{t-1}$, since different $x_{t-1}$ values indicate different movement directions and mobility patterns.

More broadly, we treat $V_X$ as any context that can affect the utility (and potentially the effective privacy risk) of releasing a perturbed version of $X$. In our vehicle-location case study (Section~\ref{sec:exp}), $X$ is the vehicle’s current location at time $t$, and $V_X$ consists of its historical locations. We use $V_X$ to estimate the distribution of future locations, which is essential for supporting high-quality LBS under realistic latency. Importantly, while imperfect modeling of $V_X$ may introduce error in utility-loss estimation, it does not weaken the formal privacy guarantee: the mDP constraints are enforced directly on the mechanism. At the same time, incorporating historical context enables a more faithful assessment of utility impact than context-agnostic perturbation.


We use $\mathcal{V}$ to denote the space of contextual information (i.e., $V_X\in\mathcal{V}$). 
We assume that both the secret record $X$ and its associated context $V_X$ are available to the user, which runs the perturbation locally, while the server observes only the released output and does not directly observe the realized $V_X$. 
Accordingly, the perturbation function $Q$ is defined as a mapping $Q:\mathcal{X}\times\mathcal{V}\rightarrow\mathcal{Y}$, and the released (perturbed) value is denoted by $Q(X,V_X)$ given the real record $X$ and its context $V_X$. 
The corresponding perturbation matrix is
\begin{equation}
\mathbf{Q}=\left\{q_{(x_i,\mathbf{v}),y_k}\right\}_{(x_i,\mathbf{v},y_k)\in\mathcal{X}\times\mathcal{V}\times\mathcal{Y}},
\end{equation}
where each entry $q_{(x_i,\mathbf{v}),y_k}$ denotes the probability of outputting $y_k$ when the real record is $X=x_i$ and the context is $V_X=\mathbf{v}$.

\subsubsection{Expected Data Utility Loss} 
We use $c_{(x_i,\mathbf{v}),y_k}$ to denote the utility loss incurred when the real record is $x_i$ with context $\mathbf{v}$ and the released (perturbed) record is $y_k$. Accordingly, the expected utility loss induced by the perturbation matrix $\mathbf{Q}$ is
\vspace{-0.00in}
\normalsize
\begin{eqnarray}
\label{eq:CA-UL}
\mathcal{L}(\mathbf{Q}) = \sum_{(x_i, \mathbf{v}), y_k} p_{(x_i, \mathbf{v})} \cdot c_{(x_i,\mathbf{v}),y_k} \cdot q_{(x_i,\mathbf{v}),y_k},
\end{eqnarray}
where $p_{(x_i,\mathbf{v})}$ is the prior probability that the secret record and its context are $(x_i,\mathbf{v})$.


In practice, the assessment of each data utility loss $c_{(x_i,\mathbf{v}),y_k}$ is contingent on the specific manner in which data is processed in downstream decision-making. In the performance evaluation in Section \ref{sec:exp}, we consider the LBS services where vehicles need to physically travel to designated locations to receive desired services such as navigation \cite{To-TMC2017}, or to fulfill tasks in spatial crowdsourcing \cite{Qiu-TMC2020}. In those applications, $c_{(x_i,\mathbf{v}),y_k}$ can be quantified by the discrepancy between the estimated travel cost (from $y_k$ to the destination) and actual travel cost (from $x_i$ to the destination). 

{\revisiondone While we use vehicle geo-location protection as a concrete instantiation, the proposed framework is not limited to this scenario. It readily extends to other applications with only minor modifications, as long as one can characterize how perturbing the released data affects the resulting utility loss (e.g., through an application-specific loss function or objective).} 

\vspace{0.00in}
\subsubsection{Privacy Constraints}
\label{subsec:privacyconstraints}
In Definition \ref{def:CAPL}, we extend the PL constraints in Definition \ref{def:PL} to the context-aware PL constraints:
\vspace{-0.05in}

\begin{definition}
\label{def:CAPL}
(Context-aware PL and its bound) Using the perturbation function $Q$, the context-aware PL of any pair of neighboring records $x_i, x_j\in \mathcal{X}$ given their context data $\mathbf{v}$ and $\mathbf{v}'$ is defined by 
\vspace{-0.03in}
\normalsize
\begin{eqnarray}
\label{eq:PLcontext}
&& \mathrm{PL}\left((x_i, \mathbf{v}), (x_j, \mathbf{v}')\right) \\ \nonumber 
&=& {\small  \sup_{y} |\ln \underbrace{\left(  \frac{\scriptsize \Pr \left[(X, V_X) = (x_i, \mathbf{v}) | Q(X, V_X) = y\right]}{\scriptsize \Pr\left[(X, V_X) = (x_j, \mathbf{v}')| Q(X, V_X) = y\right]}\right.}_{\mbox{posterior ratio}}\left\slash\underbrace{\left.\frac{p_{(x_i, \mathbf{v})}}{p_{(x_j, \mathbf{v}')}}\right)}_{\mbox{prior ratio}}\right.|} 
\end{eqnarray}
\normalsize
The corresponding context-aware privacy-loss (PL) bound is
\vspace{-0.03in}
\begin{equation}
\label{eq:PLcontextbound}
\mathrm{PL}\!\left((x_i,\mathbf{v}),(x_j,\mathbf{v}')\right)
~\le~
\epsilon\, d_{(x_i,\mathbf{v}),(x_j,\mathbf{v}')},
\end{equation}
where $d_{(x_i,\mathbf{v}),(x_j,\mathbf{v}')}$ is the context-aware distance between the augmented secrets $(x_i,\mathbf{v})$ and $(x_j,\mathbf{v}')$. We define this distance as an \emph{augmentation} of the base Haversine distance between individual locations (i.e., the great-circle distance between two geographic points on the Earth computed from their latitudes and longitudes) as follows:
{\revisiondone
\begin{equation}
\label{eq:context_metric}
d_{(x_t,\mathbf{v}),(x_t',\mathbf{v}')}
~\triangleq~
d_{x_t,x_t'}
~+~
\sum_{\tau=1}^{\Gamma} w_{t-\tau}\, d_{v_{t-\tau},v'_{t-\tau}},
\end{equation}
where $\mathbf{v} = (v_{t-1}, ..., v_{t-\Gamma})$ and $w_{t-\tau}$ ($\tau = 1,..., \Gamma$) are decay weights that control the relative importance of protecting past locations.}
\end{definition}
\vspace{-0.05in}

{\revisiondone Notably, while incorporating context can improve utility modeling and utility in practice, it may also cause the released output $Y$ to reveal information about the (unobserved) context $V_X$. Our framework explicitly accounts for this via the context-aware PL definition over $(X,V_X)$ and the corresponding C-mDP constraints, which bound the adversary's posterior gain for any pair of hypotheses $(x,\mathbf{v})$ and $(x',\mathbf{v}')$.}

\begin{proposition}
\label{prop:ContextDP}
To satisfy the context-aware PL bound in Eq. (\ref{eq:PLcontextbound}), it is sufficient to enforce the C-mDP constraints
\vspace{-0.00in}
\begin{equation}
\label{eq:C-mDP}
q_{(x_i,\mathbf{v}),y_k} - e^{\epsilon d_{(x_i,\mathbf{v}),(x_j,\mathbf{v}')}} \cdot q_{(x_j,\mathbf{v}'),y_k} \leq 0, ~\forall y_k \in \mathcal{Y}, 
\vspace{-0.00in}
\end{equation}
for every pair of neighboring records $(x_i,\mathbf{v}), (x_j,\mathbf{v}') \in \mathcal{X}\times \mathcal{V}$, i.e., those with
$d_{(x_i,\mathbf{v}),(x_j,\mathbf{v}')} \le \eta$. 
\end{proposition}
\vspace{-0.05in}
Due to the limited space, the detailed proof of Proposition \ref{prop:ContextDP}, as well as the proofs of Propositions \ref{prop:CDcorrect}--
\ref{propo:CDutility} presented subsequently, are available in Appendix. 


\vspace{0.00in}
\subsubsection{LP Formulation}
Given the newly defined data utility loss (Eq. (\ref{eq:CA-UL})) and C-mDP constraints (Eq. (\ref{eq:C-mDP})), we formulate the \textbf{C-mDP problem} as the following LP problem:  
\normalsize
\begin{eqnarray}
\label{eq:C-mDPObj}
\min && \mathcal{L}(\mathbf{Q}) = \sum_{(\mathbf{v}, x_i), y_k} p_{(x_i, \mathbf{v})} \cdot c_{(x_i,\mathbf{v}),y_k} \cdot q_{(x_i,\mathbf{v}),y_k} \\ \label{eq:C-mDP1}
\mathrm{s.t.} && \mbox{C-mDP (Eq. (\ref{eq:C-mDP})) is satisfied} \\
\label{eq:C-mDPUM}
&& \textstyle \sum_{y_k \in \mathcal{Y}}q_{(x_i,\mathbf{v}),y_k} = 1,~ \forall x_i, \mathbf{v} \\ \label{eq:CLPconstraint1}
&& q_{(x_i,\mathbf{v}),y_k} \in [0, 1], \forall x_i, y_k, \mathbf{v}. 
\end{eqnarray}
\normalsize
The decision variables in the above LP formulation (Eq. (\ref{eq:C-mDPObj})--(\ref{eq:CLPconstraint1})) are the matrix $\mathbf{Q} = \left\{q_{(x_i,\mathbf{v}),y_k}\right\}_{(x_i, \mathbf{v}, y_k) \in \mathcal{X} \times \mathcal{V}\times \mathcal{Y}}$, including a total of $O(|\mathcal{X}||\mathcal{Y}||\mathcal{V}|)$ decision variables (entries) and $O(|\mathcal{X}|^2|\mathcal{Y}||\mathcal{V}|)$ linear constraints, where $|\cdot|$ denotes the set cardinality. 
Such a high complexity of the LP formulation renders it hard to apply in large-scale or time-sensitive applications.

{\revisiondone
\smallskip
\noindent\textbf{Discussion (C-mDP vs.\ context-free mDP).}
Notably, standard mDP mechanisms considered in our baselines are \emph{context-free}: they learn a single perturbation distribution $q_{x,y}$ for each location $x\in\mathcal{X}$ and apply it regardless of any (unobserved) contextual information. In our notation, this is equivalent to additionally enforcing
\begin{equation}
q_{(x,\mathbf{v}),y} = q_{(x,\mathbf{v}'),y}, \quad \forall x\in\mathcal{X},\ \forall \mathbf{v},\mathbf{v}'\in\mathcal{V},\ \forall y\in\mathcal{Y},
\end{equation}
i.e., the perturbation distribution is invariant across contexts. By contrast, C-mDP allows context-conditioned mechanisms $q_{(x,\mathbf{v}),y}$ (or $q_{(x,\mathbf{b}),y}$ under the CD policy), while still enforcing likelihood-ratio constraints with the same privacy-budget parameter $\epsilon$ under an appropriate metric. When utility depends on mobility context (e.g., direction/speed affecting future-location prediction), removing the invariance constraint can reduce expected utility loss, which helps explain why our context-aware mechanisms can outperform context-free mDP baselines in utility.}

\vspace{-0.05in}
\subsection{Problem Complexity Reduction}
\label{subsec:complexityreduce}
\vspace{-0.02in}
In this section, we describe the method to reduce the complexity of C-mDP by leveraging the \emph{independence} relationships between the secret data and the context information. We first give some key definitions as follows: 
\begin{definition}
\label{def:CI}
(Conditional independence (CI)) Given a random variable $X$ and two sets of random variables ${A}$ and ${B}$,  $X$ and ${A}$ are called ``\textbf{conditionally independent given ${B}$}'' if and only if $\Pr\left[{B}\right]>0$ and $\Pr\left[X|{A}, {B}\right] = \Pr\left[X|{B}\right]$, 
written as: $X  \perp \!\!\!\perp {A} |~{B}$. 
\end{definition}
\DEL{
\begin{definition}
(Markov random field) Given a set of random variables ${V}$, an undirected graph $\mathcal{G} = ({V}, \mathcal{E})$ forms a \textbf{Markov random field} if, $\forall X, X_j \in {V}$
\begin{equation}
X  \perp \!\!\!\perp X_j |{V}\backslash \left\{X, X_j\right\} 
\end{equation}
if and only if $\left(X, X_j \right) \notin \mathcal{E}$.
\end{definition}}

\begin{definition}
\label{def:MarkovBlanket}
(Markov blanket) Given the random variable $X$, we call the set of random variables ${B}_{X} \subseteq V_X$ a Markov blanket of $X$ in $V_X$ if its complement set ${B}_{X}^\mathrm{c} = {V}_X\backslash {B}_{X}$ satisfies $X  \perp \!\!\!\perp {B}_{X}^\mathrm{c} |{B}_{X}$, i.e., $X$ and ${B}_{X}^\mathrm{c}$ are conditionally independent given ${B}_X$. 
\end{definition}


Based on Definitions~\ref{def:CI} and~\ref{def:MarkovBlanket}, we next formalize a \emph{conditionally dependent (CD)} policy that restricts how the perturbation mechanism may use contextual information.

\begin{definition}
\label{def:CD}
(CD policy) Let $X$ be a secret record with associated context variables $V_X$. Let $B_X \subseteq V_X$ be a Markov blanket of $X$ within $V_X$, and let $B_X^{\mathrm{c}} \triangleq V_X \setminus B_X$ denote its complement. A perturbation function $Q$ follows the CD policy if
\begin{equation}
\label{eq:CDpolicy}
Q(X,V_X) \perp \!\!\!\perp B_X^{\mathrm{c}} \,\big|\, \{X,B_X\},
\end{equation}
i.e., conditioned on $X$ and $B_X$, the output distribution is independent of the remaining context variables. Equivalently, letting $\mathbf{b}=\pi(\mathbf{v})$ denote the restriction of $\mathbf{v}$ to $B_X$, for any $x\in\mathcal{X}$, $\mathbf{v}\in\mathcal{V}$, and $y\in\mathcal{Y}$,
\begin{eqnarray}
q_{(x,\mathbf{v}),y} &\triangleq& \Pr[Q(X,V_X)=y \mid X=x,V_X=\mathbf{v}]
\\ 
&=&
\Pr[Q(X,V_X)=y \mid X=x,B_X=\mathbf{b}]
\\ 
&\triangleq& q_{(x,\mathbf{b}),y}.
\end{eqnarray}
\end{definition}
Intuitively, the CD policy enables the computation of the perturbation function $Q$ to focus solely on the context information pertinent to the protected secret data $X$, which reduces the computational complexity without compromising data utility and privacy. 

By following the CD policy, 
the perturbation matrix is represented by 
$\mathbf{Q} = \left\{q_{(x_i,\mathbf{b}),y_k}\right\}_{(x_i,\mathbf{b},y_k) \in \mathcal{X} \times \mathcal{B}\times \mathcal{Y}}$, 
where each $q_{(x_i,\mathbf{b}),y_k}$ denotes the probability of selecting $y_k$ as the perturbed data given the real record $X = x_i$ and the context information $B_X = \mathbf{b}$. We denote the space of the Markov blanket $B_X$ as $\mathcal{B}$. By focusing exclusively on the context information within $\mathcal{B}$ rather than the entire context space $\mathcal{V}$, the complexity of C-mDP is reduced. 

Next, we present theoretical proofs that, under Assumption \ref{assu:datautility}, adhering to the CD policy preserves both the privacy criteria (as per \textbf{Proposition \ref{prop:CDcorrect}}) and the optimality (as per \textbf{Proposition \ref{propo:CDutility}}) of the perturbation matrix $\mathbf{Q}$.
\DEL{
we make the following assumption: 
\begin{assumption}
\label{assum:UL}
We target the applications where the quality of services solely depends on the distribution of secret records and perturbed data. In this case, the data utility loss $c_{(x_i,\mathbf{v}),y_k}$ can be represented as a function of $p_{x_i|\mathbf{v}}$ and $ q_{(x_i,\mathbf{v}),y_k}$, i.e. $h(p_{x_i|\mathbf{v}}, q_{(x_i,\mathbf{v}),y_k})$, indicating that 
\begin{eqnarray}
c_{(x_i,\mathbf{v}),y_k}  &=&  h(p_{x_i|\mathbf{v}}, q_{(x_i,\mathbf{v}),y_k}) \\
&=& h(p_{(x_i,\mathbf{b})}, q_{(x_i,\mathbf{b}),y_k}) \\ 
&=& c_{(x_i,\mathbf{b}),y_k}.  
\end{eqnarray}

Next, we demonstrate that by following the CD policy, the privacy criterion and the minimum data utility of the C-mDP can be still guaranteed. 

The data utility loss caused by the perturbed data $Q(X, V_X)$  depends only on the Markov blanket ${B}_{X}$ 
\end{assumption}
Give a brief discussion of the rationale of this assumption ...  
}

\begin{proposition}
\label{prop:CDcorrect}
(PL guarantee)
If a perturbation matrix $\mathbf{Q}$ follows the CD policy and satisfies the corresponding mDP constraints:  
\vspace{-0.07in}
\begin{equation}
\label{eq:C-mDPred}
q_{(x_i,\mathbf{b}),y_k} - e^{\epsilon d_{(x_i,\mathbf{b}), (x_j,\mathbf{b}')}} \cdot q_{(x_j,\mathbf{b}'),y_k} \leq 0, ~\forall (x_i, \mathbf{b}), (x_j, \mathbf{b}'), y_k 
\vspace{-0.00in}
\end{equation}
where $d_{(x_i,\mathbf{b}), (x_j,\mathbf{b}')} \leq d_{(x_i, \mathbf{v}), (x_j,\mathbf{v}')}$, 
it is \textbf{sufficient} for $\mathbf{Q}$ to achieve the bounded context-aware PL as defined in Eq. (\ref{eq:PLcontextbound}). 
\end{proposition}

\begin{assumption}
\label{assu:datautility}
We assume that the data utility loss of a secret record is determined by (i) its prior distribution, (ii) the perturbation distribution applied to it, and (iii) a set of context-independent factors (constants) that do not vary across the secret record contexts, such as road topology or the underlying map. Accordingly, each utility loss $c_{(x_i,\mathbf{v}),y_k}$ can be written as
\begin{equation}
c_{(x_i,\mathbf{v}),y_k} = h\big(p_{(x_i,\mathbf{v})}, q_{(x_i,\mathbf{v}),y_k}; \boldsymbol{\theta}\big),
\end{equation}
where $\boldsymbol{\theta}$ denotes such constant factors. By the CD policy, we have
\begin{eqnarray}
c_{(x_i,\mathbf{v}),y_k} 
&=& h\!\big(p_{(x_i,\mathbf{v})}, q_{(x_i,\mathbf{v}),y_k};\, \boldsymbol{\theta}\big) \\
&=& h\!\big(p_{(x_i,\mathbf{b})},\, q_{(x_i,\mathbf{b}),y_k};\, \boldsymbol{\theta}\big) \\
&=& c_{(x_i,\mathbf{b}),y_k}.
\end{eqnarray}
\end{assumption}
\vspace{-0.03in}
\begin{proposition}
\label{propo:CDutility}
Given that Assumption \ref{assu:datautility} holds, the loss function $\mathcal{L}(\mathbf{Q})$ in Eq. (\ref{eq:CA-UL}) can be rewritten as the following reduced form:
\begin{equation}
\mathcal{L}(\mathbf{Q}) = \sum_{(x_i, \mathbf{b}), y_k}  p_{(x_i, \mathbf{b})} \cdot c_{(x_i,\mathbf{b}),y_k} \cdot q_{(x_i,\mathbf{b}),y_k}. 
\end{equation}
\end{proposition}
{\revisiondone \noindent \textbf{Discussion}. 
In our performance evaluation (Section~\ref{sec:exp}), we instantiate the framework with vehicle assignment in spatial crowdsourcing, which aligns well with Assumption~\ref{assu:datautility}. In this setting, the platform’s decision quality is driven by travel-cost estimates computed from the reported (perturbed) worker locations and the fixed road network. Hence, for a given worker--context secret $(x_i,\mathbf{v})$, the expected utility loss under a reported location $y_k$ is determined by: (i) the prior $p_{(x_i,\mathbf{v})}$ over the worker’s true location under context $\mathbf{v}$, (ii) the applied perturbation probability $q_{(x_i,\mathbf{v}),y_k}$, and (iii) context-independent constants $\boldsymbol{\theta}$ capturing the map/road topology and the assignment objective. 

By contrast, other mDP applications, such as text perturbation, may fall outside the scope of Assumption~\ref{assu:datautility}. In such cases, utility loss can depend not only on the distribution of protected tokens but also on their semantic meaning in different context, which can substantially influence downstream tasks (e.g., sentiment analysis). To accommodate scenarios where Assumption~\ref{assu:datautility} does not hold, we outline a two-stage optimization strategy as a potential remedy; the details are provided in Appendix~\ref{sec:discussion}. 
} \looseness = -1

\vspace{0.02in}
Finally, to derive the optimal perturbation matrix $\mathbf{Q}$, we formulate the \textbf{Reduced C-mDP} problem as the following LP problem:  
\normalsize
\vspace{-0.10in}
\begin{eqnarray}
\label{eq:C-mDPObjred}
\min && 
\mathcal{L}(\mathbf{Q}) = \sum_{(x_i, \mathbf{b}), y_k} p_{(x_i, \mathbf{b})} \cdot c_{(x_i,\mathbf{b}),y_k} \cdot q_{(x_i,\mathbf{b}),y_k}\\
\mathrm{s.t.} && \mbox{mDP (Eq. (\ref{eq:C-mDPred})) is satisfied} \\
&& 
\sum_{y_k \in \mathcal{Y}}q_{(x_i,\mathbf{b}),y_k} = 1,~ \forall x_i, \mathbf{b}  \\ \label{eq:CLPconstraint1red}
&& 0 \leq q_{(x_i,\mathbf{b}),y_k} \leq 1, \forall x_i, y_k, \mathbf{b}, 
\end{eqnarray}
\normalsize
of which the decision variables are the entries in the matrix $\mathbf{Q} = \left\{q_{(x_i,\mathbf{b}),y_k}\right\}_{(x_i, \mathbf{b},y_k) \in \mathcal{X}\times \mathcal{B}\times \mathcal{Y} }$, including a total of $O(|\mathcal{X}||\mathcal{Y}||\mathcal{B}|)$ decision variables and $O(|\mathcal{X}|^2|\mathcal{Y}||\mathcal{B}|)$ linear constraints, which achieves a lower complexity compared to the original C-mDP formulated in Eq. (\ref{eq:C-mDPObj})--(\ref{eq:CLPconstraint1}). To simplify, hereafter, when we refer to ``C-mDP'', we are indicating its reduced form (Eq. (\ref{eq:C-mDPObjred})--Eq. (\ref{eq:CLPconstraint1red})).

\smallskip
{\revisiondone 
Here, we emphasize that the privacy parameter $\epsilon$ in our paper is an \emph{output-privacy} budget: it bounds the information leaked about an individual user only through the released perturbed output $Y=Q(X,B_X)$ at run time. Although $\mathbf{Q}$ is obtained by solving the above LP, this computation is performed \emph{offline} and the resulting $\mathbf{Q}$ is fixed for deployment; it is therefore not an interactive run-time process that repeatedly queries users' records and releases additional per-user information. In particular, the LP coefficients are constructed from (i) distances/utility components defined on the location domain (derived from the public map data in OpenStreetMap~\cite{openstreetmap}) and (ii) population-level prior statistics estimated offline from a representative trajectory dataset (e.g., using $330\mathrm{k}$ trajectories, with a separate held-out set for evaluation). After $\mathbf{Q}$ is fixed, each user/device perturbs only its own $(X,B_X)$ locally by sampling from the corresponding row of $\mathbf{Q}$, and the only information released about that user's record is $Y$, whose leakage is bounded by the C-mDP constraints with parameter $\epsilon$.}

\DEL{
\vspace{0.03in}
\noindent \textbf{Composition property}. In Theorem \ref{thm:seqcomposition}, we prove that both C-mDP and its reduced formulation satisfy sequential composition, which guarantees that a series of computations collectively maintain privacy if each individually ensuring privacy. 
This property facilitates the breakdown of computations into smaller, manageable building blocks. 
\vspace{-0.05in}
\begin{theorem}
\label{thm:seqcomposition}
(Sequential composition) Let $Q_1, ..., Q_M$ be $M$ perturbation functions satisfying C-mDP, with the privacy budgets equal to  $\epsilon_1, ..., \epsilon_M$, respectively. \looseness = -1
\begin{itemize}
\item [S1.] If the randomization of $Q_1, ..., Q_M$ is \textbf{conditional independent} $\forall x \in \mathcal{X}$ given the context data ${B}_{X} = \mathbf{b}$, then any method outputs $Q_1(x, \mathbf{b}), ..., Q_M(x, \mathbf{b})$ is $\sum_{m=1}^M\epsilon_m$-differentially private. 
\item [S2.] If the randomization of $Q_1, ..., Q_M$ is \textbf{independent} $\forall x \in \mathcal{X}$, then any method outputs $Q_1(x, \mathbf{b}), ..., Q_M(x, \mathbf{b})$ is $\sum_{m=1}^M\epsilon_m$-differentially private.
\end{itemize}
\end{theorem}
\vspace{-0.05in}
\begin{theorem}
\label{thm:parrallcomposition}
(Parallel composition) Let $Q_1, ..., Q_M$ be $M$ perturbation functions that all satisfy C-mDP, with the privacy budgets equal to  $\epsilon_1, ..., \epsilon_M$, respectively. If we partition the secret dataset $\mathcal{X}$ into $M$ disjoint subsets $\mathcal{X}_1, ..., \mathcal{X}_M$ and use $Q_m$ to perturb a record $x\in \mathcal{X}_m$ ($m = 1, ..., M$), then the perturbed result is $(\max_{m = 1, ..., M}\epsilon_m)$-differentially private. 
\end{theorem}}

\vspace{-0.05in}
\section{Markov Blanket Identification}
\label{sec:MBI}
\vspace{-0.00in}


Within the computation framework outlined in Section \ref{subsec:complexityreduce}, solving the reduced C-mDP requires the identification of the Markov blanket $B_X$ for the secret record $X$. In this section, we present our \emph{Markov Blanket Identification (MBI) framework}. 

To illustrate, we apply MBI in the context of vehicle privacy protection, where vehicles need to report perturbed locations to a central server \cite{Qiu-TMC2020}. 

\vspace{-0.05in}
\subsection{Additional Assumptions and Notations} 
Typically, the design of MBI relies on the two assumptions  \cite{Tsamardinos2003AlgorithmsFL}: 
\vspace{-0.03in}
\begin{itemize}
\item [(A1)] the data under consideration was generated by a \emph{Bayesian network} faithful to it, and 
\item [(A2)] there exists a reliable statistical method for testing the CI between the target random variables.
\end{itemize}
\vspace{-0.03in}
In the context of vehicle location privacy protection, the secret record $X_t$ - representing the location of a target vehicle at time $t$ - is correlated to the vehicle's preceding locations $V_{X_t} = \left\{X_{t-1}, ..., X_{t-T}\right\}$ at time slots $t-1, ..., t-T$, where $T$ is the size of $V_{X_t}$. This dependency can be naturally modeled by a higher-order Markov process \cite{Qiao-TITS2015}, commonly seen as a special case of Bayesian networks. Note that \emph{our methodology
remains applicable in general mDP applications given the two assumptions (A1) and (A2) hold}.

To test the CI between $X_t$ and its context variables $V_{X_t}$, we employ a statistical CI test using analytic kernel embeddings of location distributions \cite{Scetbon-ICML2022}. Specifically, the CI testing summarizes the evidence in the observational data against a null hypothesis $\mathcal{H}_0: X_t   \perp \!\!\!\perp X_{t-l} | B_{X_t}$ ($X_{t-l} \in V_{X_t}\backslash B_{X_t}$), and returns a $p$-value, representing the probability of making a \emph{Type I error} - rejecting $\mathcal{H}_0$ when it is true. If $p$-value $\leq 0.05$, it is typically considered to be statistically significant, in which case we \emph{reject} $\mathcal{H}_0$; otherwise, we \emph{fail to reject} $\mathcal{H}_0$ \cite{Bellot-NIPS2019}.

Notably, identifying the Markov blanket requires performing CI testing between protected records and their context variables, which has relatively high time complexity \cite{Scetbon-ICML2022}. As a result, it becomes challenging to identify the Markov blanket in real-time due to the time-sensitive nature of vehicle location reports. This challenge motivates us to predict $B_{X_t}$ using a pre-trained \emph{deep neural network (DNN)} rather than performing CI testing during each perturbation. 
The DNN establishes an empirical relationship between $B_{X_t}$ and the features that possibly influence vehicle mobility, including \emph{speeds}, \emph{regions}, and \emph{time}. This relationship, represented as a function $f(time, speed, region, \mathcal{H}_0)$, returns either ``\emph{Reject}'' or ``\emph{Fail to reject}'' for the hypothesis $\mathcal{H}_0$, which enables the vehicle to promptly identify the Markov blanket of its location based on its current states $\left\{time, speed, region\right\}$ without doing the CI test that might introduce significant delays.


\vspace{-0.05in}
\subsection{Methodology}  
\vspace{-0.00in}
Fig. \ref{fig:MBIframework} shows the framework of Markov Blanket Identification, including two stages (a) \emph{Markov Blanket Discovery} and (b) \emph{Markov Blanket Prediction}. 
\vspace{-0.05in}

\subsubsection{Stage (a): Markov Blanket Discovery} 
\vspace{-0.00in}
We first carry out an empirical study on the taxicab trajectory datasets from two different cities, "\emph{Rome, Italy}" \cite{roma-taxi-20140717} and "\emph{Porto, Portugal}" \cite{taxiPorto} to analyze the dependence of Markov blankets on factors including \emph{vehicle speed}, \emph{geographic regions} (e.g., downtown vs suburban areas), and \emph{time} (e.g., peak hours vs off-peak hours). In contrast to many existing studies (e.g., \cite{Qiu-SIGSPATIAL2022}) that assume the mobility of vehicles follows a first-order Markov process (so that a Markov blanket contains only one preceding location), \emph{our discovery reveals a notable variation in the Markov blankets given the different factors}. The detailed results of the empirical study can be found in Section \ref{subsec:empirical}.  

\begin{figure}[t]
\centering
\hspace{0.00in}
\begin{minipage}{0.50\textwidth}
  \subfigure{
\includegraphics[width=0.99\textwidth]{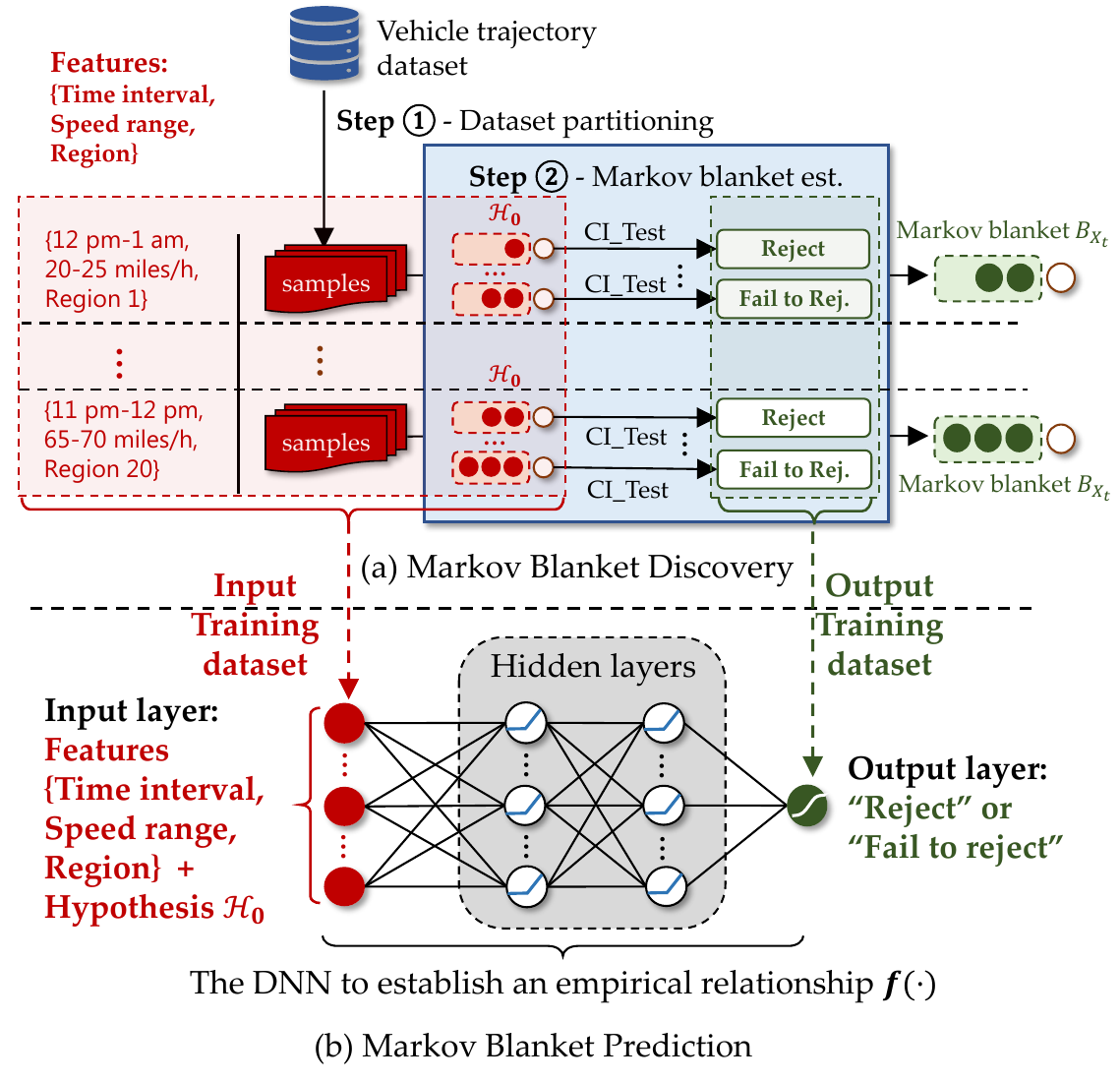}}
\vspace{-0.20in}
\end{minipage}
\caption{Markov blanket identification framework.}
\label{fig:MBIframework}
\vspace{-0.15in}
\end{figure}
\vspace{0.02in}
\textbf{Step \textcircled{1}: Dataset partitioning}. 
Considering the variance of Markov blankets under different features, as Fig. \ref{fig:MBIframework}(a) shows, we partition the initial dataset into groups according to the combination of the features $\{time, speed, region\}$. This ensures that the Markov blankets within each sample group exhibit a relatively consistent pattern (with each group comprising a minimum of 200 samples): 
\vspace{-0.00in}
\newline (1) \emph{Time}: According to the 24-hour time format, we 
categorize the trajectories into 24 groups [12 am--1 am), [1 am--2 am), ..., [11 pm--12 am) on the UTC +0 timestamps provided in the initial dataset. 
\newline  (2) \emph{Speed}: We calculate the average speed of each \emph{examined sub-trajectory} $\{X_t, \hat{B}_{X_t}\}$, comprising the current location $X_t$ and the tested Markov blanket $\hat{B}_{X_t}$. We divide all the sub-trajectories into groups of 5 mph intervals across the entire speed range from 0 to 120 mph.
\vspace{-0.00in}
\newline (3) \emph{Regions:} For the Rome dataset, we use a bounding box from 41.64°N to 42.12°N and 12.23°E to 12.83°E as the approximated boundary of Rome city, and divide the whole region into 4$\times$5 subregions with each cell size 12'$\times$12', then we categorize the sub-trajectories $\{X_t, \hat{B}_{X_t}\}$ into 20 groups based on their located regions. For the Porto dataset, the approximated boundary of the bounding box is 41.03°N to 41.27°N and 8.49°E to 8.73°E, and it is divided into 3$\times$3 subregions with each cell size 8'$\times$8', then we categorize the sub-trajectories into 9 groups based on their located regions. \looseness = -1
\vspace{-0.00in}

\vspace{0.02in}
\textbf{Step \textcircled{2}: Markov blanket identification}. Following the dataset partitioning, we proceed to estimate the Markov blanket $B_{X_t}$ of locations $X_t$ in each sample group. By assuming the locations in $V_{X_t}$ with closer time stamps to the current location as more correlated, we test the locations in $V_{X_t}$ in the order of decreasing time stamps. 



\vspace{-0.00in}
\begin{algorithm}[h]
\SetKwFunction{CITest}{CI\_Test}
\SetKwInOut{Input}{Input}
\SetKwInOut{Output}{Output}
\SetKwComment{tcp}{// }{}
\SetKwRepeat{Do}{do}{while}
\SetAlgoNoEnd
\small

\caption{\small Markov blanket identification in \textbf{Step \textcircled{2}}.}
\label{al:MBD}

\Input{Samples of $X_t$ and its context variables $V_{X_t}$}
\Output{Markov blanket $B_{X_t}$}

Let $B_{X_t}$ be an empty Markov blanket\;
$m \leftarrow 1$\tcp*[r]{Iteration index}

\Do{$p$-value $\leq 0.05$}{
    Add $X_{t-m}$ to $B_{X_t}$\;
    $p$-value $\leftarrow$ \CITest{$\mathcal{H}_0: X_t \perp\!\!\!\perp X_{t-m-1}\mid B_{X_t}$}\;
    Increase $m$ by 1\;
}

\Return{$B_{X_t}$\;}

\normalsize
\end{algorithm}
\vspace{-0.00in}

Algorithm \ref{al:MBD} shows the pseudocode: It initializes $B_{X_t}$ by $X_{t-1}$ (lines 1 and 4). After that, in each iteration $m$, it applies the function \texttt{CI\_test()}  \cite{Scetbon-ICML2022} to test the null hypothesis $\mathcal{H}_0: X_t  \perp \!\!\!\perp X_{t-m-1} | B_{X_t}$ (lines 2-7). If $\mathcal{H}_0$ is rejected, i.e. $p$-value $\leq 0.05$, then we add $X_{t-m}$ to $B_{X_t}$. The algorithm ends when the $p$-value returned by \texttt{CI\_test()} is higher than 0.05 (line 7), indicating that the hypothesis $\mathcal{H}_0$ fails to be rejected. In this case, we do not reject the Markov blanket $B_{X_t}$ and return it as the result (line 8). After Algorithm \ref{al:MBD}, we also label each Hypothesis $\mathcal{H}_0: X_t  \perp \!\!\!\perp X_{t-m-1} | B$ with $B_{X_t} \subset B \subset V_{X_t}$ as ``Fail to reject''.

\vspace{-0.04in}
\subsubsection{Stage (b) Markov Blanket Prediction}
\vspace{-0.01in}
Upon completing Stage (a), the status of each null hypothesis $\mathcal{H}_0$ is determined as either ``Reject'' or ``Fail to reject'' in each sample group sharing the same feature combination $\left\{time, speed, regions\right\}$. As depicted in Fig. \ref{fig:MBIframework}(b), we compile these features along with each $\mathcal{H}_0$ to create the input training dataset. The corresponding conclusions, ``Reject'' or ``Fail to reject'', serve as the output training dataset.


\vspace{-0.00in}
\begin{algorithm}[h]
\SetKwFunction{DNN}{DNN}
\SetKwInOut{Input}{Input}
\SetKwInOut{Output}{Output}
\SetKwComment{tcp}{// }{}
\SetKwRepeat{Do}{do}{while}
\SetAlgoNoEnd
\small

\caption{\small Markov blanket prediction using the pre-trained DNN.}
\label{al:DNN}

\Input{Current speed, time, and region}
\Output{Markov blanket $B_{X_t}$}

Let $B_{X_t}$ be an empty Markov blanket\;
$m \leftarrow 1$\tcp*[r]{Iteration index}

\Do{Indicator = ``Reject''}{
    Add $X_{t-m}$ to $B_{X_t}$\;
    Indicator $\leftarrow$ \DNN{time, speed, region, $m$}\;
    Increase $m$ by 1\;
}

\Return{$B_{X_t}$\;}

\normalsize
\end{algorithm}
\vspace{-0.00in}

Subsequently, a DNN is trained on this dataset to establish an empirical relationship $f$, enabling the identification of whether a hypothesis $\mathcal{H}_0$ is tested as ``Reject'' or ``Fail to reject'' given the features $\left\{time, speed, regions\right\}$. 

Algorithm \ref{al:DNN} shows the detailed pseudocode of the Markov blanket prediction introduced in \textbf{Stage (b)}. The inputs of the algorithm include the vehicle's current speed, time, and region. The output is the predicted Markov blanket. The algorithm commences by initializing the Markov blanket $B_{X_t}$ using $X_{t-1}$ (lines 1 and 4). Subsequently, in each iteration $m$, the algorithm leverages the pre-trained \texttt{DNN()} to predict whether it should ``\emph{reject}'' or ``\emph{fail to reject}'' the null hypothesis $\mathcal{H}_0: X_t  \perp \!\!\!\perp X_{t-m-1} | B_{X_t}$. The algorithm adds the preceding locations to $B_{X_t}$ sequentially (line 2-7),  concluding when it ``\emph{fails to reject}'' $\mathcal{H}_0$ (line 7), and returns $B_{X_t}$ as the identified Markov blanket (line 8).

\vspace{-0.00in}
\section{Performance Evaluation}
\label{sec:exp}
\vspace{-0.00in}
In this section, we evaluate our proposed context-aware mechanism, \emph{LP+C-mDP}, in a vehicle-based LBS setting~\cite{Qiu-EDBT2024, Pappachan-EDBT2023}, where a central server collects a \emph{perturbed} location report from a participating vehicle and recommends a destination (e.g., for {\revisiondone navigation} or spatial crowdsourcing). We use two real-world taxicab mobility datasets from \emph{Rome, Italy} and \emph{Porto, Portugal}: Rome contains 367{,}052 trajectories from 320 taxis over 30+ days\footnote{https://ieee-dataport.org/open-access/crawdad-romataxi}, and Porto contains 1{,}666{,}766 trajectories from 442 taxis over 540 days\footnote{https://www.kaggle.com/datasets/crailtap/taxi-trajectory}; each record includes a timestamp and GPS coordinates. For each city, we map trajectories onto a road network extracted from OpenStreetMap~\cite{openstreetmap} and model the network as a weighted directed graph~\cite{Qiu-TMC2020} (Rome: 43{,}160 nodes/89{,}739 edges; Porto: 5{,}033 nodes/10{,}537 edges).

{\revisiondone We begin with an empirical study in Section~\ref{subsec:empirical} that characterizes conditional dependencies in the mobility traces, and thus the effective Markov blanket size, across coarse factors such as region, speed, and time, motivating adaptive context selection.} Section~\ref{subsec:MBIperformance} then evaluates the accuracy and inference latency of our Markov-Blanket Identification (MBI) module and compares utility loss across privacy budgets against representative mDP baselines. Finally, Section~\ref{subsec:UL} reports computational efficiency, including both offline mechanism construction and online perturbation (sampling) costs.




\vspace{-0.00in}
\subsection{Correlation Between Time, Speed, Regions, and CI Testing}
\label{subsec:empirical}




{\revision 
\begin{table}[h]
\vspace{-0.05in}
\caption{Correlation between different features and the $p$-values of the CI testing \cite{Scetbon-ICML2022}.}
\vspace{-0.10in}
\label{Tb:Feature}
\centering
\small 
\begin{tabular}{ c|c|c|c|c}
\hline
\hline
\multicolumn{1}{ c  }{}& \multicolumn{4}{ c }{Features} \\
\cline{2-5}
\multicolumn{1}{ c|  }{Correlation}
 & \multicolumn{1}{ |c| }{Speed}&\multicolumn{2}{ |c| }{Region}
 & \multicolumn{1}{ |c }{Time}
 \\ 
\cline{3-4}
\multicolumn{1}{ c|  }{measures}
& \multicolumn{1}{ |c| }{}&\multicolumn{1}{ |c| }{Long.}&\multicolumn{1}{ c| }{Lat.}
& \multicolumn{1}{ |c }{}
\\ 
\hline
\hline
\multicolumn{1}{ c|  }{ } & \multicolumn{4}{ c  }{ Rome, Italy} \\
\cline{2-5}
\multicolumn{1}{ c|  }{ Pearson} & -0.4886 & -0.0081 & -0.0028 & 0.0410 \\
\multicolumn{1}{ c|  }{ Spearman's rank} & -0.7122 & -0.0734 & 0.0124 & 0.0620 \\
\multicolumn{1}{ c|  }{ Kendall's tau} & -0.5430 & -0.0539 & 0.0102 & 0.0420 \\ 
\multicolumn{1}{ c|  }{Type I error of CCIT} & 1.14e-30 & 0.0974 & 0.3977 & 0.2871
\\ 
\hline
\hline
\multicolumn{1}{ c|  }{ } & \multicolumn{4}{ c  }{Porto, Portugal} \\
\cline{2-5}
\multicolumn{1}{ c|  }{ Pearson} & -0.1848 & -0.0913 & 0.0905 & -0.0189 \\
\multicolumn{1}{ c|  }{ Spearman's rank} & -0.1283 & -0.2405 & 0.1144 & -0.0296 \\
\multicolumn{1}{ c|  }{ Kendall's tau} & -0.1075 & -0.1804 & 0.0912 & -0.0202 \\ 
\multicolumn{1}{ c|  }{Type I error of CCIT} & 9.55e-10 & 3.72e-14 & 0.2447 & 0.0152
\\ 
\hline
\end{tabular}
\vspace{-0.00in}
\end{table}}

Table \ref{Tb:Feature} lists multiple statistical measures to demonstrate a correlation between $p$-values and the other features like \emph{speed}, \emph{regions}, and \emph{time}. The correlation measures include Pearson's correlation, Spearman's rank, and Kendall's tau, providing insights into the linear and rank-based relationships between the features. The Type I error rate of the CCIT\footnote{Classifier CI Test (CCIT)  \cite{Sen-NIPS2017} is a non-parametric method that can test independence between continuous random variables. In Table \ref{Tb:Feature}, if the probability of its Type I error is lower than 0.05, the null hypothesis is rejected, meaning that the $p$-value of the CI test in \cite{Scetbon-ICML2022} and the feature are not independent.} testing is also included. 

The table reveals distinct correlation features between the selected features and $p$-values in the two cities. In Rome, \emph{speed} exhibits a strong negative correlation with $p$-values across all three correlation metrics, indicating that as \emph{speed} increases, $p$-values tend to decrease. In Porto, \emph{speed} also shows a negative correlation with $p$-values, and the CCIT results suggest that this relationship remains statistically significant, albeit much weaker than in Rome. The correlations between $p$-values and \emph{regions} (Longitude and Latitude) in Rome are weak, and the CCIT values indicate that they are not statistically significant. However, Longitude shows a slightly stronger negative correlation with $p$-values in Porto. The \emph{time} feature reveals a weak correlation in both cities.


Next, we provide detailed empirical analysis of the correlation between \emph{regions, speed, and time}, and the $p$-value of the CI test.

\begin{figure*}[t]
\centering
\begin{minipage}{1.0\textwidth}
\centering
  \subfigure[\footnotesize $|B_{X_t}| = 1$ and $m = 1$]{
\includegraphics[width=0.18\textwidth, height = 0.12\textheight]{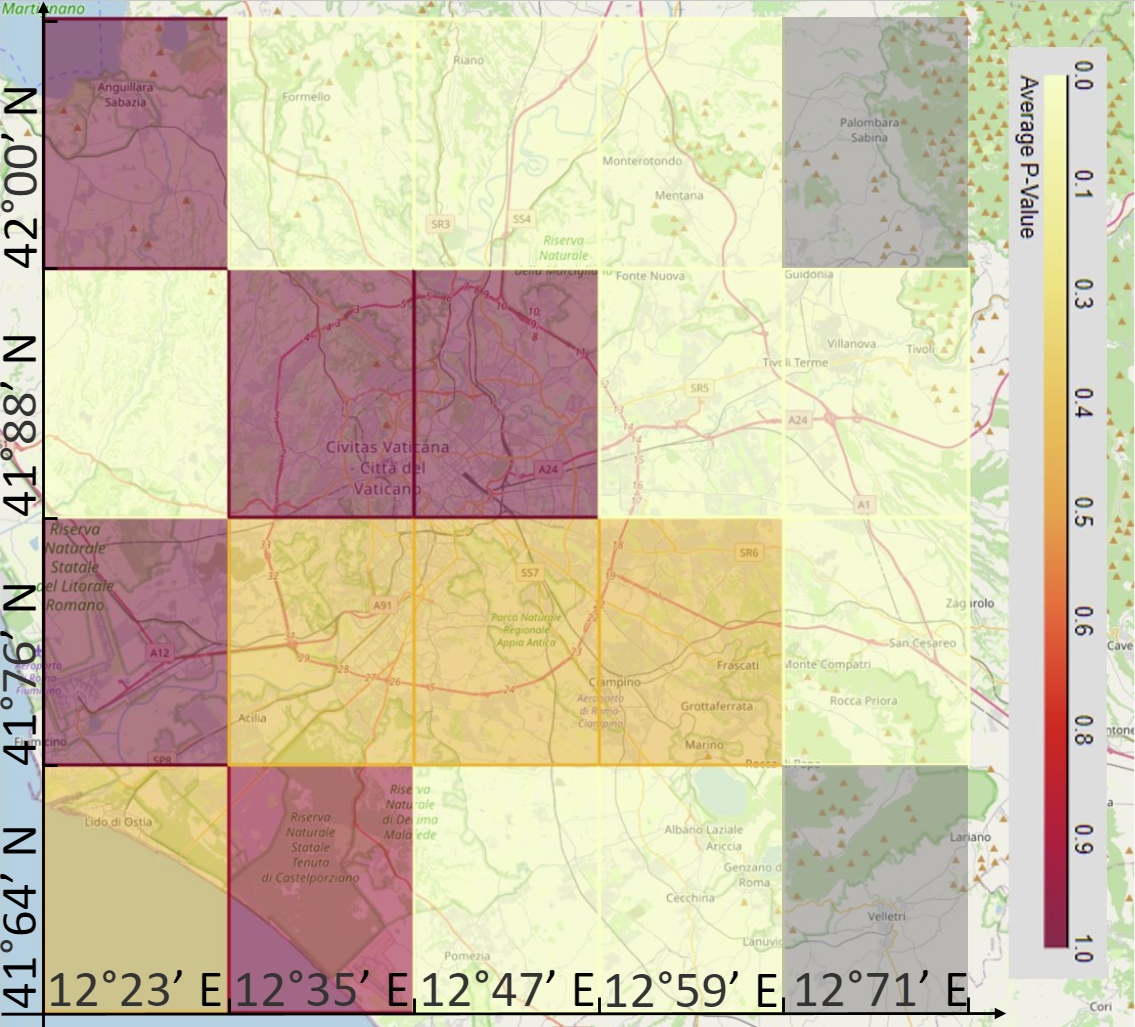}}
\label{}
\centering
  \subfigure[\footnotesize $|B_{X_t}| = 2$ and $m = 2$]{
\includegraphics[width=0.18\textwidth, height = 0.12\textheight]{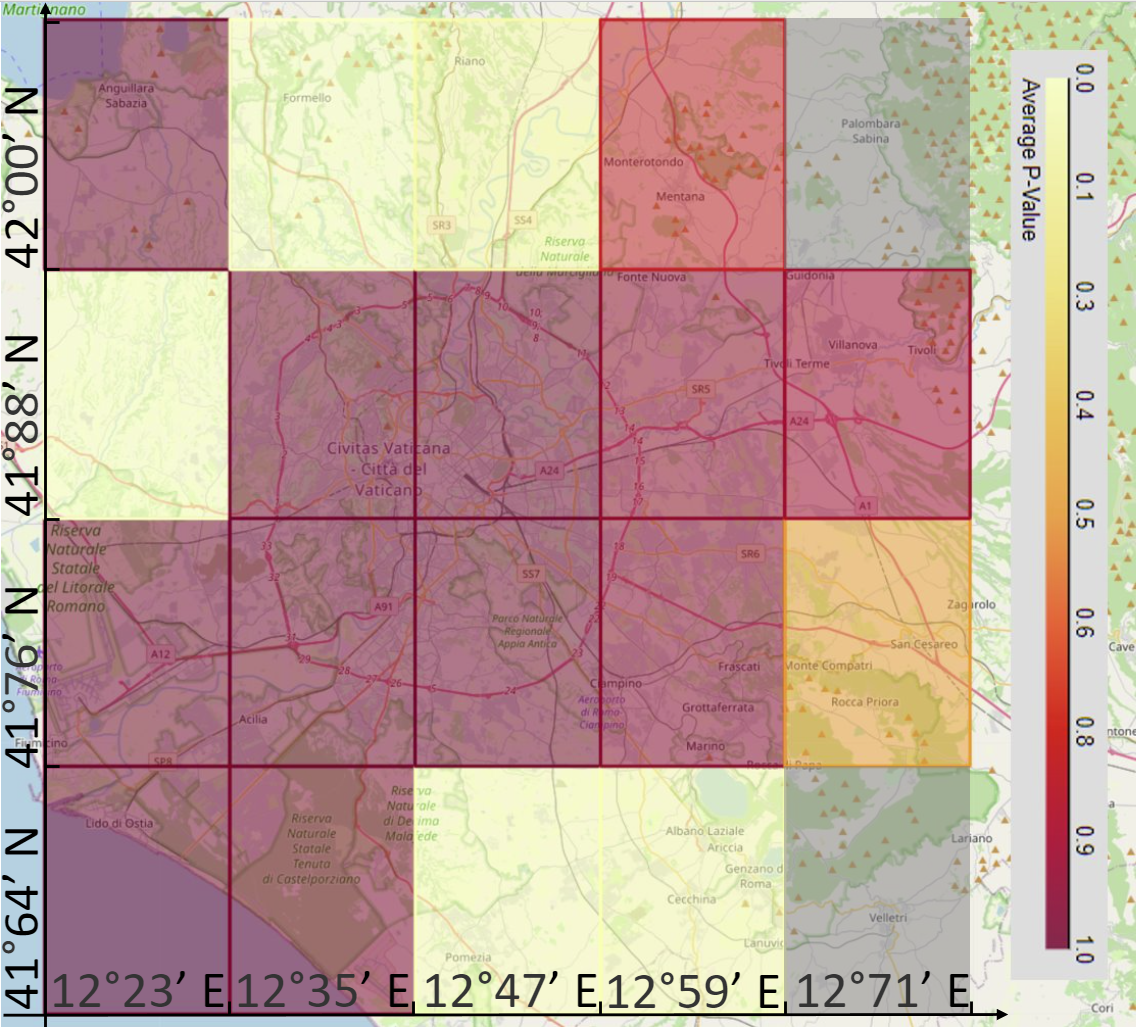}}
\label{}
\centering
  \subfigure[\footnotesize $|B_{X_t}| = 3$ and $m = 3$]{
\includegraphics[width=0.18\textwidth, height = 0.12\textheight]{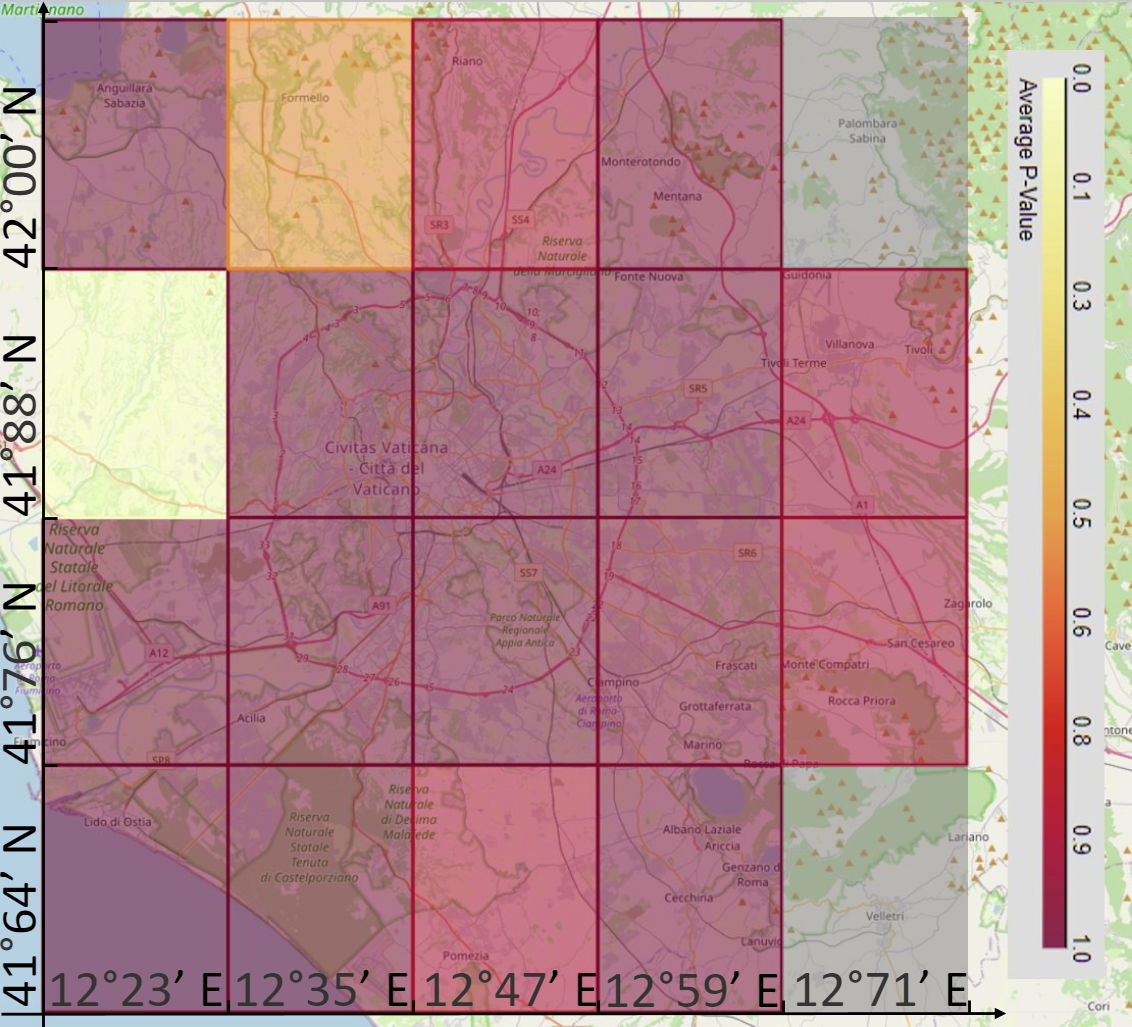}}
\label{}
\centering
  \subfigure[\footnotesize $|B_{X_t}| = 4$ and $m = 4$]{
\includegraphics[width=0.18\textwidth, height = 0.12\textheight]{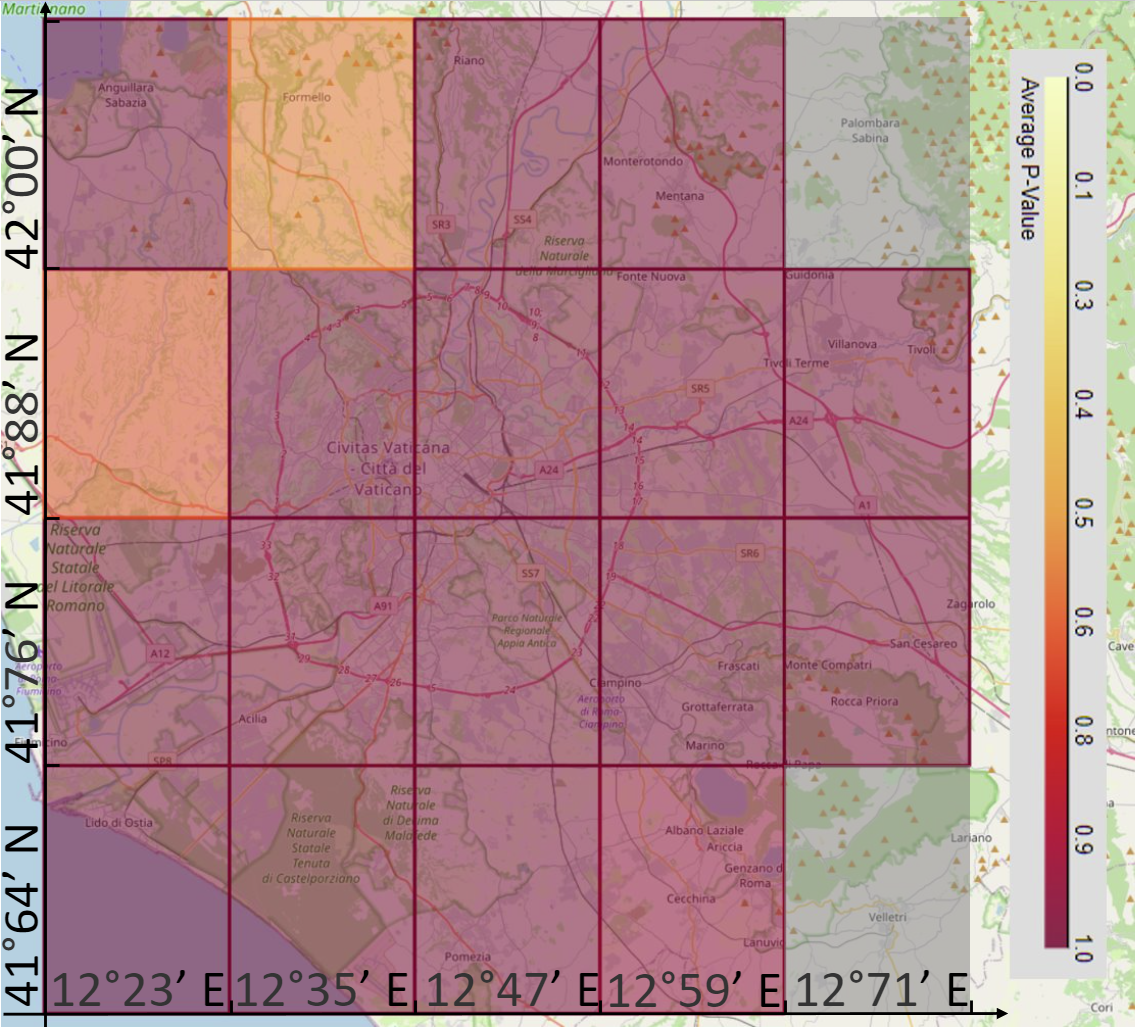}}
\label{}
\centering
  \subfigure[\footnotesize $|B_{X_t}| = 5$ and $m = 5$]{
\includegraphics[width=0.18\textwidth, height = 0.12\textheight]{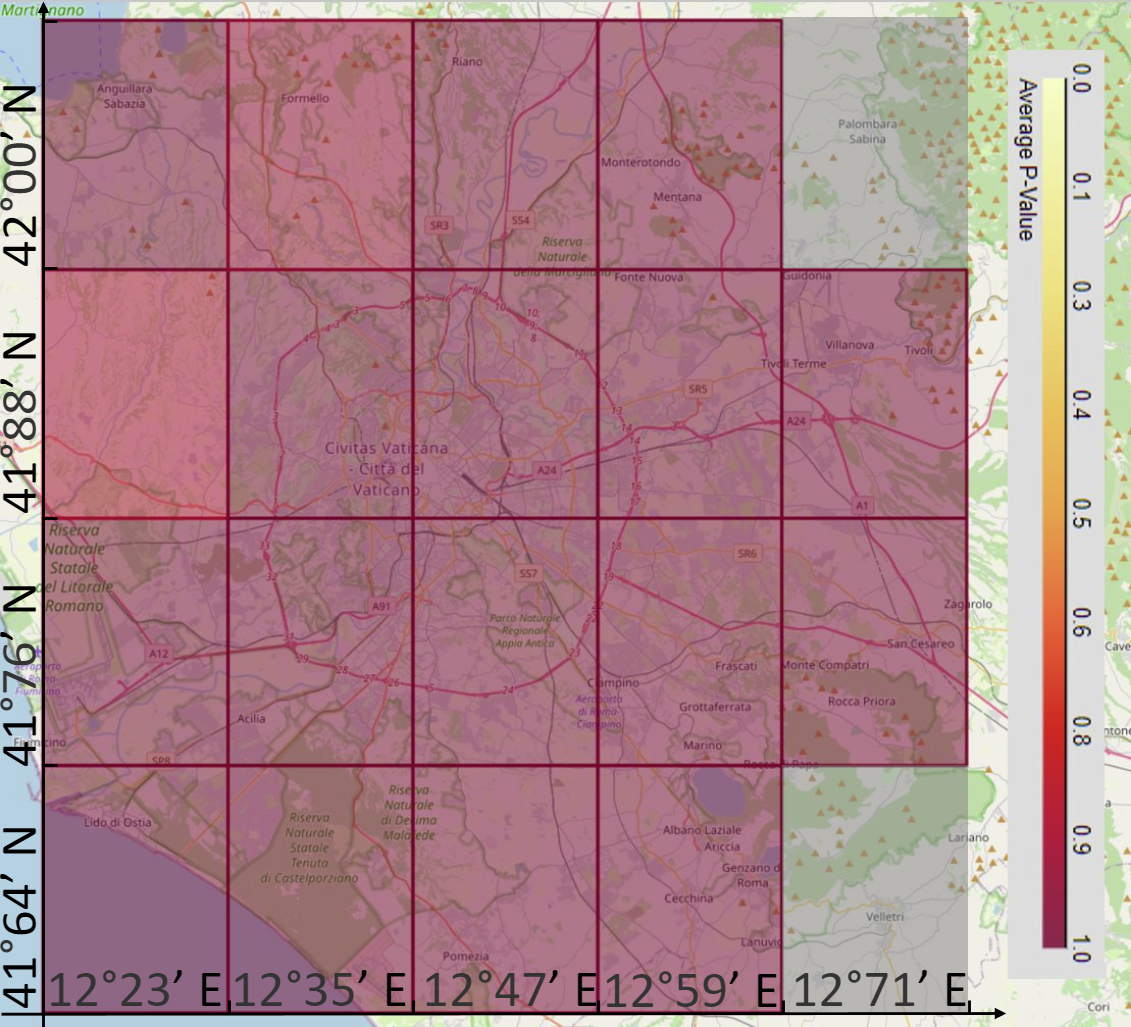}}
\label{}
\end{minipage}
\vspace{-0.08in}
\caption{Heatmap of $p$-values returned by the null hypothesis $\mathcal{H}_0$ across different regions in Rome, Italy.}
\label{fig:regionvspvalueRome}
\vspace{-0.05in}
\end{figure*}

\begin{figure*}[t]
\centering
\begin{minipage}{1.0\textwidth}
\centering
  \subfigure[\footnotesize $|B_{X_t}| = 1$ and $m = 1$]{
\includegraphics[width=0.18\textwidth, height = 0.12\textheight]{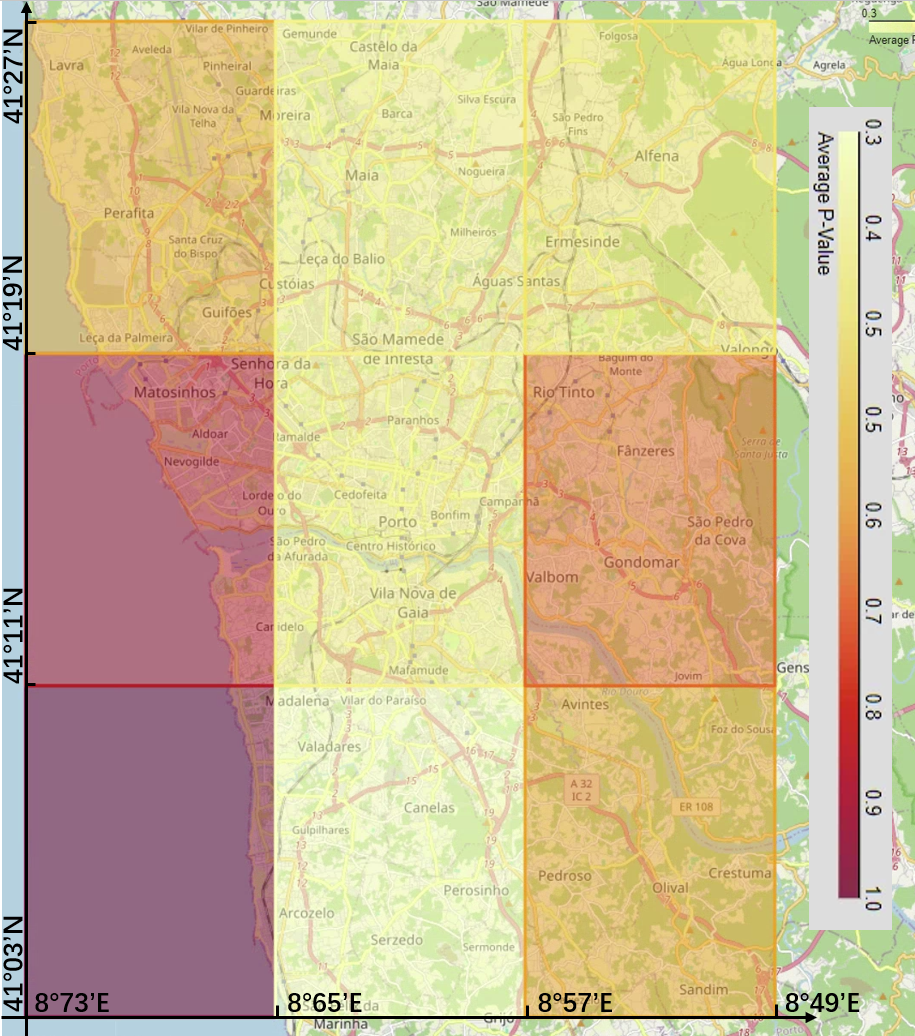}}
\label{}
\centering
  \subfigure[\footnotesize $|B_{X_t}| = 2$ and $m = 2$]{
\includegraphics[width=0.18\textwidth, height = 0.12\textheight]{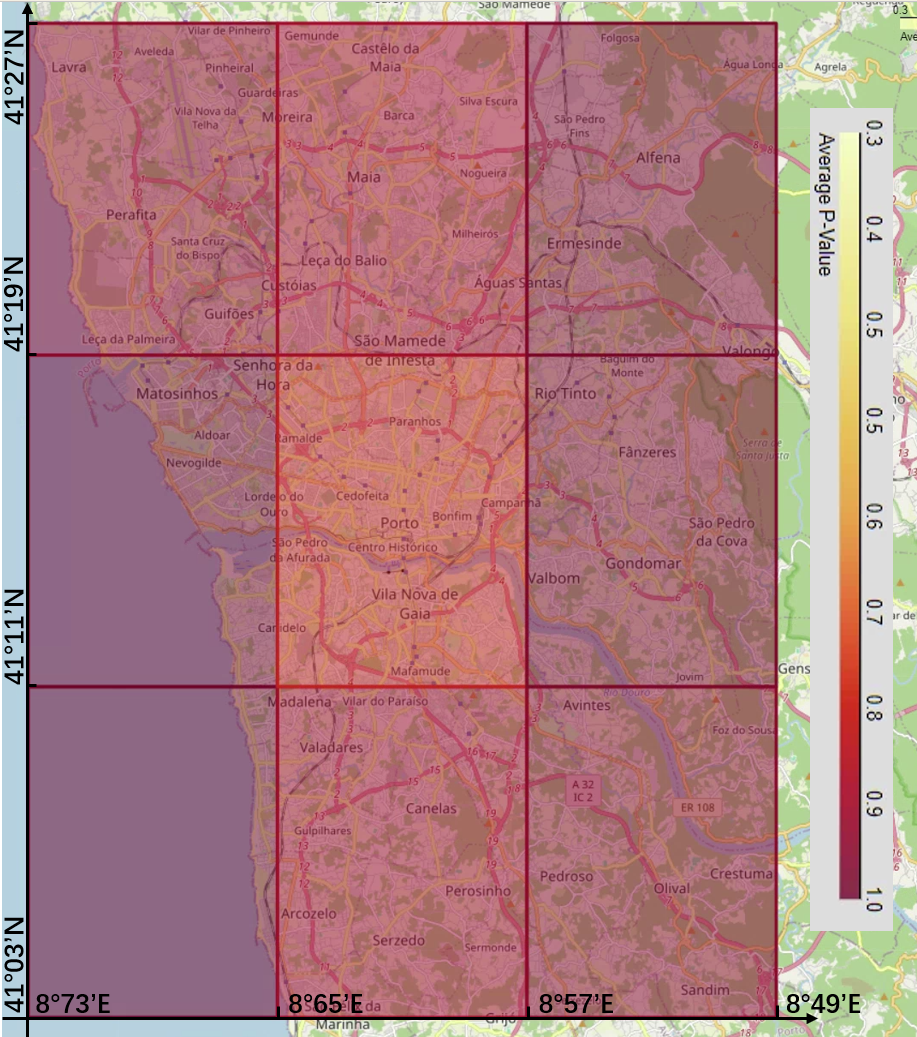}}
\label{}
\centering
  \subfigure[\footnotesize $|B_{X_t}| = 3$ and $m = 3$]{
\includegraphics[width=0.18\textwidth, height = 0.12\textheight]{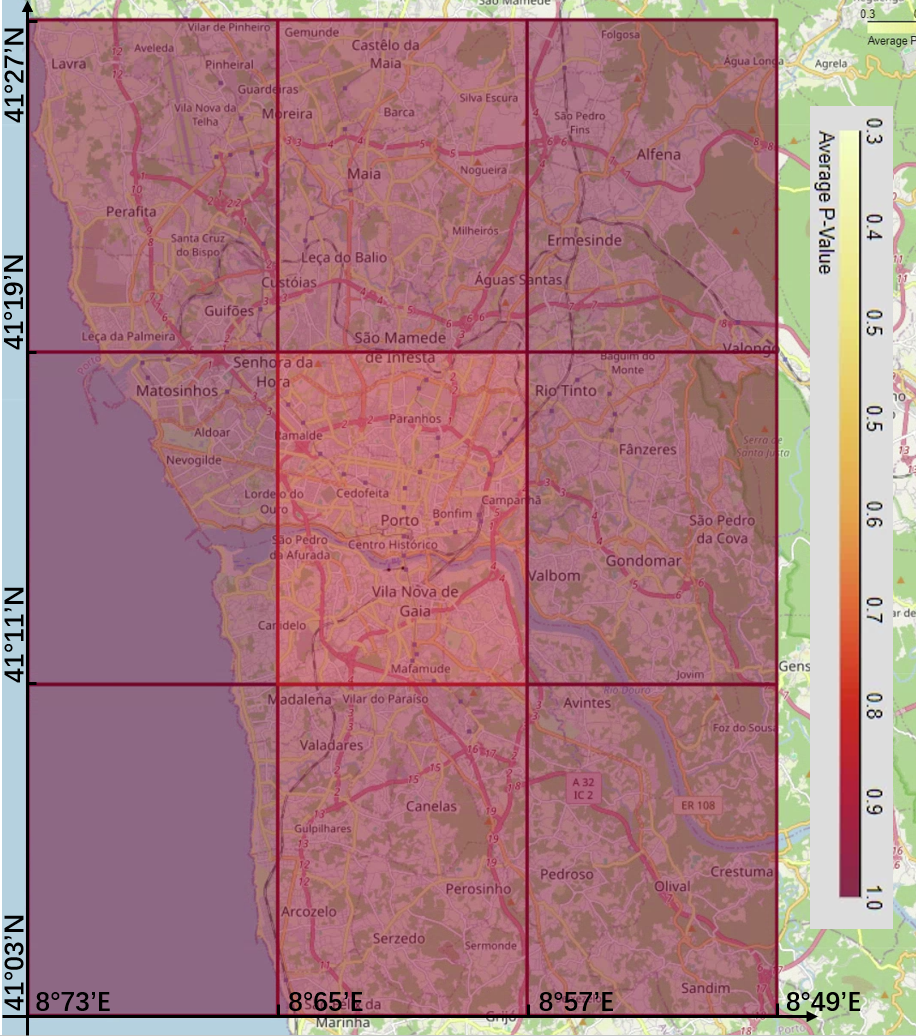}}
\label{}
\centering
  \subfigure[\footnotesize $|B_{X_t}| = 4$ and $m = 4$]{
\includegraphics[width=0.18\textwidth, height = 0.12\textheight]{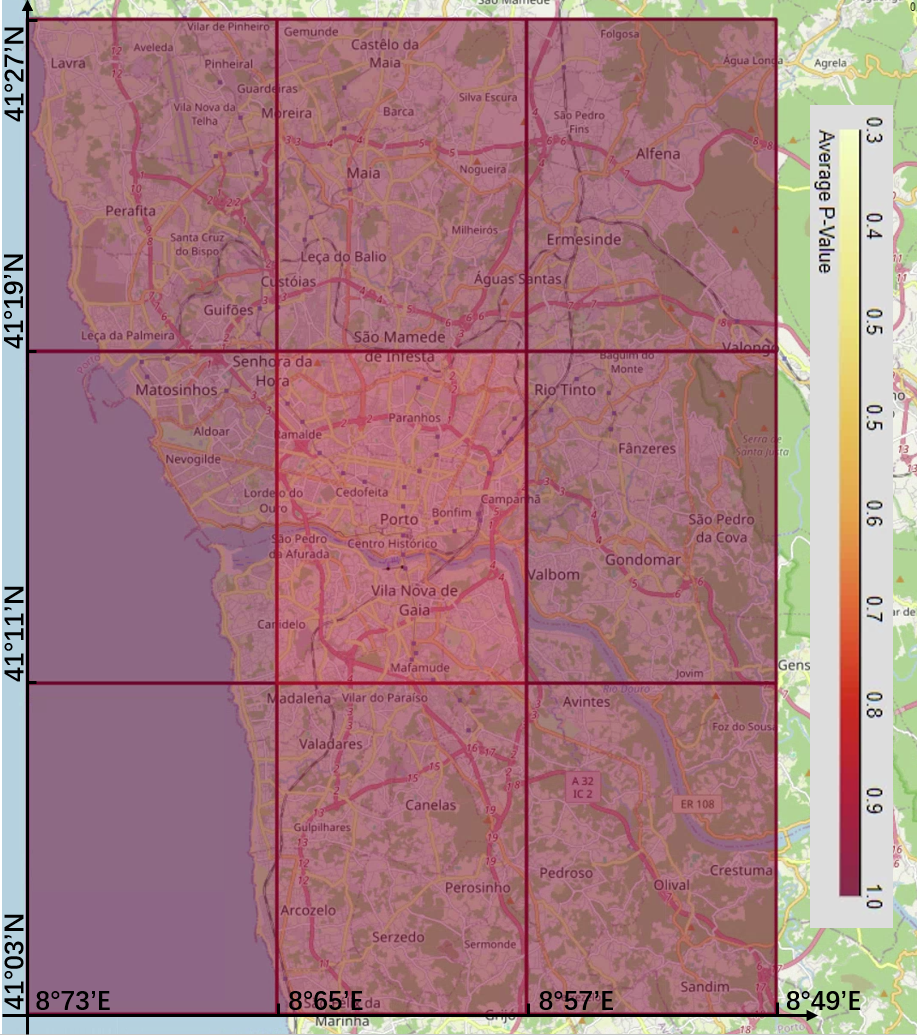}}
\label{}
\centering
  \subfigure[\footnotesize $|B_{X_t}| = 5$ and $m = 5$]{
\includegraphics[width=0.18\textwidth, height = 0.12\textheight]{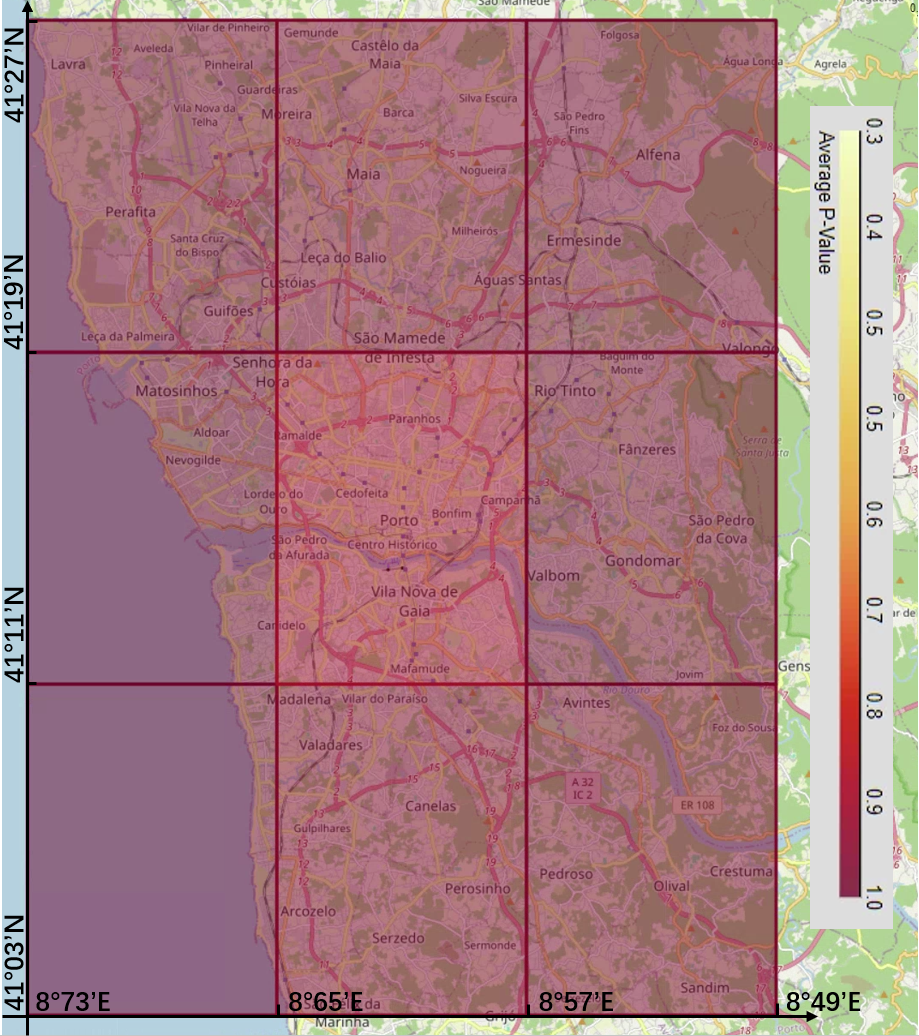}}
\label{}
\end{minipage}
\vspace{-0.08in}
\caption{Heatmap of $p$-values returned by the null hypothesis $\mathcal{H}_0$ across different regions in Porto, Portugal.}
\label{fig:regionvspvaluePorto}
\vspace{-0.05in}
\end{figure*}

\vspace{0.03in}
\noindent \textbf{(1) $p$-value vs. regions}. As Fig. \ref{fig:regionvspvalueRome} and Fig. \ref{fig:regionvspvaluePorto} show, we define the approximated boundary of Rome and Porto using two bounding boxes (as described earlier). Due to insufficient data in two regions of the Rome dataset, we excluded these two regions from our analysis (of which the bounding boxes are (i) from 42.00°N to 42.12°N and 12.71°E to 12.83°E, and (ii) from 41.64°N to 41.76°N and 12.71°E to 12.83°E, respectively). We test the null hypothesis $\mathcal{H}_0: X_t  \perp \!\!\!\perp X_{t-m-1} | B_{X_t}$ for the trajectories using the CI test. In Fig. \ref{fig:regionvspvalueRome}(a)-(e), we display the heatmap of the $p$-values returned by the CI test across the 18 regions in Rome when $|B_{X_t}| = 1, 2, 3, 4, 5$, respectively. The figures show that when $m = 1$ or $2$, the average $p$-value is higher in downtown compared to the suburban area, suggesting that the hypothesis $\mathcal{H}_0$ is more likely to be rejected in the suburban area than in downtown when $m = 1$ or $2$. The heatmap for Porto in Fig. \ref{fig:regionvspvaluePorto}(a)-(e) reveals different conclusions. When $m = 1$, the average $p$-value is lower in downtown compared to the suburban area, and when $m = 2$, all $p$-value turn relatively high. The varying $p$-values across regions reflect differences in vehicle mobility patterns. For example, vehicles may move more slowly in downtown areas (potentially leading to higher $p$-values) and faster in suburban areas (potentially resulting in lower $p$-values).



\begin{figure}[h]
\centering
\begin{minipage}{0.50\textwidth}
\centering
  \subfigure[Rome, Italy]{
\includegraphics[width=0.48\textwidth, height = 0.12\textheight]{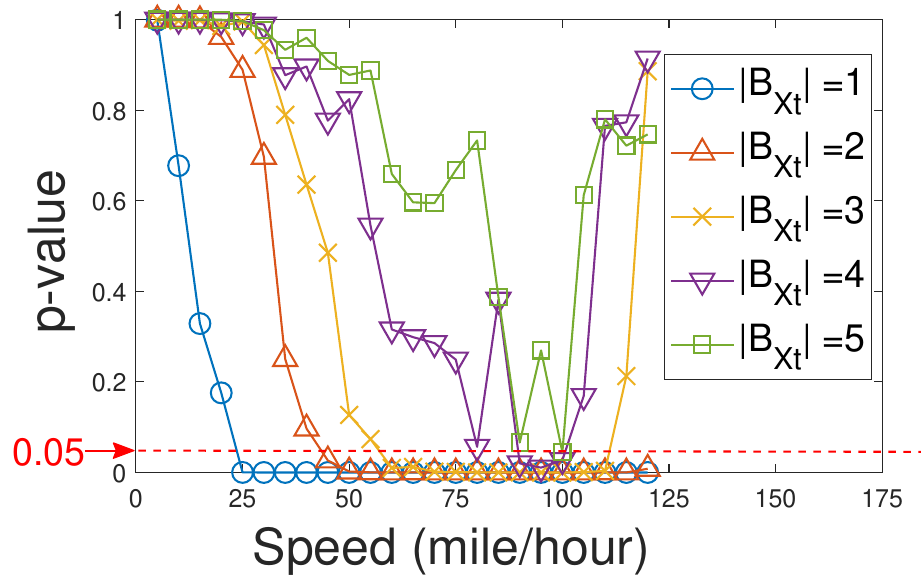}}
  \subfigure[Porto, Portugal]{
\includegraphics[width=0.48\textwidth, height = 0.12\textheight]{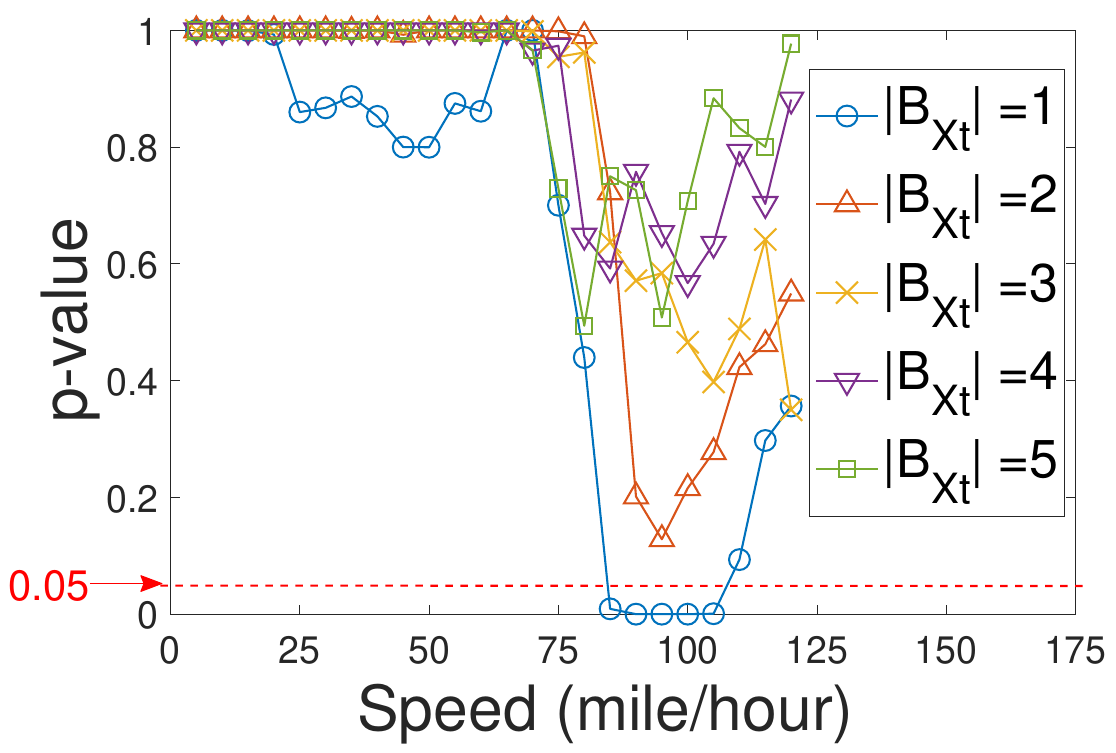}}
\label{}
\end{minipage}
\vspace{-0.15in}
\caption{Relationship between speeds and $p$-values. 
\newline $p$-values are returned by the null hypothesis $\mathcal{H}_0$ given different speed ranges.}
\label{fig:pvaluevsspeed}
\vspace{-0.05in}
\end{figure}

\begin{figure*}[h]
\centering
\begin{minipage}{1.0\textwidth}
\centering
\vspace{-0.0in}
  \subfigure[\small $|B_{X_t}| = 1$ and $m = 1$]{
\includegraphics[width=0.19\textwidth, height = 0.08\textheight]{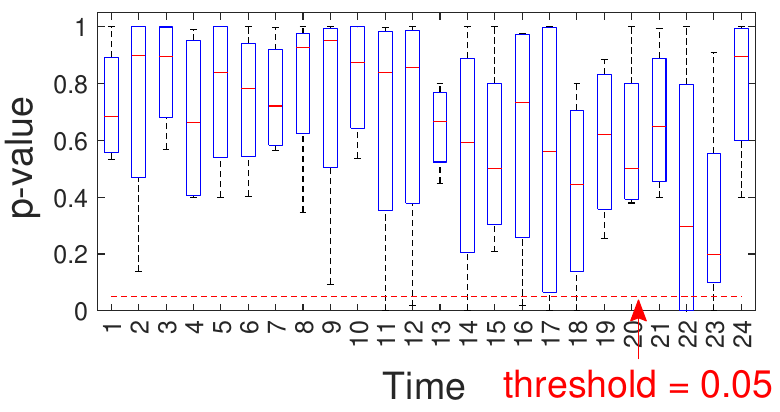}}
\vspace{-0.0in}
  \subfigure[\small $|B_{X_t}| = 2$ and $m = 2$]{
\includegraphics[width=0.19\textwidth, height = 0.08\textheight]{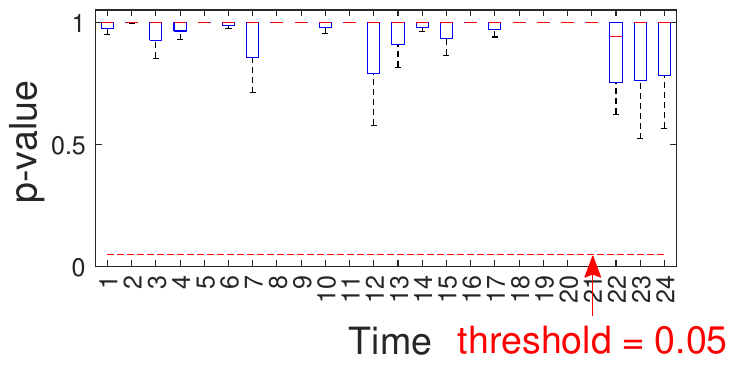}}
\vspace{-0.0in}
  \subfigure[\small $|B_{X_t}| = 3$ and $m = 3$]{
\includegraphics[width=0.19\textwidth, height = 0.08\textheight]{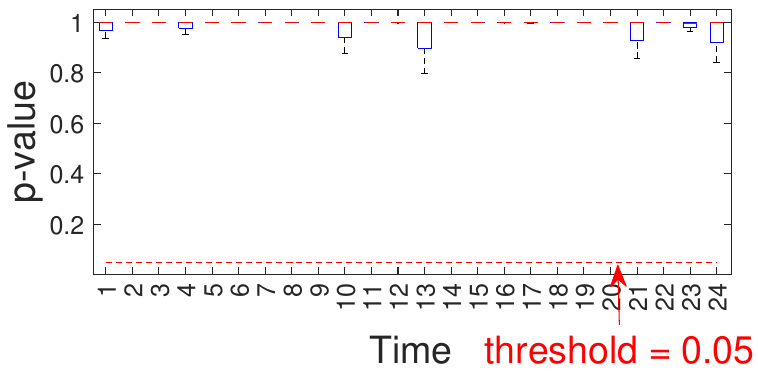}}
\vspace{-0.0in}
  \subfigure[\small $|B_{X_t}| = 4$ and $m = 4$]{
\includegraphics[width=0.19\textwidth, height = 0.08\textheight]{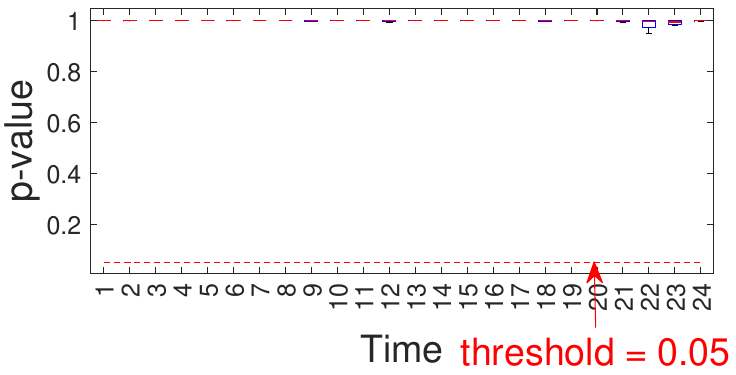}}
\vspace{-0.0in}
  \subfigure[\small $|B_{X_t}| = 5$ and $m = 5$]{
\includegraphics[width=0.19\textwidth, height = 0.08\textheight]{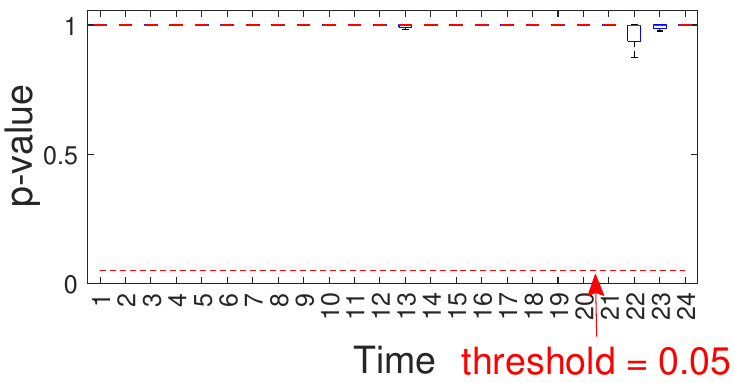}}
\end{minipage}
\vspace{-0.1in}
\caption{Relationship between time and $p$-values in Rome, Italy (returned by testing $\mathcal{H}_0$ in different intervals). }
\label{fig:pvaluevstimeRome}
\vspace{-0.05in}
\end{figure*}

\begin{figure*}[h]
\centering
\begin{minipage}{1.0\textwidth}
\centering
\vspace{-0.0in}
  \subfigure[\small $|B_{X_t}| = 1$ and $m = 1$]{
\includegraphics[width=0.19\textwidth, height = 0.08\textheight]{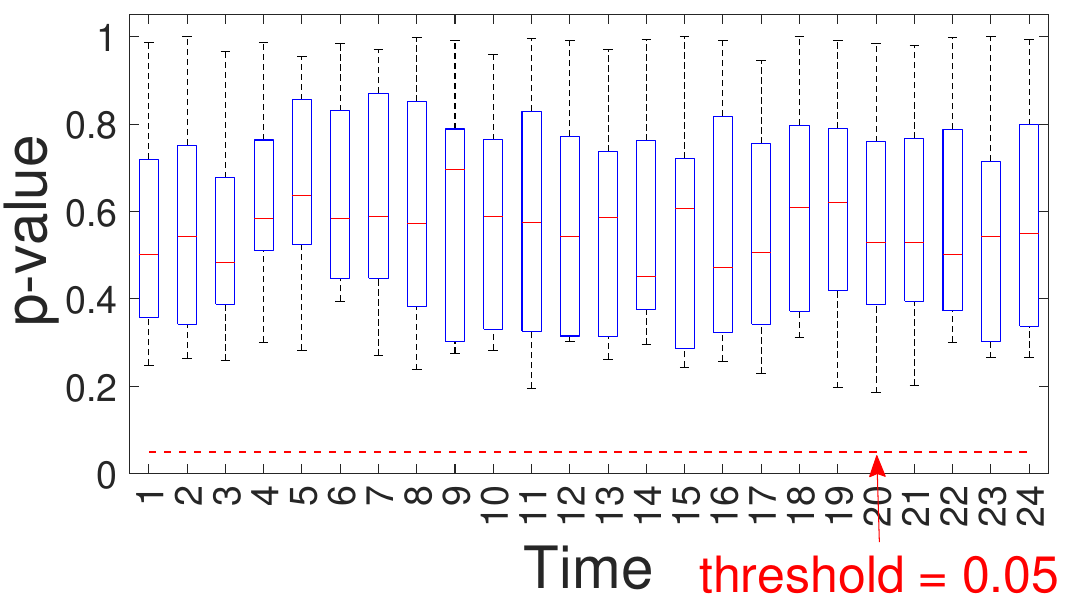}}
\vspace{-0.00in}
  \subfigure[\small $|B_{X_t}| = 2$ and $m = 2$]{
\includegraphics[width=0.19\textwidth, height = 0.08\textheight]{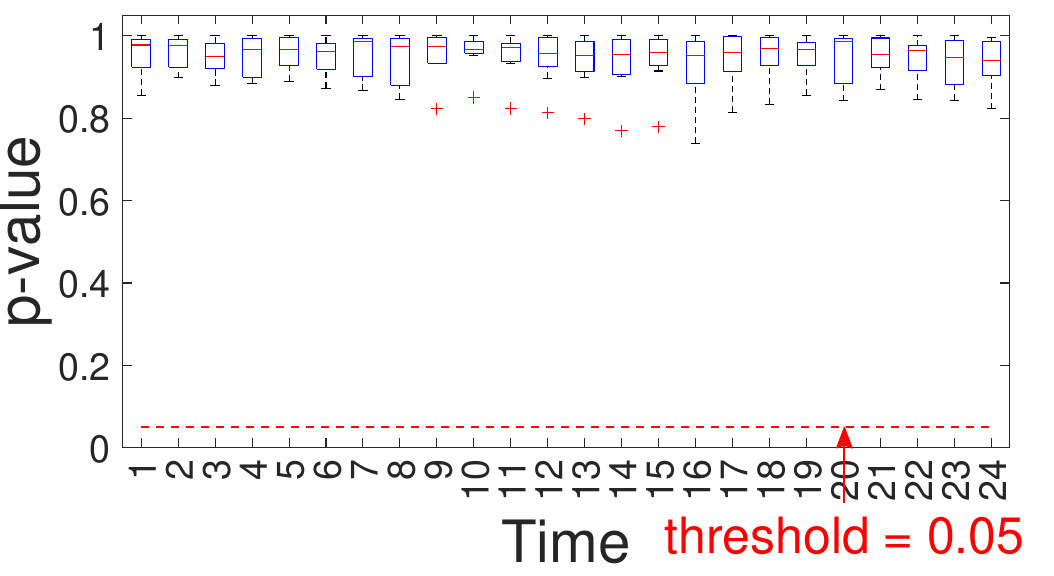}}
\vspace{-0.00in}
  \subfigure[\small $|B_{X_t}| = 3$ and $m = 3$]{
\includegraphics[width=0.19\textwidth, height = 0.08\textheight]{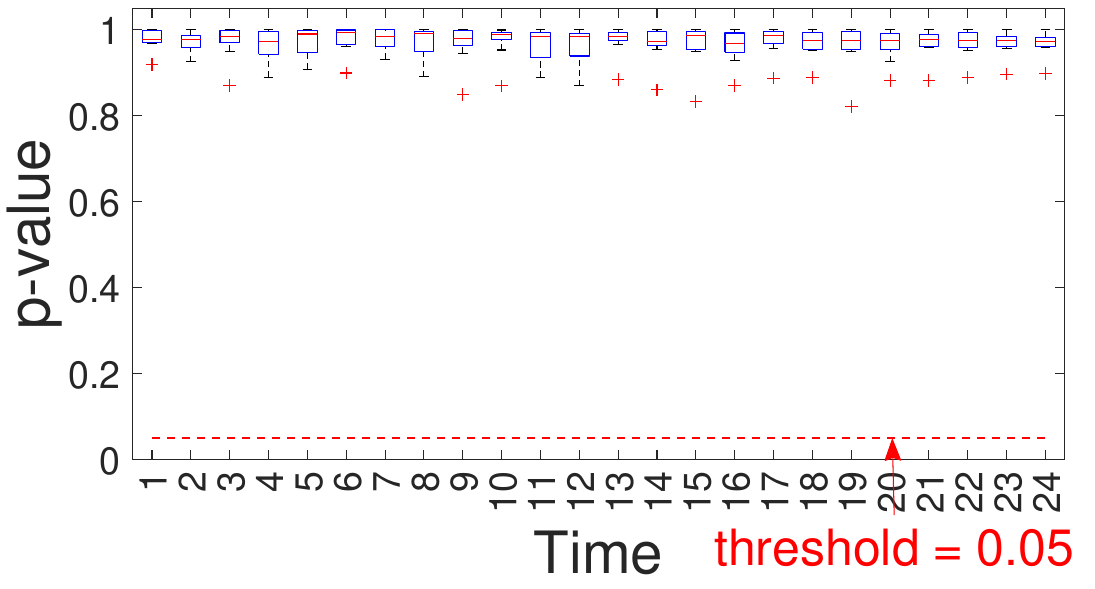}}
\vspace{-0.00in}
  \subfigure[\small $|B_{X_t}| = 4$ and $m = 4$]{
\includegraphics[width=0.19\textwidth, height = 0.08\textheight]{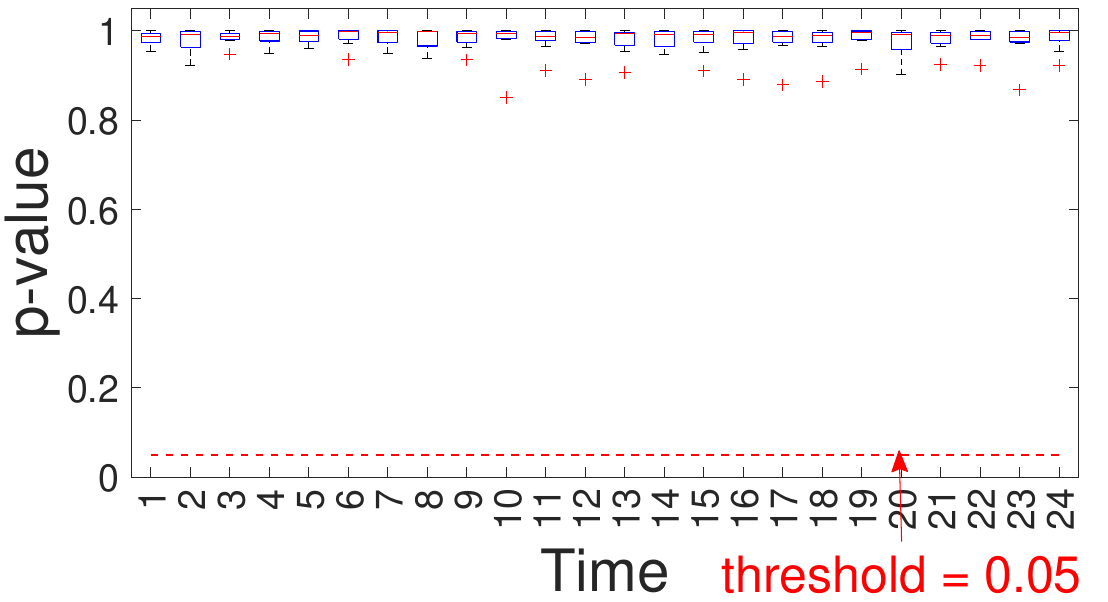}}
\vspace{-0.00in}
  \subfigure[\small $|B_{X_t}| = 5$ and $m = 5$]{
\includegraphics[width=0.19\textwidth, height = 0.08\textheight]{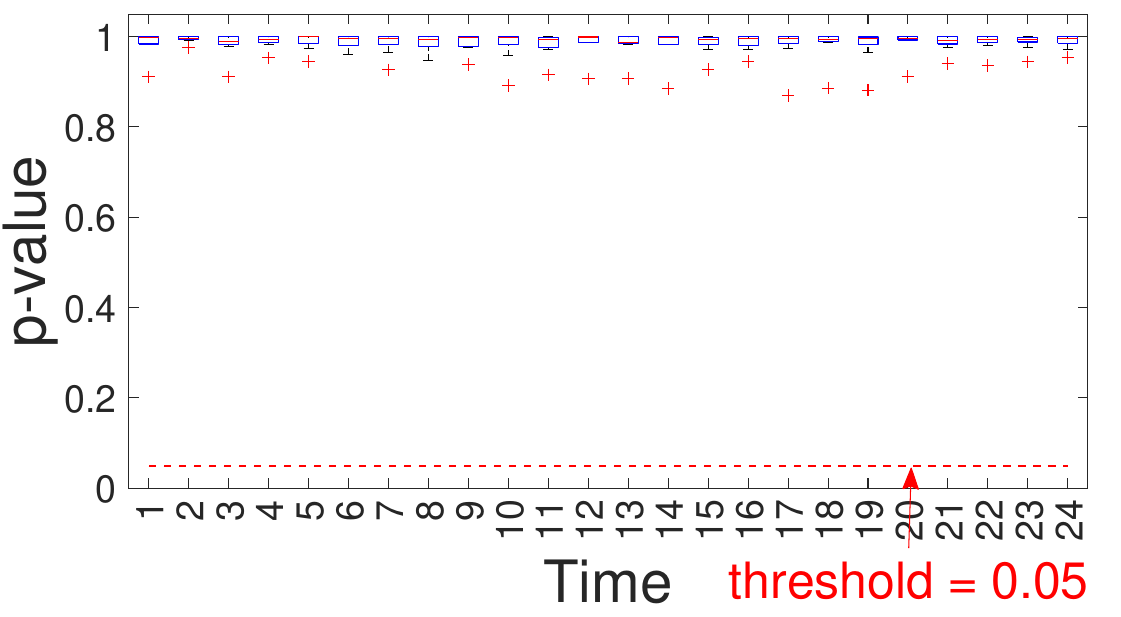}}
\end{minipage}
\vspace{-0.1in}
\caption{Relationship between time and $p$-values in Porto, Portugal (returned by testing $\mathcal{H}_0$ in different intervals). }
\label{fig:pvaluevstimePorto}
\vspace{-0.05in}
\end{figure*}

\vspace{0.03in}
\noindent \textbf{(2) $p$-value vs. speeds}. 
We calculate the average speed of each \emph{examined sub-trajectory} $\{X_t, \hat{B}_{X_t}\}$, based on which we divide all the sub-trajectories into groups of 5 mph intervals across the entire speed range from 0 to 120 mph. We test the null hypothesis $\mathcal{H}_0$ for the trajectories in different speed intervals and display the corresponding $p$-values in the two cities in Fig. \ref{fig:pvaluevsspeed}(a)(b), respectively, given $|B_{X_t}| = 1, 2, 3, 4, 5$. Results from both cities reveal a similar trend: within a certain range, speed is negatively correlated with the $p$-value, but beyond this range, the correlation becomes positive, forming a U-shaped pattern. In other words, at very low and very high speeds, we fail to reject the null hypothesis, indicating a high correlation between trajectory location points. Conversely, at intermediate speed levels, we reject the null hypothesis and accept the alternative hypothesis, suggesting lower correlation between trajectory location points. Specifically, in the Rome dataset, $p$-values reach their minimum within the speed range of 75 to 100, before rising again, while in the Porto dataset, this minimum occurs at a higher speed range, highlighting city-specific differences.


\vspace{0.03in}
\noindent \textbf{(3) $p$-value vs. time intervals}. 
According to the 24-hour time format, we group the trajectories as described earlier. We test the null hypothesis $\mathcal{H}_0$ for the trajectories in different time intervals and display the corresponding $p$-values when $|B_{X_t}| = 1, 2, 3, 4, 5$ in Fig. \ref{fig:pvaluevstimeRome}(a)-(e) and Fig. \ref{fig:pvaluevstimePorto}(a)-(e), respectively. We observe that when $|B_{X_t}| = 3, 4, 5$, all the $p$-values are close to 1 in both cities. Conversely, when $|B_{X_t}| = 1, 2$, the $p$-values for Rome are relatively lower during the intervals 11-14 (11 am -- 2 pm ), 17-18 (5 pm -- 6 pm) and 22-23 (10 pm -- 11 pm). This observation aligns with the findings in Fig. \ref{fig:pvaluevsspeed}, where during these time intervals, vehicle traffic is relatively lower, indicating higher average vehicle speeds. Consistent with Fig. \ref{fig:pvaluevsspeed}, $p$-values decrease as speed increases up to an intermediate range (approximately 75–100 mph) and then rise again. 


\noindent\textbf{Conclusion.} From all the figures above, we find \textbf{\emph{the results of CI testing (the $p$-values) exhibit variations across different regions, time intervals, and speed intervals}}. This observation challenges the assumption made by many related works, which posit that targets' mobility follows a first-order Markov process \cite{Qiu-SIGSPATIAL2022, Cao-ICDE2017}. 

Moreover, we observe that when more locations are encompassed in $B_{X_t}$, there is a reduced likelihood of rejecting the null hypothesis $\mathcal{H}_0$. This observation supports the design of MBI  in Algorithm \ref{al:MBD}: If we fail to reject $\mathcal{H}_0$, we consider that adding any additional locations to $B_{X_t}$ would further decrease the likelihood of rejecting $\mathcal{H}_0$. As a result, we consider that $B_{X_t}$ has been successfully identified and label $\mathcal{H}_0: X_t  \perp \!\!\!\perp X_{t-m-1} | B$ with $B\supset B_{X_t}$ as ``fail to reject''. \looseness = -1

\DEL{
\vspace{-0.00in}
Next, we provide detailed empirical analysis of the correlation between \emph{regions, speed, and time}, and the $p$-value of the CI test.

\vspace{0.03in}
\noindent \textbf{(1) $p$-value vs. regions}. As Fig. \ref{fig:regionvspvalueRome} and Fig. \ref{fig:regionvspvaluePorto} show, we define the approximated boundary of Rome and Porto using two bounding boxes (as described earlier). Due to insufficient data in two regions of the Rome dataset, we excluded these two regions from our analysis. We test the null hypothesis $\mathcal{H}_0: X_t  \perp \!\!\!\perp X_{t-m-1} | B_{X_t}$ for the trajectories using the CI test. In Fig. \ref{fig:regionvspvalueRome}(a)(b)(c)(d)(e), we display the heatmap of the $p$-values returned by the CI test across the 18 regions in Rome when $|B_{X_t}| = 1, 2, 3, 4, 5$, respectively. The figures show that when $m = 1$ or $2$, the average $p$-value is higher in downtown compared to the suburban area, suggesting that the hypothesis $\mathcal{H}_0$ is more likely to be rejected in the suburban area than in downtown when $m = 1$ or $2$. The heatmap for Porto in Fig. \ref{fig:regionvspvaluePorto}(a)(b)(c)(d)(e) reveals different conclusions. When $m = 1$, the average $p$-value is lower in downtown compared to the suburban area, and when $m = 2$, all $p$-value turn relatively high. 


\vspace{0.03in}
\noindent \textbf{(2) $p$-value vs. speeds}. 
We calculate the average speed of each \emph{examined sub-trajectory} $\{X_t, \hat{B}_{X_t}\}$, based on which we divide all the sub-trajectories into groups of 5 mph intervals across the entire speed range from 0 to 120 mph. We test the null hypothesis $\mathcal{H}_0$ for the trajectories in different speed intervals and display the corresponding $p$-values in the two cities in Fig. \ref{fig:pvaluevsspeed}(a)(b), respectively, given $|B_{X_t}| = 1, 2, 3, 4, 5$. Results from both cities reveal a similar trend: within a certain range, speed is negatively correlated with the $p$-value, but beyond this range, the correlation becomes positive, forming a U-shaped pattern. In other words, at very low and very high speeds, we fail to reject the null hypothesis, indicating a high correlation between trajectory location points. Conversely, at intermediate speed levels, we reject the null hypothesis and accept the alternative hypothesis, suggesting lower correlation between trajectory location points. Specifically, in the Rome dataset, $p$-values reach their minimum within the speed range of 75 to 100, before rising again, while in the Porto dataset, this minimum occurs at a higher speed range, highlighting city-specific differences.

\begin{figure*}[t]
\centering
\begin{minipage}{1.0\textwidth}
\centering
\vspace{-0.0in}
  \subfigure[\small $|B_{X_t}| = 1$ and $m = 1$]{
\includegraphics[width=0.19\textwidth, height = 0.088\textheight]{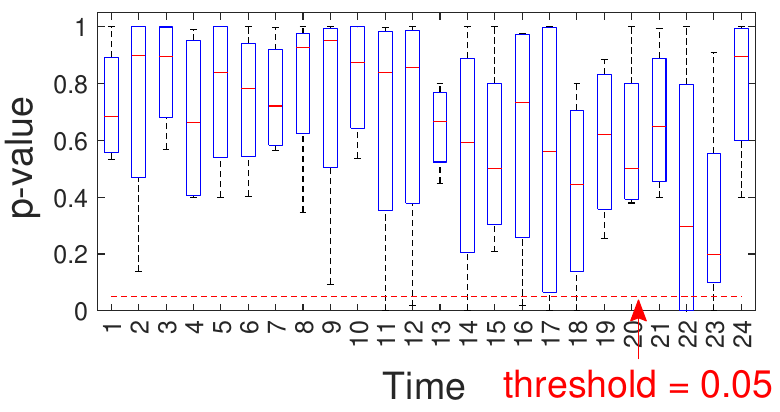}}
\vspace{-0.05in}
  \subfigure[\small $|B_{X_t}| = 2$ and $m = 2$]{
\includegraphics[width=0.19\textwidth, height = 0.088\textheight]{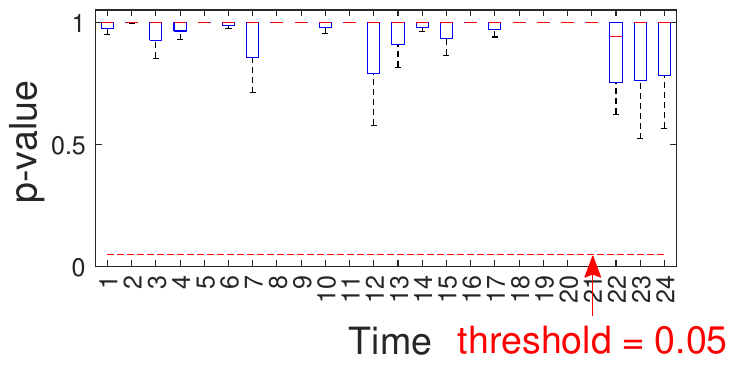}}
\vspace{-0.05in}
  \subfigure[\small $|B_{X_t}| = 3$ and $m = 3$]{
\includegraphics[width=0.19\textwidth, height = 0.088\textheight]{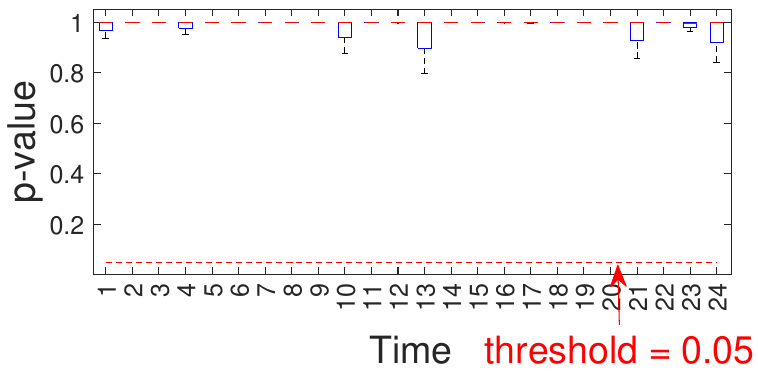}}
\vspace{-0.05in}
  \subfigure[\small $|B_{X_t}| = 4$ and $m = 4$]{
\includegraphics[width=0.19\textwidth, height = 0.088\textheight]{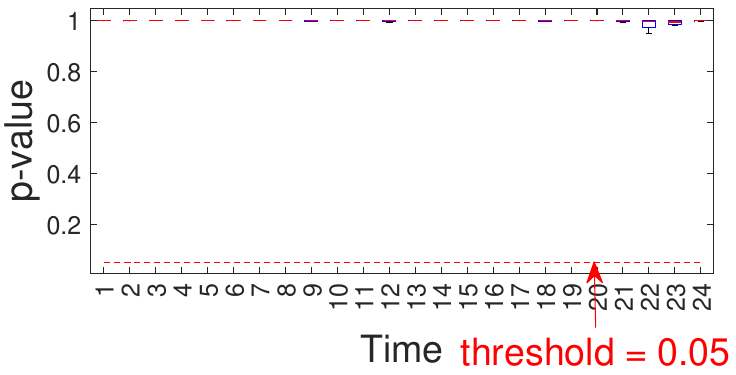}}
\vspace{-0.05in}
  \subfigure[\small $|B_{X_t}| = 5$ and $m = 5$]{
\includegraphics[width=0.19\textwidth, height = 0.088\textheight]{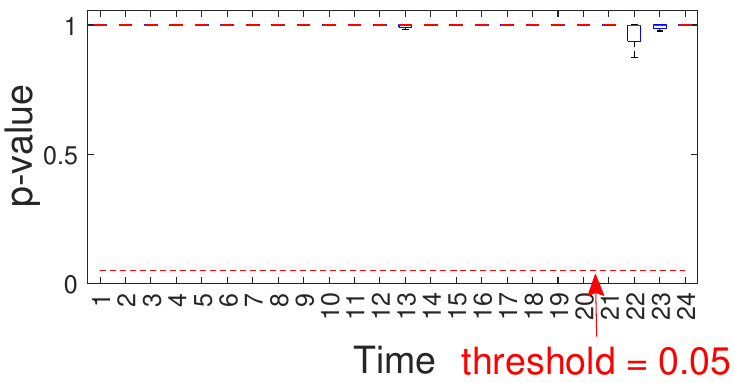}}
\end{minipage}
\vspace{-0.00in}
\caption{Relationship between time and $p$-values in Rome, Italy (returned by testing $\mathcal{H}_0$ in different intervals). }
\label{fig:pvaluevstimeRome}
\vspace{-0.12in}
\end{figure*}

\begin{figure*}[t]
\centering
\begin{minipage}{1.0\textwidth}
\centering
\vspace{-0.0in}
  \subfigure[\small $|B_{X_t}| = 1$ and $m = 1$]{
\includegraphics[width=0.19\textwidth, height = 0.088\textheight]{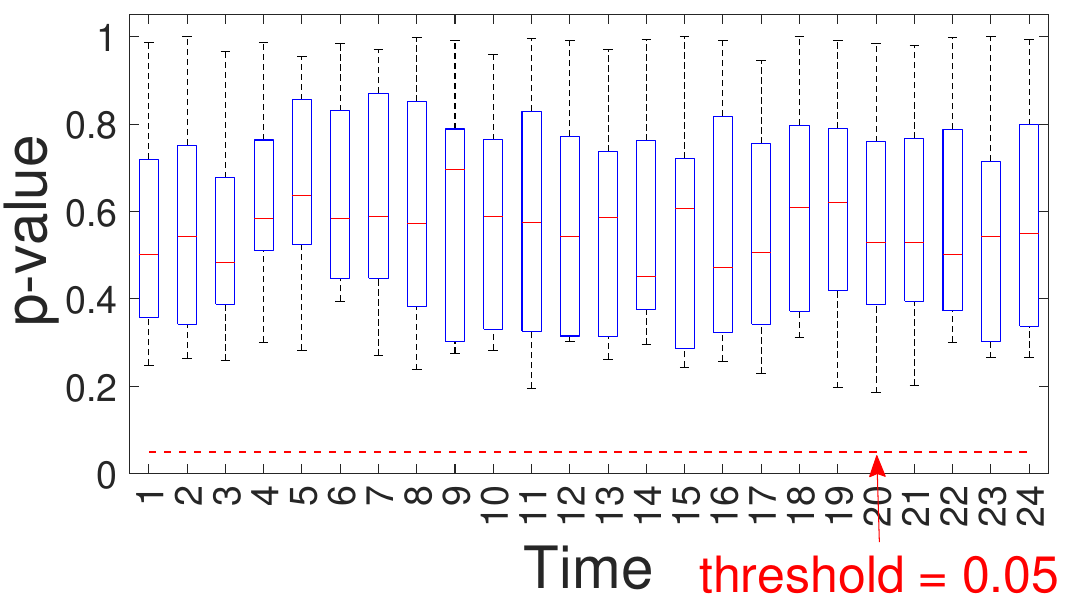}}
\vspace{-0.05in}
  \subfigure[\small $|B_{X_t}| = 2$ and $m = 2$]{
\includegraphics[width=0.19\textwidth, height = 0.088\textheight]{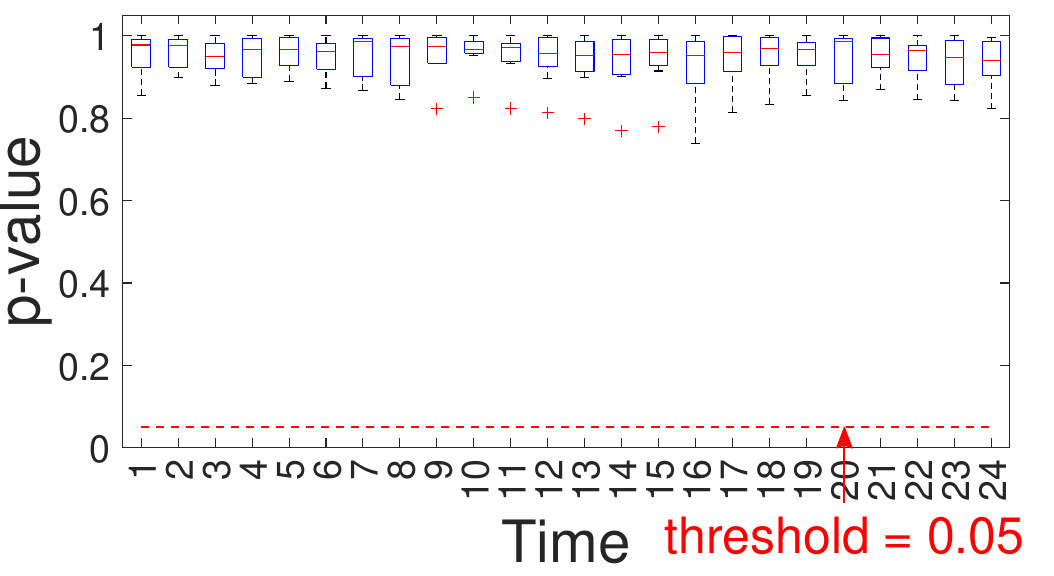}}
\vspace{-0.05in}
  \subfigure[\small $|B_{X_t}| = 3$ and $m = 3$]{
\includegraphics[width=0.19\textwidth, height = 0.088\textheight]{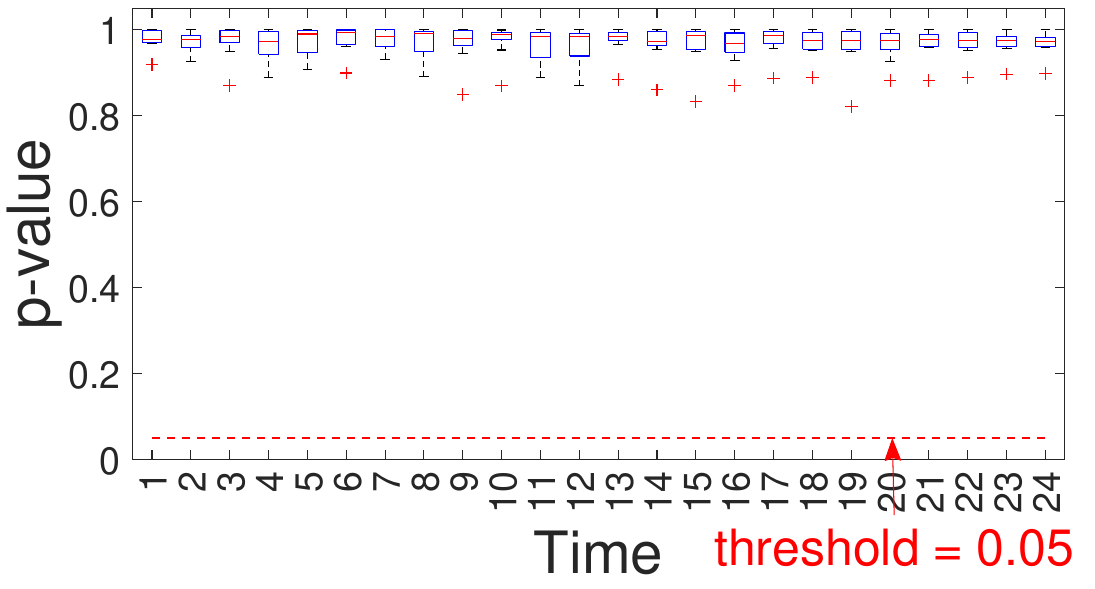}}
\vspace{-0.05in}
  \subfigure[\small $|B_{X_t}| = 4$ and $m = 4$]{
\includegraphics[width=0.19\textwidth, height = 0.088\textheight]{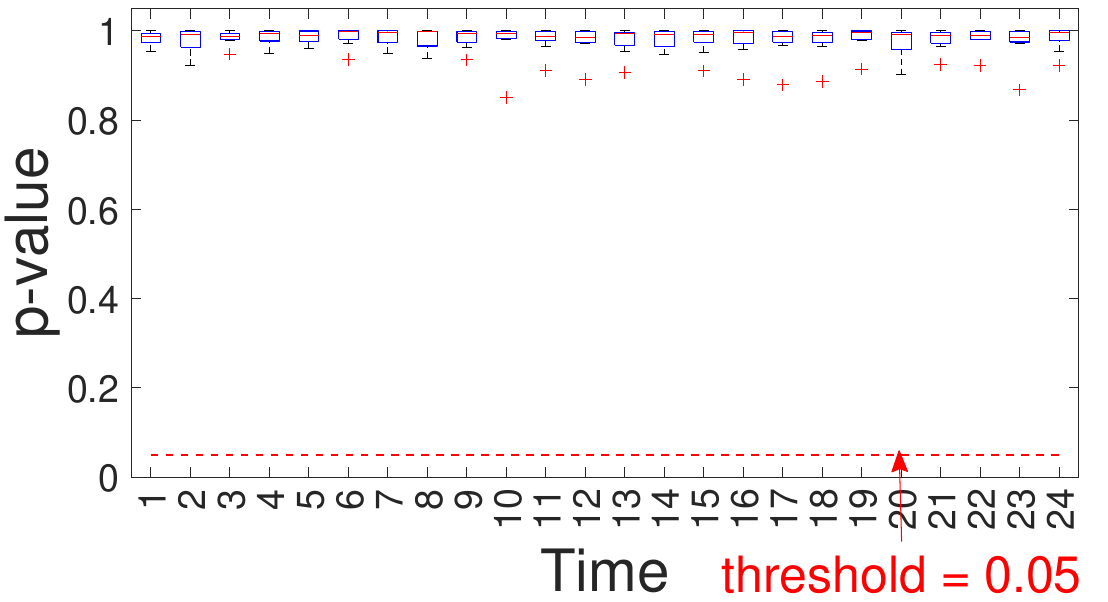}}
\vspace{-0.05in}
  \subfigure[\small $|B_{X_t}| = 5$ and $m = 5$]{
\includegraphics[width=0.19\textwidth, height = 0.088\textheight]{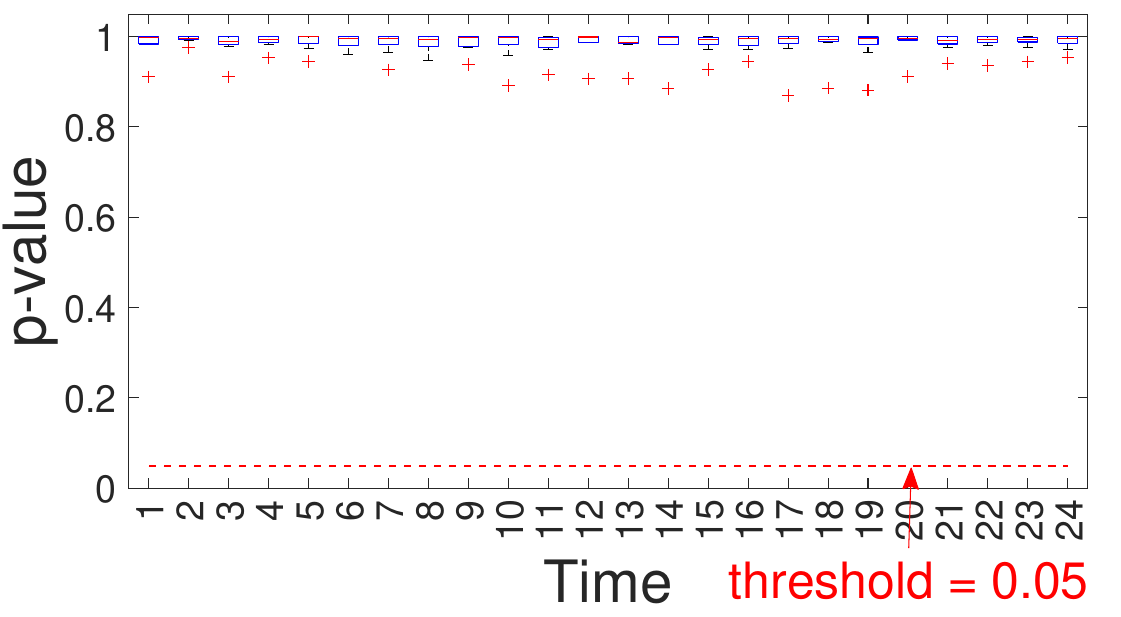}}
\end{minipage}
\vspace{-0.00in}
\caption{Relationship between time and $p$-values in Porto, Portugal (returned by testing $\mathcal{H}_0$ in different intervals). }
\label{fig:pvaluevstimePorto}
\vspace{-0.10in}
\end{figure*}


\vspace{0.03in}
\noindent \textbf{(3) $p$-value vs. time intervals}. 
According to the 24-hour time format, we categorize the trajectories into 24 groups [12 am--1 am), [1 am--2 am), ..., [11 pm--12 am) on the UTC +0 timestamps provided in the initial datasets. 
We test the null hypothesis $\mathcal{H}_0$ for the trajectories in different time intervals and display the corresponding $p$-values when $|B_{X_t}| = 1, 2, 3, 4, 5$ in Fig. \ref{fig:pvaluevstimeRome}(a)-(e) and Fig. \ref{fig:pvaluevstimePorto}(a)-(e), respectively. We observe that when $|B_{X_t}| = 3, 4, 5$, all the $p$-values are close to 1 in both cities. Conversely, when $|B_{X_t}| = 1, 2$, the $p$-values for Rome are relatively lower during the intervals 11-14 (11 am -- 2 pm ), 17-18 (5 pm -- 6 pm) and 22-23 (10 pm -- 11 pm). This observation aligns with the findings in Fig. \ref{fig:pvaluevsspeed}, where during these time intervals, vehicle traffic is relatively lower, indicating higher average vehicle speeds. Consequently, lower speeds, especially below 100 miles/hour, result in lower $p$ values. In contrast, the results for Porto appear more stable. 


From all the figures above, we find \textbf{the results of CI testing (the $p$-values) exhibit variations across different regions, time intervals, and speed intervals}. This observation challenges the assumption made by many related works, which posit that targets' mobility follows a first-order Markov process \cite{Qiu-SIGSPATIAL2022, Cao-ICDE2017}. 

Moreover, we observe that when more locations are encompassed in $B_{X_t}$, there is a reduced likelihood of rejecting the null hypothesis $\mathcal{H}_0$. This observation supports the design of MBI  in Algorithm \ref{al:MBD}: If we fail to reject $\mathcal{H}_0$, we consider that adding any additional locations to $B_{X_t}$ would further decrease the likelihood of rejecting $\mathcal{H}_0$. As a result, we consider that $B_{X_t}$ has been successfully identified and label $\mathcal{H}_0: X_t  \perp \!\!\!\perp X_{t-m-1} | B$ with $B\supset B_{X_t}$ as ``fail to reject''. \looseness = -1
}


\vspace{-0.00in}
\subsection{MBI Accuracy Evaluation} 
\label{subsec:MBIperformance}
\vspace{-0.00in}

\subsubsection{Experiment settings} 
\label{subsubsec:settings}
We implement the MBI predictor in PyTorch~\cite{PyTorch} as a three-layer fully connected binary classifier (Fig.~\ref{fig:MBIframework}(b)) with ReLU activations and a sigmoid output.
We split the dataset into training/validation/testing sets with a ratio of 60\%, 20\%, 20\%. The model is pretrained offline and only used for fast Markov-blanket prediction at runtime. As shown in Fig.~\ref{fig:MBItime}, the average inference time is 0.1557\,ms (Rome) and 0.1588\,ms (Porto), and over 98\% of predictions finish within 2\,ms. Additional training details are provided in Appendix~\ref{sec:additional_detail}.

{\revisiondone In our framework, we train the MBI module using the full city-scale trajectory corpus. We report utility/privacy of the proposed perturbation mechanism on a subset of about $500$ trajectories; this evaluation subset is sampled from the same corpus and therefore may overlap with the training data. The full city-scale trajectories are used to learn population-level regularities (e.g., how dependency strength and Markov-blanket size vary with coarse features), while our formal privacy budget $\epsilon$ is defined for the online release $Y=Q(X,V_X)$ (details in Appendix~\ref{sec:additional_detail}).}


\vspace{-0.02in}
\subsubsection{Experimental results} We display the performance of Markov Blanket Prediction in Table \ref{Tb:exp:accuracy}. The table shows that the prediction accuracy, defined as the proportion of instances that are accurately predicted, is 0.9512 and 0.9413 for Rome and Porto, respectively. Due to the imbalance between positive and negative samples, in order to capture the imbalanced binary classification of ``reject'' and ``fail to reject'', we also test the \emph{precision}, \emph{recall}, and \emph{F1 score} of the predictive model. In terms of positive F measures, the \emph{precision} and \emph{recall} for Rome are 0.9580 and 0.9868, for Porto are 0.9710 and 0.9630, suggesting that both models excel at identifying positive cases (``fail to reject'') with minimal false positives and negatives. The \emph{F1 scores}, at 0.9722 and 0.9670, confirm the well-balanced trade-off between precision and recall, further demonstrating that the models effectively handle both errors. The high \emph{recall} scores are particularly notable, indicating that the models are adept at identifying nearly all relevant positive instances, reducing the risk of missed detections.


\begin{figure}[t]
\centering
\begin{minipage}{0.240\textwidth}
\centering
\includegraphics[width=1.00\textwidth, height = 0.13\textheight]{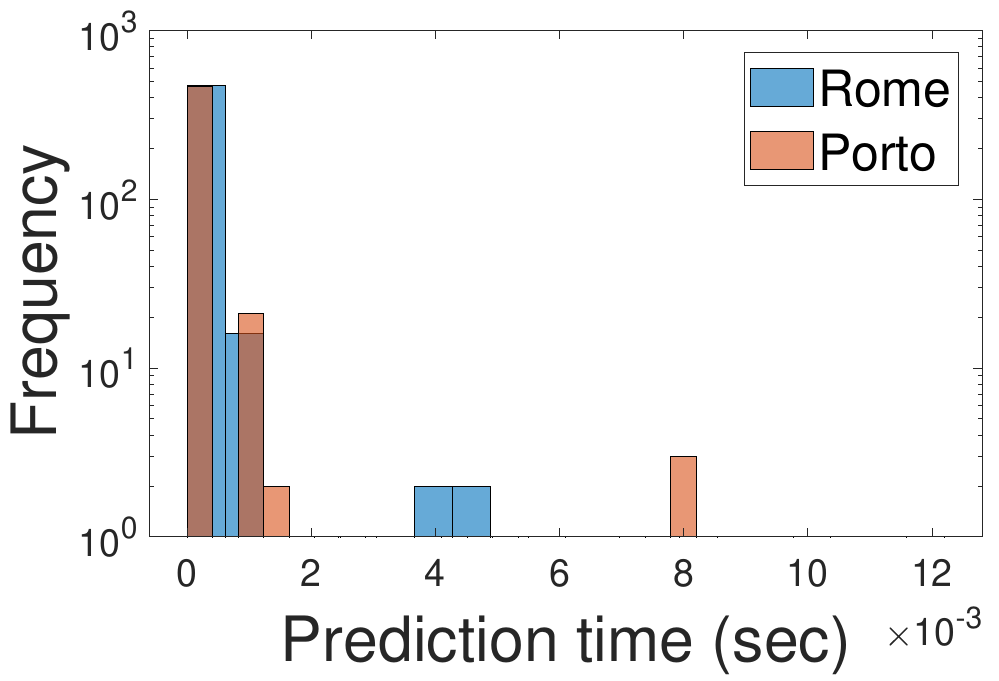}
\vspace{-0.18in}
\caption{Markov blanket prediction time distribution.}
\label{fig:MBItime}
\end{minipage}
\vspace{-0.05in}
\hspace{0.05in}
\begin{minipage}{0.203\textwidth}
\centering
  \subfigure{
\includegraphics[width=1.00\textwidth, height = 0.13\textheight]{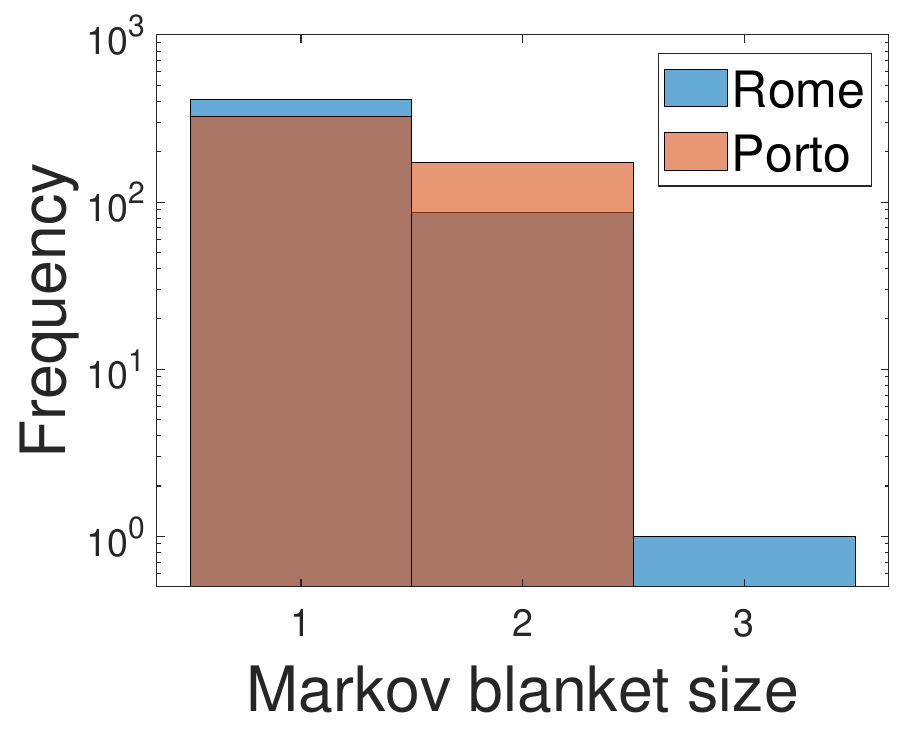}}
\vspace{-0.18in}
\caption{Markov blanket size distribution.}
\label{fig:MBI_size}
\end{minipage}
\vspace{-0.05in}
\end{figure}

Additionally, we calculate the \emph{negative precision}, \emph{negative recall}, and \emph{negative F1 score} by treating ``reject'' as a ``negative instance''. These metrics specifically assess the model's performance in predicting ``reject''. We can see that the models' performance declines when dealing with negative samples. For Rome, the table displays the \emph{negative precision} and \emph{negative recall} as 0.8972 and 0.7273. While the \emph{negative precision} remains relatively high, the \emph{negative recall} is notably lower, which suggests that the model is less effective in correctly identifying negative classes compared to its performance in the positive class, likely a result of the class imbalance within the dataset, where “fail to reject” dominates. For Porto, the \emph{negative precision} drops to 0.7123, significantly lower than Rome's 0.8972, indicating a higher rate of false negatives in predicting “reject” cases. However, \emph{negative recall} improves, reflecting a better capability in detecting true negatives compared to Rome. The \emph{negative F1 score} for Rome is 0.8033, which suggests that while the model maintains a reasonable balance, there remains some room for improvement. The \emph{negative F1 score} for Porto is slightly lower than Rome's, suggesting that while Porto's model is better at identifying negatives, the increased false negatives hinder its overall performance.

\vspace{-0.02in}
\begin{table}[h]
\caption{Prediction performance of DNN.}
\vspace{-0.02in}
\label{Tb:exp:accuracy}
\centering
\small 
\begin{tabular}{ c|c|c|c}
\toprule
City &\multicolumn{2}{ c|  }{Metrics} & Value \\
\hline
\hline 
& \multicolumn{2}{ |c|  }{BCE loss} & 0.1211 \\
\cline{2-4}
& \multicolumn{2}{ c|  }{ Prediction accuracy} & 0.9512  \\
\cline{2-4}
& \multicolumn{1}{ c|  }{ } & Precision & 0.9580
 \\ 
Rome, Italy & \multicolumn{1}{ c|  }{ } & Recall & 0.9868
 \\ 
& \multicolumn{1}{ c|  }{F measures} & F1 score & 0.9722
\\ 
\cline{3-4}
& \multicolumn{1}{ c|  }{ } & Negative precision & 0.8972
\\ 
& \multicolumn{1}{ c|  }{ } & Negative recall & 0.7273 \\ 
& \multicolumn{1}{ c|  }{} & Negative F1 score & 0.8033
\\ 
\hline
\hline
& \multicolumn{2}{ c|  }{BCE loss} & 0.1321 \\
\cline{2-4}
& \multicolumn{2}{ c|  }{ Prediction accuracy} & 0.9413  \\
\cline{2-4}
& \multicolumn{1}{ c|  }{ } & Precision & 0.9710
 \\ 
Porto, Portugal & \multicolumn{1}{ c|  }{ } & Recall & 0.9630
 \\ 
& \multicolumn{1}{ c|  }{F measures} & F1 score & 0.9670
\\ 
\cline{3-4}
& \multicolumn{1}{ c|  }{ } & Negative precision & 0.7123
\\ 
& \multicolumn{1}{ c|  }{ } & Negative recall & 0.7610 \\ 
& \multicolumn{1}{ c|  }{} & Negative F1 score & 0.7358
\\ 
\hline
\end{tabular}
\vspace{-0.05in}
\end{table}


We observe that the model's performance differs slightly between Rome and Porto, with Rome achieving higher overall prediction accuracy and a better balance in handling both positive and negative cases. Porto’s model, while more precise in predicting positive instances, struggles with higher false negatives when it comes to negative classification. 

Recall that the number of decision variables in the LP problem (Eq. (28)-(31)) is $O(|\mathcal{X}||\mathcal{Y}||\mathcal{B}|)$, where $|\mathcal{X}|$, $|\mathcal{Y}|$, and $|\mathcal{B}|$ respectively denote the size of the secret dataset, the size of perturbed dataset, and the number of possible values of the random variables within the Markov blanket. Fig. \ref{fig:MBI_size} shows the distributions of the predicted Markov blanket sizes in both cities. From the figure, we observe that 99.8\% and 100.0\% of the predicted Markov blanket sizes are no larger than 2 in Rome and Porto respectively. This indicates that the distribution of most locations primarily depends on their previous two locations. Hence, the number of possible Markov blanket states remains constant as the road network expands, as the number of neighboring locations within a road network is generally limited and does not scale with the network size.


\vspace{-0.00in}
\subsection{Utility Loss Evaluation} 
\vspace{-0.00in}
\label{subsec:UL}


\subsubsection{Utility loss measure} 
\label{subsubsec:ULmeasure} 
Given the task location $x_{\mathrm{task}}$ and the real (next) location $x_\ell$, the utility loss caused by a perturbed location $y_k$ is given by $\left|path(x_\ell,x_{\mathrm{task}}) - path(y_k,x_{\mathrm{task}})\right|$, where $path(x_\ell,x_{\mathrm{task}})$ (resp.  $path(y_k,x_{\mathrm{task}})$) represents the path distance from $x_\ell$ (resp.  $y_k$) to $x_{\mathrm{task}}$. The data utility loss $c_{(x_i,\mathbf{b}),y_k}$ is calculated by 
\vspace{-0.05in}
\begin{equation}
\label{eq:UL}
\small c_{(x_i,\mathbf{b}),y_k} =\sum_{x_{\mathrm{task}} \in \mathcal{V}}p_{x_{\mathrm{task}}} \cdot \sum_{x_\ell\in \mathcal{X}}p_{(x_\ell,\mathbf{b})}\cdot \left|path(x_\ell,x_{\mathrm{task}}) - path(y_k,x_{\mathrm{task}})\right|.
\end{equation}
where $p_{x_{\mathrm{task}}}$ is the prior probability of the task being at location $x_{\mathrm{task}}$ and $p_{(x_\ell,\mathbf{b})}$ is the probability of real location being at $x_\ell$ given the current location $x_i$ and the context information $\mathbf{b}$. Here, when computing the path distance a vehicle travels to reach its destination, we approximate both the vehicle's location and the destination location as ``nodes'' within the road network, where $\mathcal{V}$ is the node set. Then, we employ the Dijkstra algorithm \cite{Algorithm} to determine the shortest path distance from the vehicle's location to the destination node.

\vspace{0.03in}
\noindent \emph{Time complexity of utility loss measure:} Given a task with the location $x_{\mathrm{task}}$, we first build the shortest path tree rooted at $x_{\mathrm{task}}$ using Dijkstra's algorithm \cite{Algorithm}, of which the time complexity is $O(|\mathcal{V}|+|\mathcal{E}|)$, where $\mathcal{V}$ and $\mathcal{E}$
are the node (location) set and edge set of the road network of the target city, respectively. Since there are $|\mathcal{V}|$ possible task locations, the time complexity of building the shortest path trees for possible task locations is $O(|\mathcal{V}|^2+|\mathcal{E}||\mathcal{V}|)$. The shortest path tree can be used to calculate $path(x_i,x_{\mathrm{task}})$ and  $path(y_k,x_{\mathrm{task}})$ for all pairs $(x_i, y_k)$, with $O(|\mathcal{X}||\mathcal{V}|)$ ($|\mathcal{X}| \leq |\mathcal{V}|$) subtraction operations to calculate $c_{(x_i,\mathbf{b}),y_k}$ defined in Eq. (\ref{eq:UL}). Considering that there are $|\mathcal{X}|$ different possible $x_i$, $|\mathcal{Y}|$ different $y_k$, and $|\mathcal{B}|$ different Markov blanket $\mathbf{b}$, the time complexity of calculating all $c_{(x_i,\mathbf{b}),y_k}$ is $O(|\mathcal{B}||\mathcal{X}|^2|\mathcal{Y}||\mathcal{V}|+|\mathcal{V}|^2+|\mathcal{E}||\mathcal{V}|)$.


\DEL{
\vspace{0.02in}
\noindent {\revision \textbf{Case II: The error rate of finding the nearest vehicle.} In this case, we evaluate the error rate when identifying the nearest vehicle to a given task location $x_{\mathrm{task}}$ among a group of candidate vehicles (50 vehicles) randomly deployed within the target region. Given a vehicle, we perturbed its location and compute the utility loss of its perturbed location, $y_k$, relative to its actual location $x_i$ and context $\mathbf{v}$ as follows: If the perturbed location $y_k$ results in an incorrect outcome—either identifying the vehicle as the nearest when it is not, or failing to identify it as the nearest when it is—we classify it as an error. Formally, 
\vspace{-0.03in}
\begin{equation}
\nonumber 
I^{x_{\mathrm{task}}}_{x_\ell, y_k} = \left\{\begin{array}{ll} 1 & \mbox{if $path(x_\ell,x_{\mathrm{task}})$ has the shortest distance} \\ 
& \mbox{but $path(y_k,x_{\mathrm{task}})$ does not; } \\
1 & \mbox{if $path(x_\ell,x_{\mathrm{task}})$ doesn't have the shortest distance} \\ 
& \mbox{but $path(y_k,x_{\mathrm{task}})$ does;} \\
0 & \mbox{otherwise.} \end{array}\right.
\vspace{-0.00in}
\end{equation}
Then, the data utility loss is calculated by 
\begin{equation}
\vspace{-0.05in}
\label{eq:UL2}
\textstyle \overline{c}_{(x_i,\mathbf{b}),y_k} =\sum_{x_{\mathrm{task}} \in \mathcal{V}}p_{x_{\mathrm{task}}}\sum_{x_\ell\in \mathcal{X}}p_{x_\ell|(x_i,\mathbf{b})}I^{x_{\mathrm{task}}}_{x_\ell, y_k}.
\end{equation}
Considering that vehicles are deployed randomly, resulting different $I^{x_{\mathrm{task}}}_{(x_i, \mathbf{b}), y_k}$ each time. Therefore, we repeat deploying vehicles 10,000 times and calculate $I^{x_{\mathrm{task}}}_{x_\ell, y_k}$ for each trial. The average error over 10,000 trials is then calculated to quantify the utility loss caused by the perturbed location $y_k$ compared to the real location $x_i$:
$c_{(x_i,\mathbf{b}),y_k} =\frac{\sum_{\tau=1}^{10000} \overline{c}^\tau_{(x_i,\mathbf{b}),y_k}}{10000}$, where $\overline{c}^\tau_{(x_i,\mathbf{b}),y_k}$ represents the utility loss estimated in the $\tau$th trial.}

\vspace{0.03in} 
{\revision \noindent \emph{Time complexity of utility loss measure:} Similar to Case I, we first build the shortest path tree rooted at each task location $x_{\mathrm{task}}$. We then use the shortest path tree to find the nearest real vehicle location (which takes 50 comparisons, or $O(1)$ time complexity), and check whether $y_k$ results an error (which takes 1 comparison with the nearest vehicle), leading to totally $O(1)$ operations to calculate $I^{x_{\mathrm{task}}}_{x_\ell, y_k}$ and $O(|\mathcal{X}||\mathcal{V}|)$ to derive $\overline{c}_{(x_i,\mathbf{b}),y_k}$ (repeated for $T$ times ,where $T = 10,000$ in our experiment). 
Since there are $|\mathcal{X}|$ different possible $x_i$, $|\mathcal{Y}|$ different $y_k$, and $|\mathcal{B}|$ different Markov blanket $\mathbf{b}$, the total time complexity is $O(T|\mathcal{B}||\mathcal{X}|^2|\mathcal{Y}||\mathcal{V}|+|\mathcal{V}|^2+|\mathcal{E}||\mathcal{V}|)$.}}

\vspace{-0.02in}
\subsubsection{Compared methods} 
\label{subsubsec:comparedmethods}
In each dataset, we randomly sampled 500 trajectories, selecting one location within each trajectory to represent the target vehicles' positions. For each vehicle, we placed 1 destination within the target region. We evaluate our method ``\emph{LP+C-mDP}'' 
by comparing it with the following benchmarks, which are all based on mDP: 
\newline (1) ``\emph{LP}'' \cite{Qiu-TMC2020}, which optimizes mDP using LP framework with the consideration of vehicles' mobility in the road network.  Notably, other works such as \cite{Pappachan-EDBT2023, Qiu-EDBT2024} also fall into this category, as they both use an LP framework to optimize location perturbation. However, they differ in their approaches—\cite{Pappachan-EDBT2023} represents locations hierarchically, while \cite{Qiu-EDBT2024} incorporates quality constraints into the optimization process.
\newline (2) ``\emph{ExpMech}'' \cite{McSherry-FOCS2007}, where the perturbation probability of each secret record follows an exponential distribution. 
\newline (3) ``\emph{ConstOPTMech}'' or ``\emph{ConstOPT}'' \cite{ImolaUAI2022}, which integrates the exponential mechanism into the LP framework. 
{\revisiondone \newline (4) \emph{``LP+TrueMB''} follows our LP+C-mDP framework but replaces the Markov blanket predicted by the learning-based MBI module with the \emph{true} Markov blanket obtained via CI testing on the trajectory data. We use this benchmark as an oracle reference to quantify how blanket-prediction errors affect utility. Importantly, the MBI prediction errors do not change our privacy guarantee: C-mDP is enforced by constraints over the chosen secret domain and metric, independent of predictor accuracy, and the required distances are computed exactly from the map metric by the user/device. Consequently, prediction errors primarily impact utility (and efficiency) by selecting less/more relevant context, rather than weakening the mechanism-level privacy bound on the released output.}
\newline
(5) \emph{``LP+Markov''} removes the Markov blanket identification module (Section~\ref{sec:MBI}) and instead assumes a fixed-order Markov mobility model. Concretely, \emph{LP+Markov1} uses the first-order context $X_{t-1}$.
{\revisiondone
We do not include higher-order Markov baselines because the optimization cost grows rapidly with the model order: optimizing $q_{(x_t,x_{t-1}),y}$  already requires $\mathcal{O}(|\mathcal{X}|^2|\mathcal{Y}|)$ variables, whereas optimizing $q_{(x_t,x_{t-1},x_{t-2}),y}$ (i.e., a second-order Markov model) increases this to $\mathcal{O}(|\mathcal{X}|^3|\mathcal{Y}|)$, with a comparable blow-up in the number of constraints. In contrast, our Markov-blanket-based approach avoids committing to a fixed global order by selecting only the relevant past locations when needed.}

{\revisiondone 
Notably, the privacy budget $\epsilon$ is a \emph{shared scalar} across all guarantees, appearing in the standard likelihood-ratio bound in the form $\exp\!\big(\epsilon \cdot d(\cdot,\cdot)\big)$. However, the resulting indistinguishability is defined with respect to a particular \emph{secret domain} and its associated \emph{distance metric}. Context-free mechanisms (e.g., LP, ConstOPT, and ExpMech) enforce $\epsilon$-mDP on the single-location domain $\mathcal{X}$ using the base location metric $d_{x,x'}$. In contrast, context-aware mechanisms (including our C-mDP and the Markov-based variants) enforce the \emph{same} $\epsilon$ on an augmented secret domain (e.g., $\mathcal{X}\times\mathcal{V}$ or its Markov-blanket reduction) under an extended metric $d_{(x,\mathbf{v}),(x',\mathbf{v}')}$ that accounts for multiple (historical) locations. Thus, fixing $\epsilon$ keeps the privacy \emph{budget parameter} consistent across mechanisms, while changing the secret domain/metric changes \emph{what is being protected}; empirical comparisons at a fixed $\epsilon$ should therefore be interpreted as utility trade-offs under different protection targets (single-location secrecy versus context/trajectory-aware secrecy).
}

The detailed description of the objective functions of the LP-based methods, LP, ConstOPTMech, and LP+Markov, is given in Appendix \ref{sec:measureEUL}. 
We applied the above five mDP methods to perturb the vehicles' locations and recommended destinations based on their perturbed locations. We set $\eta = 5.0 km$  by default. 

\looseness = -1

\DEL{
\vspace{-0.00in}
\begin{table}[h]
\caption{Performance of different data perturbation methods.}
\vspace{-0.10in}
\label{Tb:exp:QL}
\centering
\small 
\begin{tabular}{ c|c||c|c}
\hline
\multicolumn{1}{ c  }{Perturbation}& \multicolumn{3}{ c }{Metrics} \\
\cline{2-4}
\multicolumn{1}{ c|  }{algorithms}
& \multicolumn{1}{ |c|| }{Utility loss}
& \multicolumn{1}{ c| }{PL}&\multicolumn{1}{ |c }{Expected PL } \\
\hline
\multicolumn{4}{ c }{Rome, Italy} 
\\
\hline
\hline
\multicolumn{1}{ c|  }{ \textbf{LP+C-mDP}} & \textbf{0.065$\pm$0.162} & \textbf{0.099$\pm$0.015} & \textbf{0.062$\pm$0.014}  \\ 
\multicolumn{1}{ c|  }{ LP+Markov} & 0.071$\pm$0.203 & 0.098$\pm$0.013 & 0.062$\pm$0.013  \\ 
\multicolumn{1}{ c|  }{ LP} & 0.077$\pm$0.216 & 0.100$\pm$0.000 & 0.064$\pm$0.009   \\
\multicolumn{1}{ c|  }{ConstOPT} & 0.084$\pm$0.263 & 0.100$\pm$0.000 & 0.073$\pm$0.009  \\ 
\multicolumn{1}{ c|  }{ ExpMech} & 0.146$\pm$0.326 & 0.074$\pm$0.010 & 0.028$\pm$0.008   \\ 
\hline
\multicolumn{4}{ c }{Porto, Portugal} 
\\
\hline
\hline
\multicolumn{1}{ c|  }{ \textbf{LP+C-mDP}} & \textbf{0.097$\pm$0.291} & \textbf{0.098$\pm$0.050} & \textbf{0.056$\pm$0.028}  \\ 
\multicolumn{1}{ c|  }{ LP+Markov} & 0.102$\pm$0.297 & 0.098$\pm$0.055 & 0.054$\pm$0.030  \\ 
\multicolumn{1}{ c|  }{ LP} & 0.110$\pm$0.317 & 0.100$\pm$0.000 & 0.051$\pm$0.017   \\
\multicolumn{1}{ c|  }{ConstOPT} & 0.130$\pm$0.376 & 0.100$\pm$0.000 & 0.069$\pm$0.010  \\ 
\multicolumn{1}{ c|  }{ ExpMech} & 0.234$\pm$0.270 & 0.072$\pm$0.023 & 0.023$\pm$0.015   \\ 
\hline
\end{tabular}
\vspace{-0.00in}
\end{table}
\vspace{-0.00in}
\begin{table}[t]
\caption{Utility loss of different data perturbation methods (mean$\pm$1.96$\times$standard deviation).}
\vspace{-0.10in}
\label{Tb:exp:QL}
\centering
\footnotesize 
\begin{tabular}{ c|c|c||c|c}
\hline
\multicolumn{1}{ c  }{Perturbation}& \multicolumn{2}{ |c|| }{Case I}
& \multicolumn{2}{ |c }{Case II} \\
\cline{2-5}
\multicolumn{1}{ c|  }{algorithms}
& \multicolumn{1}{ |c| }{Rome}
& \multicolumn{1}{ c|| }{Porto}&\multicolumn{1}{ |c }{Rome}&\multicolumn{1}{ |c }{Porto} \\
\hline
\multicolumn{1}{ c|  }{ \textbf{LP+C-mDP}} & \textbf{0.034$\pm$0.138} & \textbf{0.097$\pm$0.291} & \textbf{0.2586$\pm$1.60} & \textbf{0.2130$\pm$3.76} \\ 
\multicolumn{1}{ c|  }{ LP+Markov1} & 0.037$\pm$0.155 & 0.102$\pm$0.297 & 0.5887$\pm$4.10 & 0.4837$\pm$0.80\\ 
\multicolumn{1}{ c|  }{ LP+Markov2} & 0.037$\pm$0.155 & 0.102$\pm$0.297 & 0.5887$\pm$4.10 & 0.4837$\pm$0.80\\ 
\multicolumn{1}{ c|  }{ LP} & 0.044$\pm$0.138 & 0.110$\pm$0.317 & 0.6702$\pm$5.20 &  2.7031$\pm$22.16\\
\multicolumn{1}{ c|  }{ConstOPT} & 0.038$\pm$0.220 & 0.130$\pm$0.376 & 0.7752$\pm$6.60 & 2.8320$\pm$22.55\\ 
\multicolumn{1}{ c|  }{ ExpMech} & 0.045$\pm$0.177 & 0.234$\pm$0.270 & 1.612$\pm$15.35 & 6.6985$\pm$51.32\\ 
\hline
\end{tabular}
\vspace{-0.10in}
\end{table}}


\begin{figure}[t]
\centering
\begin{minipage}{0.5\textwidth}
\centering
\subfigure[\small Rome, Italy]{
\includegraphics[width=0.46\textwidth, height = 0.135\textheight]{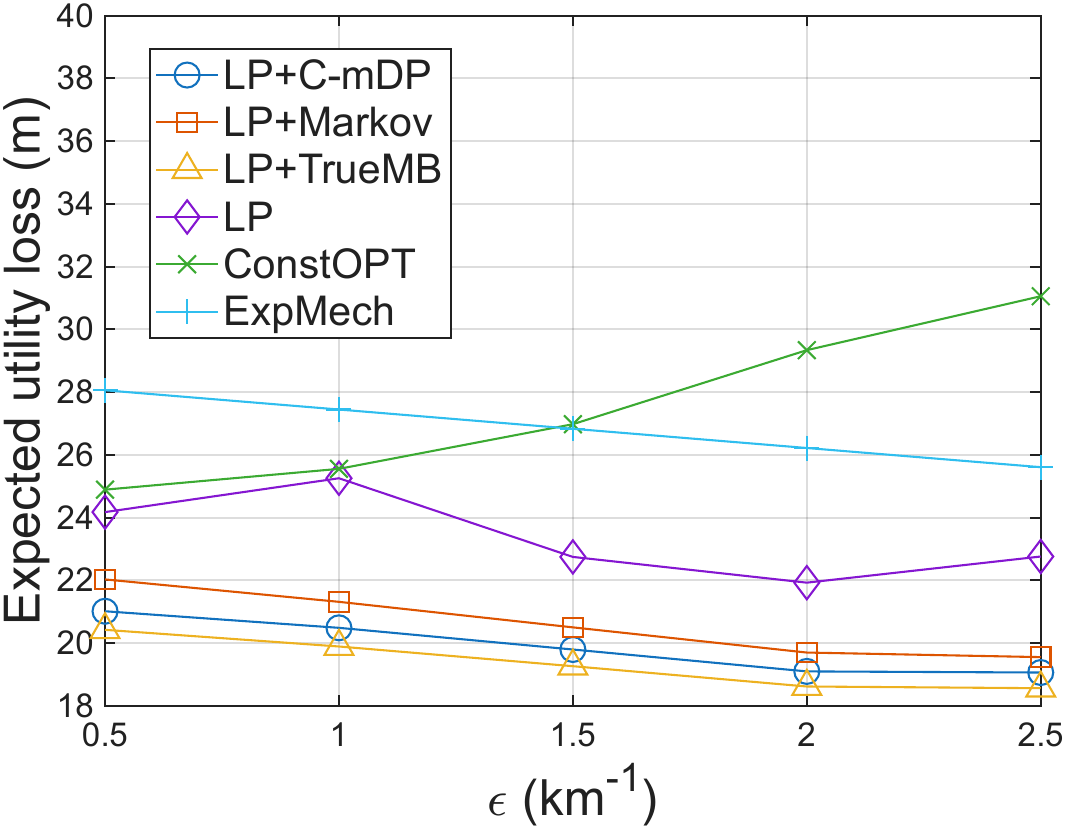}
}
\subfigure[\small Porto, Portugal]{
\includegraphics[width=0.46\textwidth, height = 0.135\textheight]{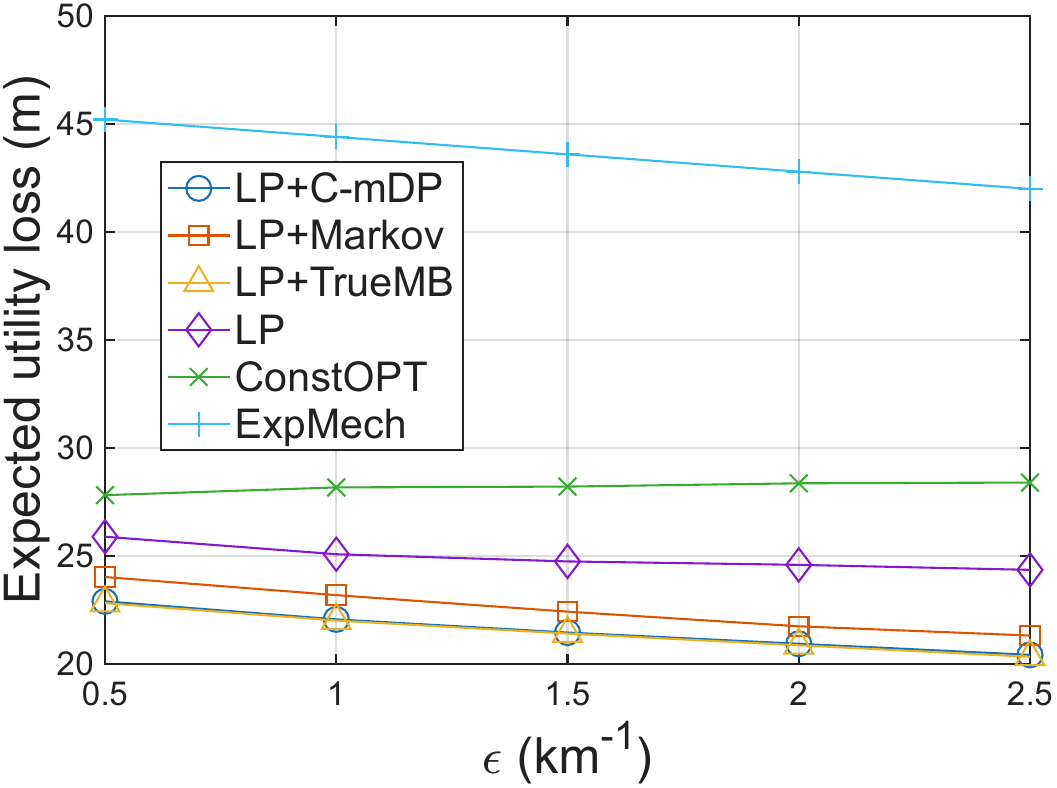}
}
\end{minipage}
\vspace{-0.15in}
\caption{Expected utility loss vs. privacy budget $\epsilon$.}
\label{fig:ULepsilon}
\vspace{-0.05in}
\end{figure}

\vspace{-0.05in}
\subsubsection{Experimental results} We evaluate the expected utility loss of the six data perturbation methods by varying the privacy budget $\epsilon$ from 0.5 $km^{-1}$ to 0.25 $km^{-1}$ and present the results in Fig. \ref{fig:ULepsilon}(a)(b). The experimental results demonstrate the superior performance of LP+C-mDP in terms of utility loss across both the Rome and Porto datasets. 
Across all evaluated $\epsilon$ values, LP+C-mDP achieves approximately $13\%$--$17\%$ lower utility loss than LP, $25\%$--$36\%$ lower than ConstOPT, and more than $25\%$--$50\%$ lower than ExpMech (with the largest gaps observed in Porto under moderate-to-large $\epsilon$). 
Moreover, LP+C-mDP performs very close to LP+TrueMB, typically within $1\%$--$2\%$, indicating that the C-mDP formulation effectively approximates the full Markov blanket dependency structure. 
In contrast, LP+Markov consistently incurs about $3\%$--$8\%$ higher utility loss than LP+C-mDP, reflecting the limitation of modeling vehicle mobility using only a first-order Markov assumption.

The performance gain of LP+C-mDP over LP stems from its incorporation of context information within the Markov blanket, which enables a more accurate characterization of context-dependent utility loss during the optimization process. 
By contrast, ConstOPT and ExpMech rely (fully or partially) on exponential mechanisms that do not explicitly account for road-network constraints and mobility dependencies, leading to higher utility loss. 
Furthermore, LP+Markov assumes a first-order Markov chain and ignores additional preceding locations that are not conditionally independent of the next location, which explains its consistently higher utility loss compared to LP+C-mDP.



\begin{figure}[t]
\centering
\begin{minipage}{0.48\textwidth}
\centering
  \subfigure[Perturbation optimization]{
\includegraphics[width=0.48\textwidth, height = 0.13\textheight]{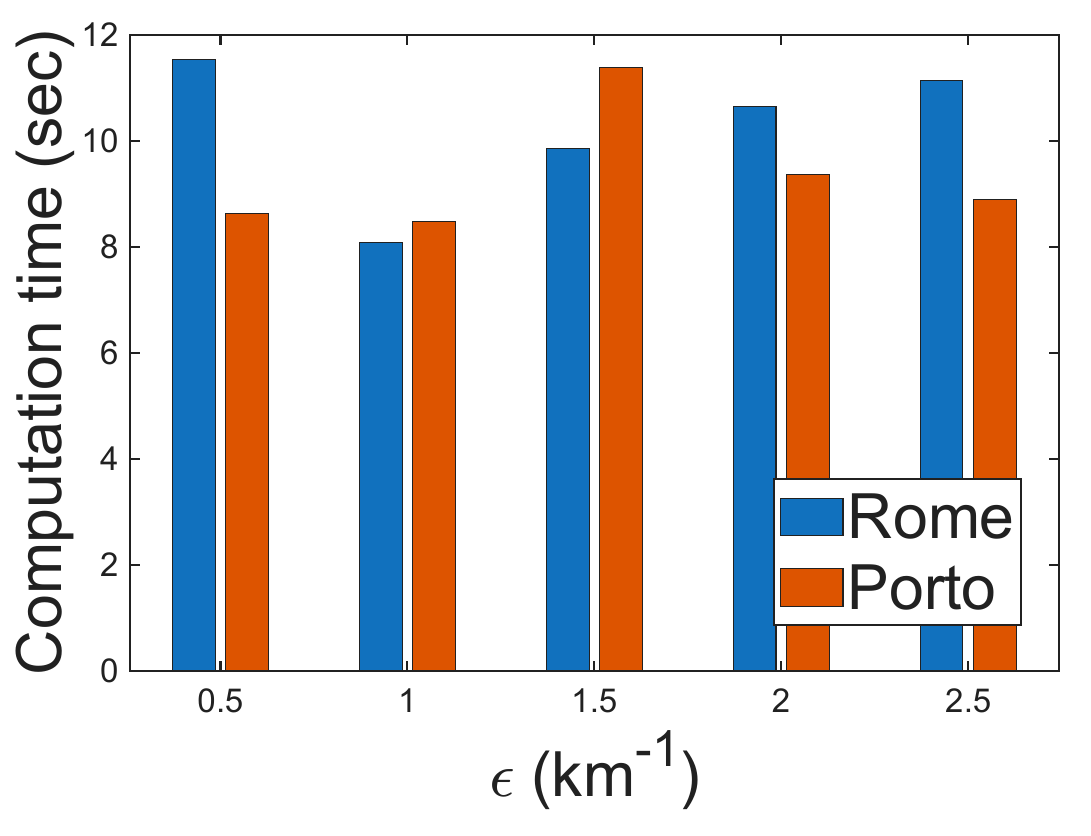}}
  \subfigure[Perturbation record selection]{
\includegraphics[width=0.48\textwidth, height = 0.13\textheight]{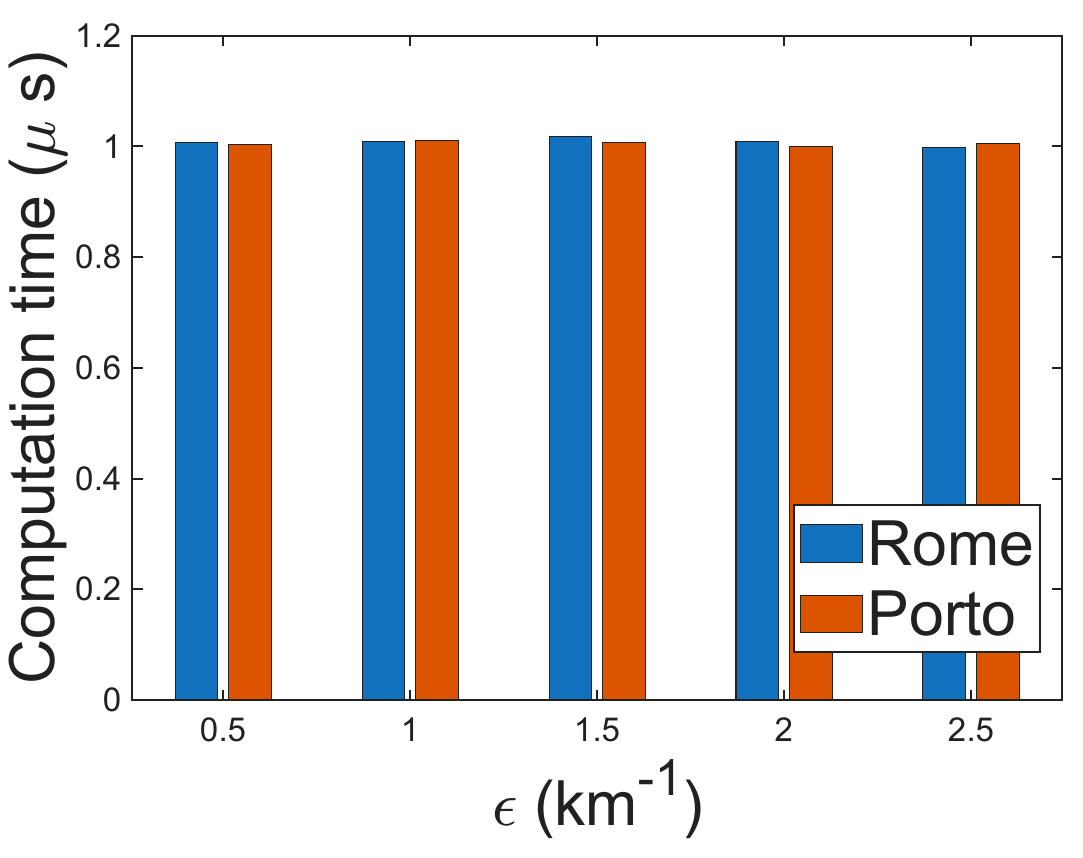}}
\label{}
\end{minipage}
\vspace{-0.15in}
\caption{Computation time of data perturbation.}
\label{fig:LPtime}
\vspace{-0.05in}
\end{figure}

\vspace{0.02in}
\noindent \textbf{Tradeoff between privacy and utility}. 
In the mDP framework, the privacy budget $\epsilon$ represents the required privacy level. A lower $\epsilon$ indicates a higher privacy level, requiring more noise to be added to the protected data, which can degrade utility due to reduced data fidelity. 
As expected, a clear privacy--utility tradeoff is observed in both Fig. \ref{fig:ULepsilon}(a)(b). As the privacy budget $\epsilon$ increases, the utility loss of all mechanisms (LP+C-mDP, LP, ExpMech, LP+TrueMB, and LP+Markov) decreases. This trend arises because a larger $\epsilon$ relaxes the privacy constraints, enabling the mechanism to select perturbed locations that are closer to the true location.  While this leads to improved utility by preserving higher location accuracy, it simultaneously weakens privacy protection, as the released data reveals more information about the original location.



\vspace{-0.00in}
\DEL{
\begin{table}[t]
\caption{Prediction performance of DNN (the 3rd column displays mean$\pm$standard deviation values across the groups with different sample sizes)}
\vspace{-0.10in}
\label{Tb:exp:accuracy}
\centering
\footnotesize 
\begin{tabular}{ c|c|c|c}
\hline 
\hline 
\multicolumn{2}{ c|  }{Merics} & Overall & Diff. sample size \\
\hline
\hline 
\multicolumn{2}{ c|  }{BCE loss} & 0.1211 & 0.1736$\pm$0.1286\\
\hline
\multicolumn{2}{ c|  }{ Prediction accuracy} & 0.9512 & 0.9221$\pm$0.0796 \\
\hline
\multicolumn{1}{ c|  }{ } & Positive precision & 0.9580 & 0.9247$\pm$0.0735
\\ 
\multicolumn{1}{ c|  }{ } & Positive recall & 0.9868 & 0.9890$\pm$0.0246\\ 
\multicolumn{1}{ c|  }{\revision F-scores} & Positive F1 score & 0.9722 & 0.9543$\pm$0.0455
\\ 
\cline{2-4}
\multicolumn{1}{ c|  }{ } & Negative precision & 0.8972 & 0.7847$\pm$0.3829
 \\ 
\multicolumn{1}{ c|  }{ } & Negative recall & 0.7273 & 0.6384$\pm$0.3664
 \\ 
\multicolumn{1}{ c|  }{ } & Negative F1 score & 0.8033 & 0.6839$\pm$0.3604
\\ 
\hline
\end{tabular}
\vspace{-0.15in}
\end{table}}




\DEL{

\begin{figure*}[t]
\centering
\begin{minipage}{1.0\textwidth}
\centering
\subfigure[\small Case I: Rome, Italy]{
\includegraphics[width=0.23\textwidth, height = 0.105\textheight]{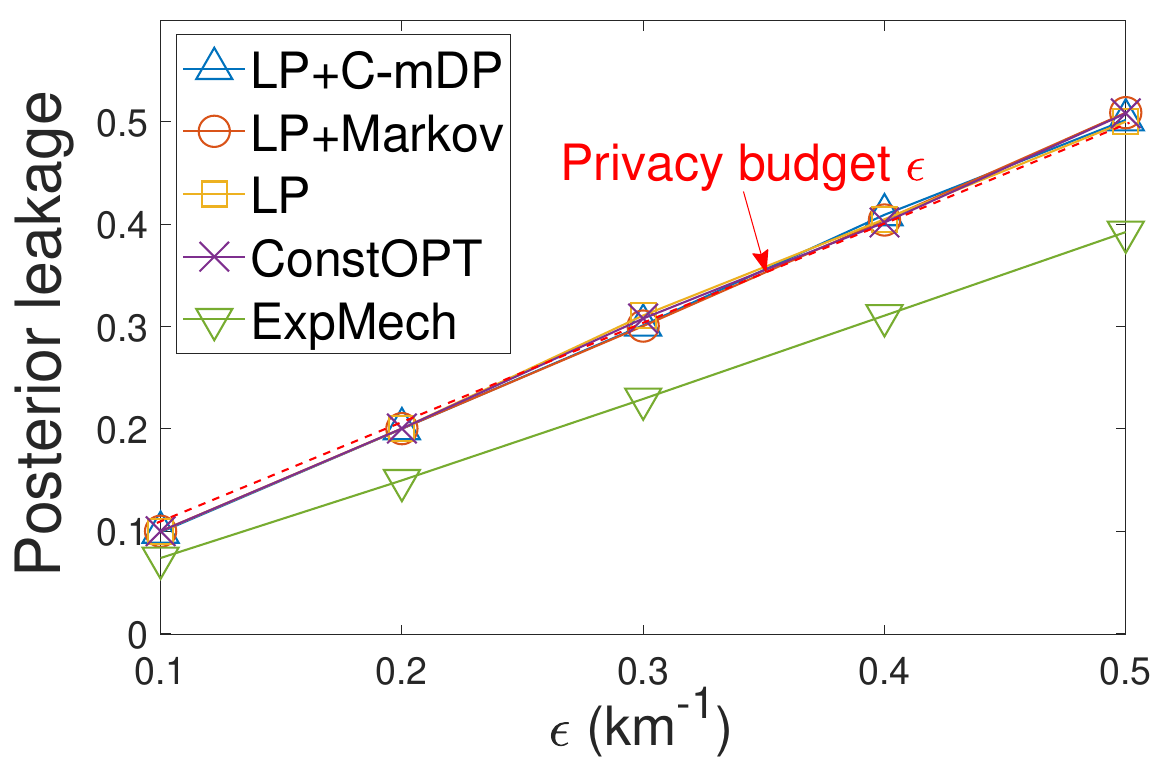}
}
\vspace{-0.00in}
\subfigure[\small Case I: Porto, Portugal]{
\includegraphics[width=0.23\textwidth, height = 0.105\textheight]{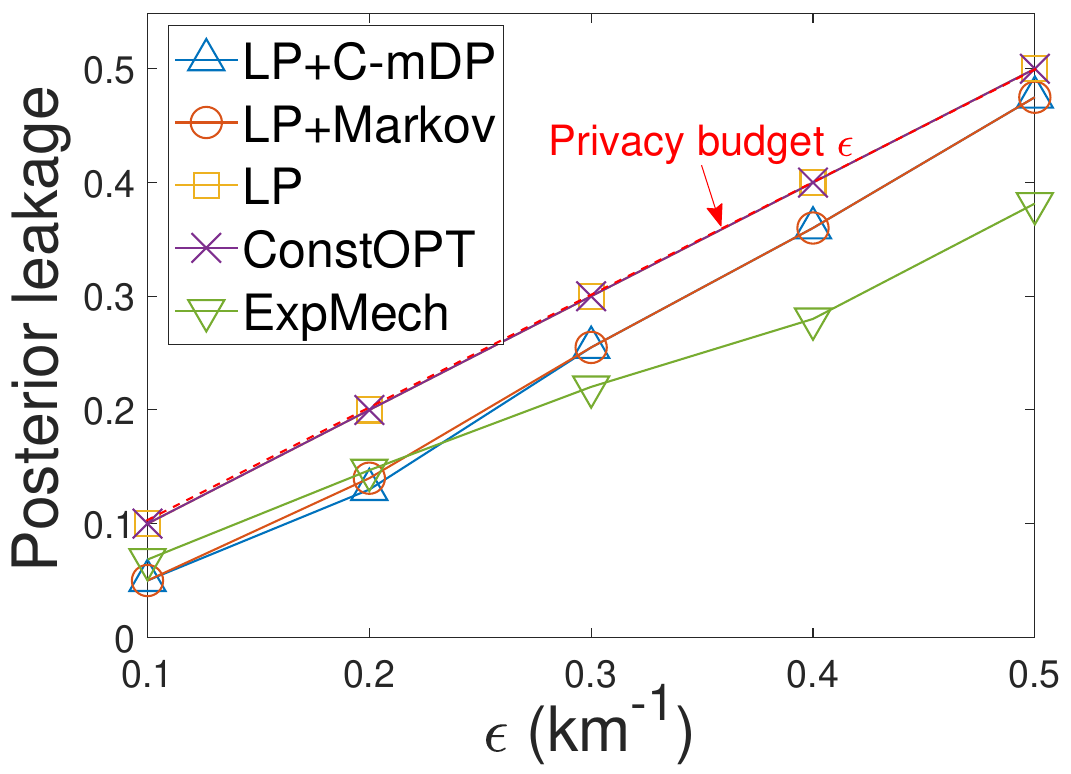}
}
\subfigure[\small Case II: Rome, Italy]{
\includegraphics[width=0.23\textwidth, height = 0.105\textheight]{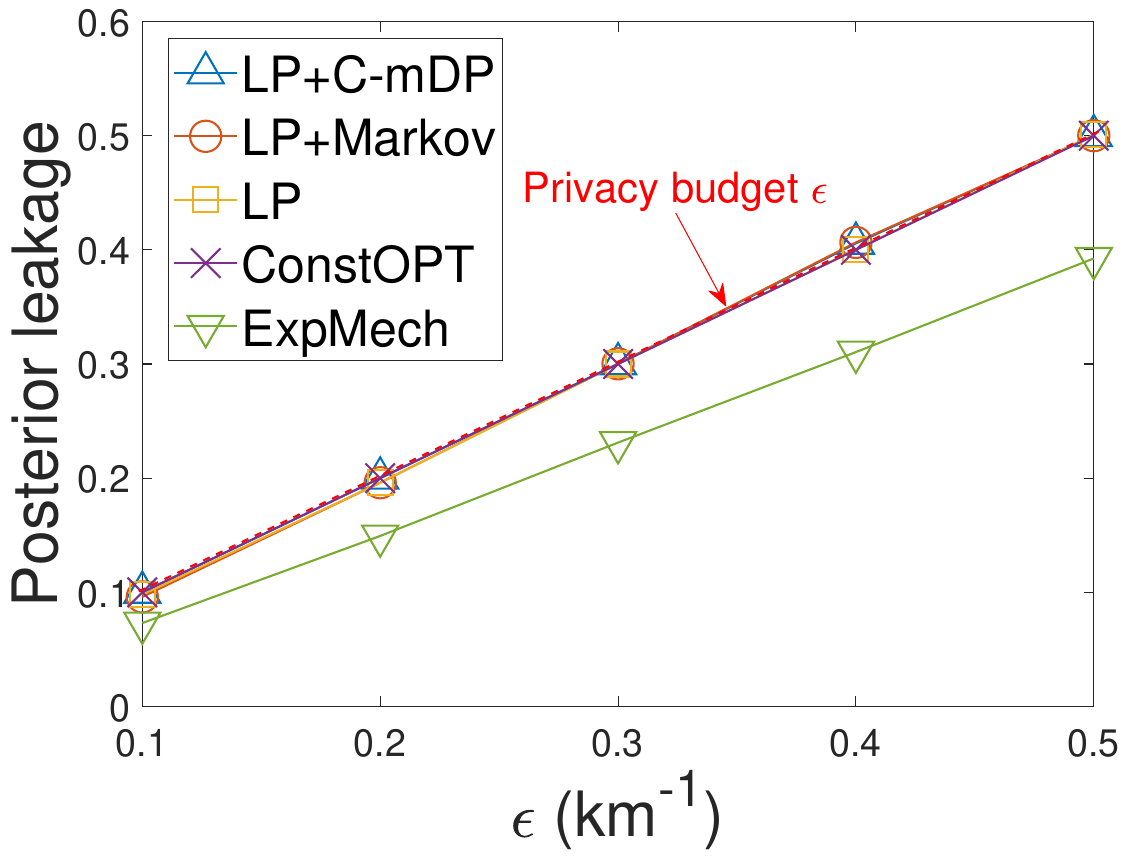}
}
\subfigure[\small Case II: Porto, Portugal]{
\includegraphics[width=0.23\textwidth, height = 0.105\textheight]{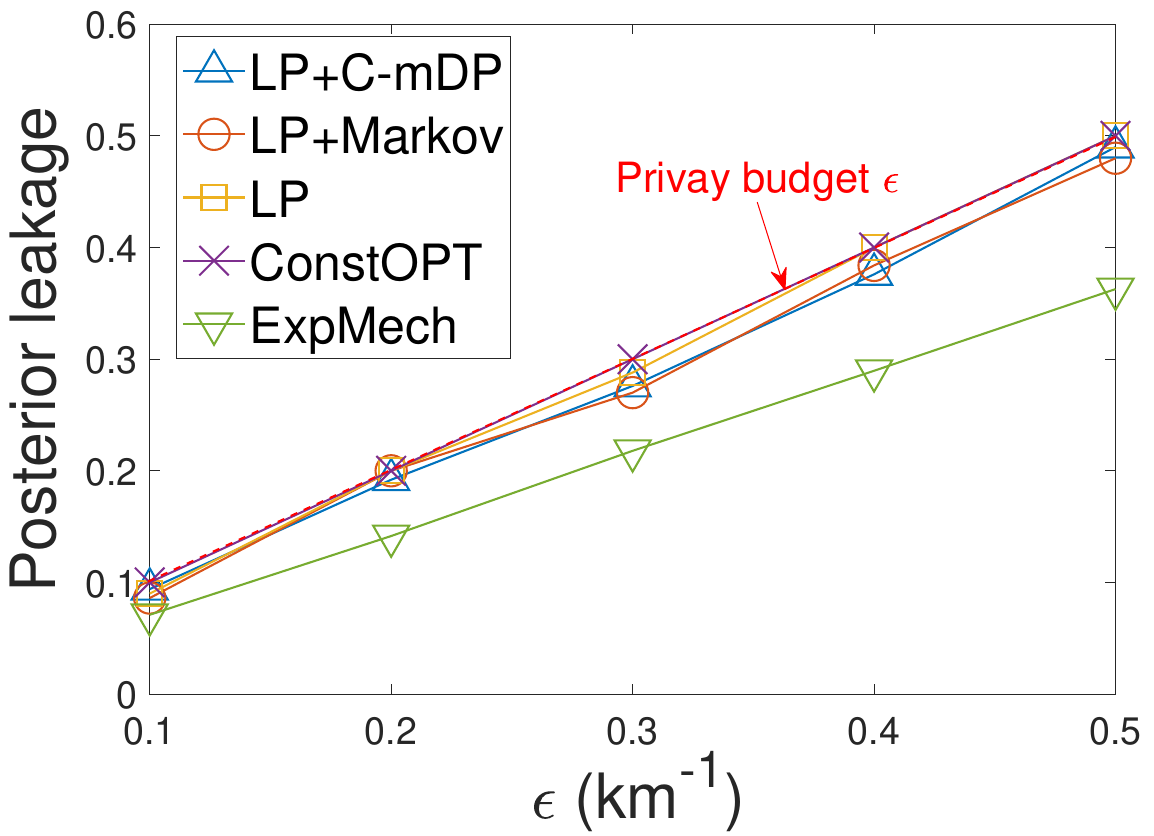}
}
\end{minipage}
\vspace{-0.20in}
\caption{Posterior leakage vs. privacy budget $\epsilon$.}
\label{fig:PLepsilon}
\vspace{-0.15in}
\end{figure*}

\begin{figure*}[t]
\centering
\begin{minipage}{1.0\textwidth}
\centering
\subfigure[\small Case I: Rome, Italy]{
\includegraphics[width=0.23\textwidth, height = 0.105\textheight]{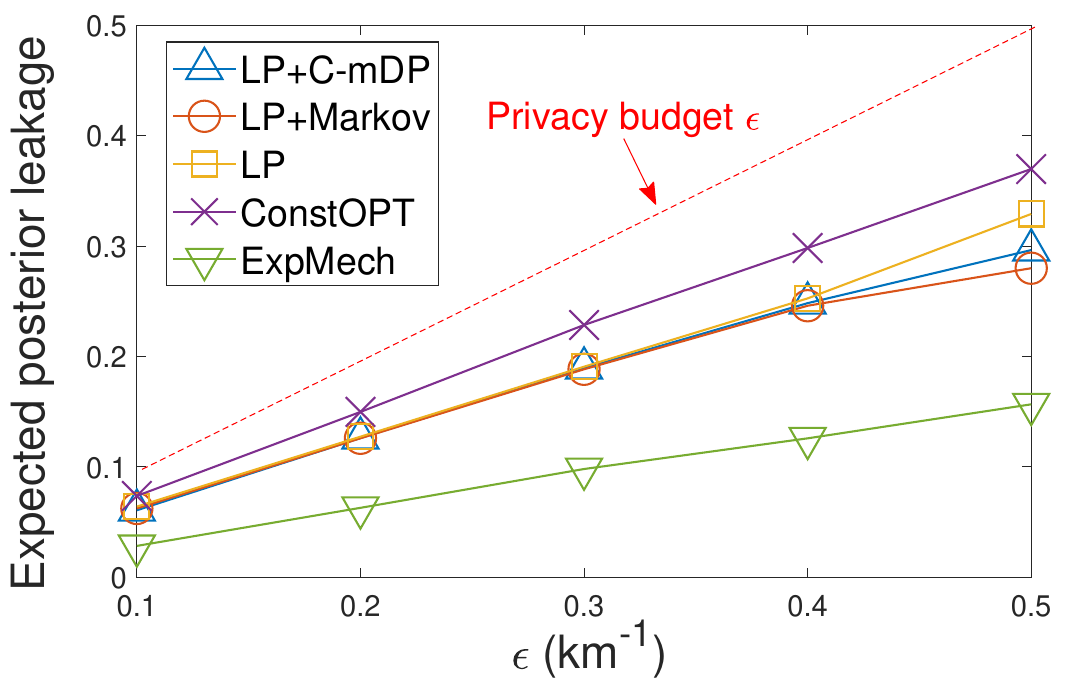}
}
\vspace{-0.00in}
\subfigure[\small Case I: Porto, Portugal]{
\includegraphics[width=0.23\textwidth, height = 0.105\textheight]{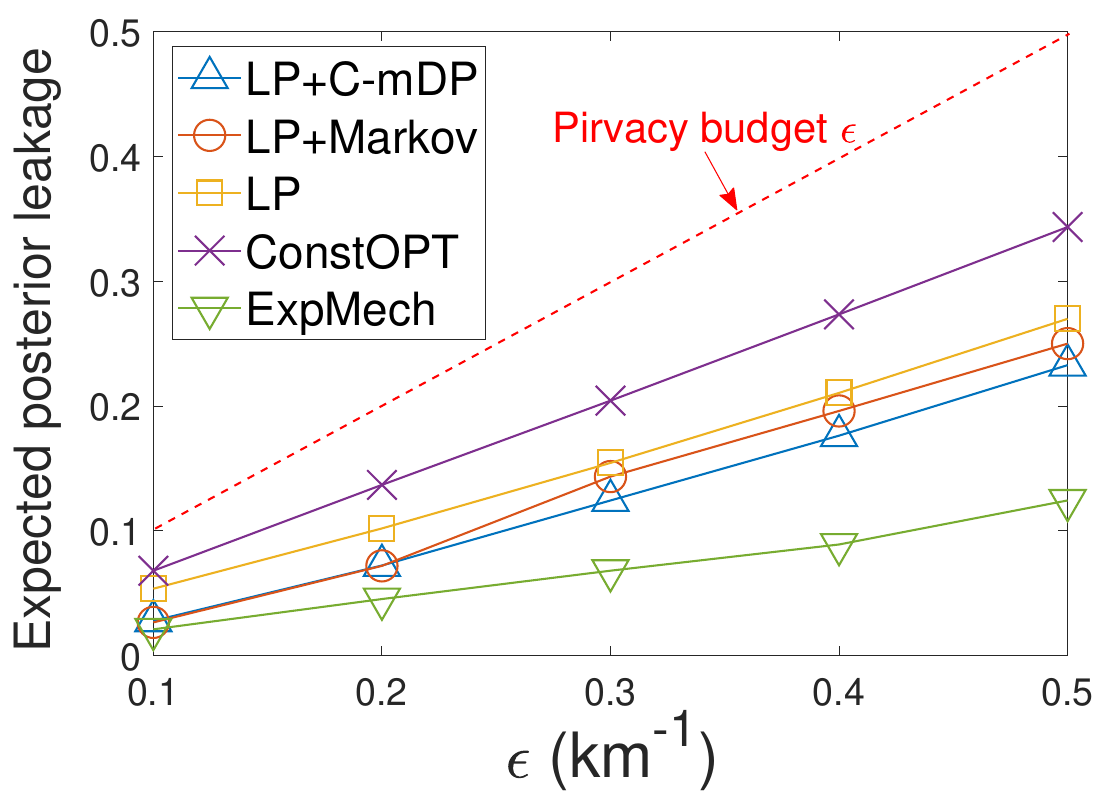}
}
\subfigure[\small Case II: Rome, Italy]{
\includegraphics[width=0.23\textwidth, height = 0.105\textheight]{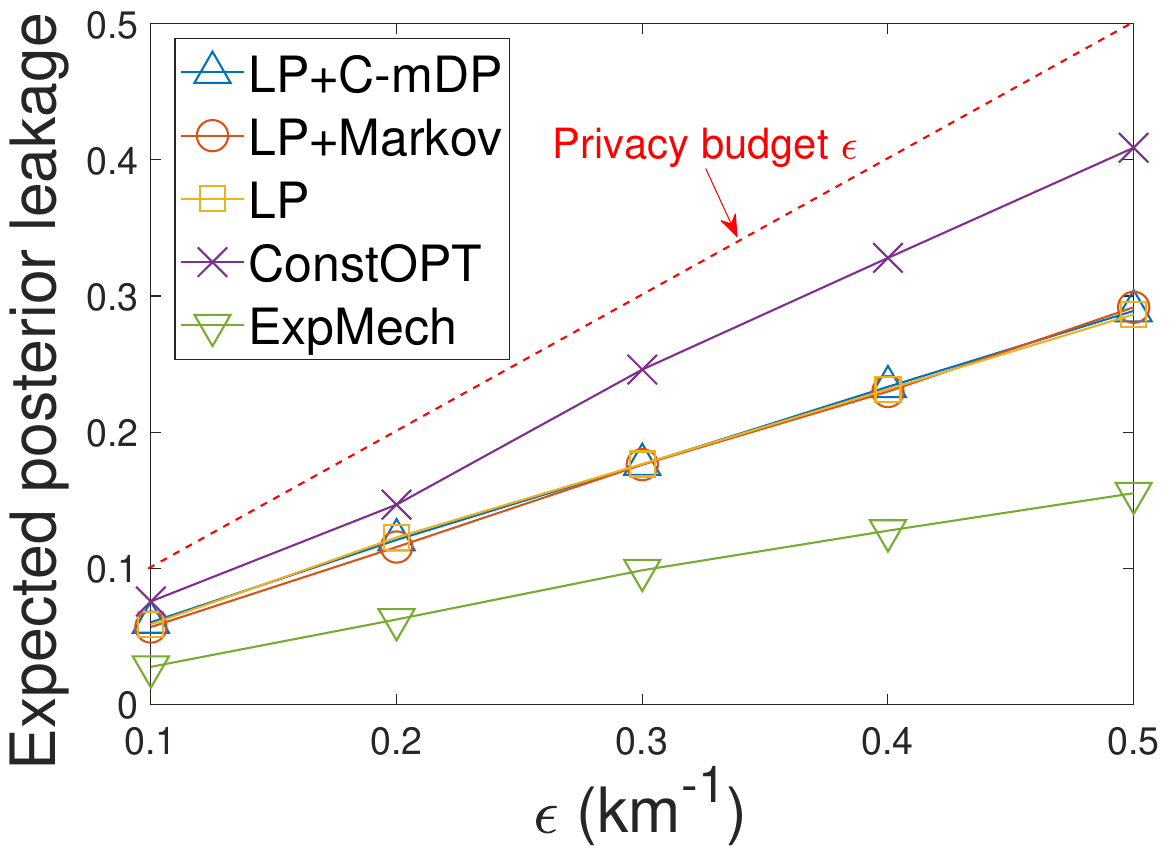}
}
\subfigure[\small Case II: Porto, Portugal]{
\includegraphics[width=0.23\textwidth, height = 0.105\textheight]{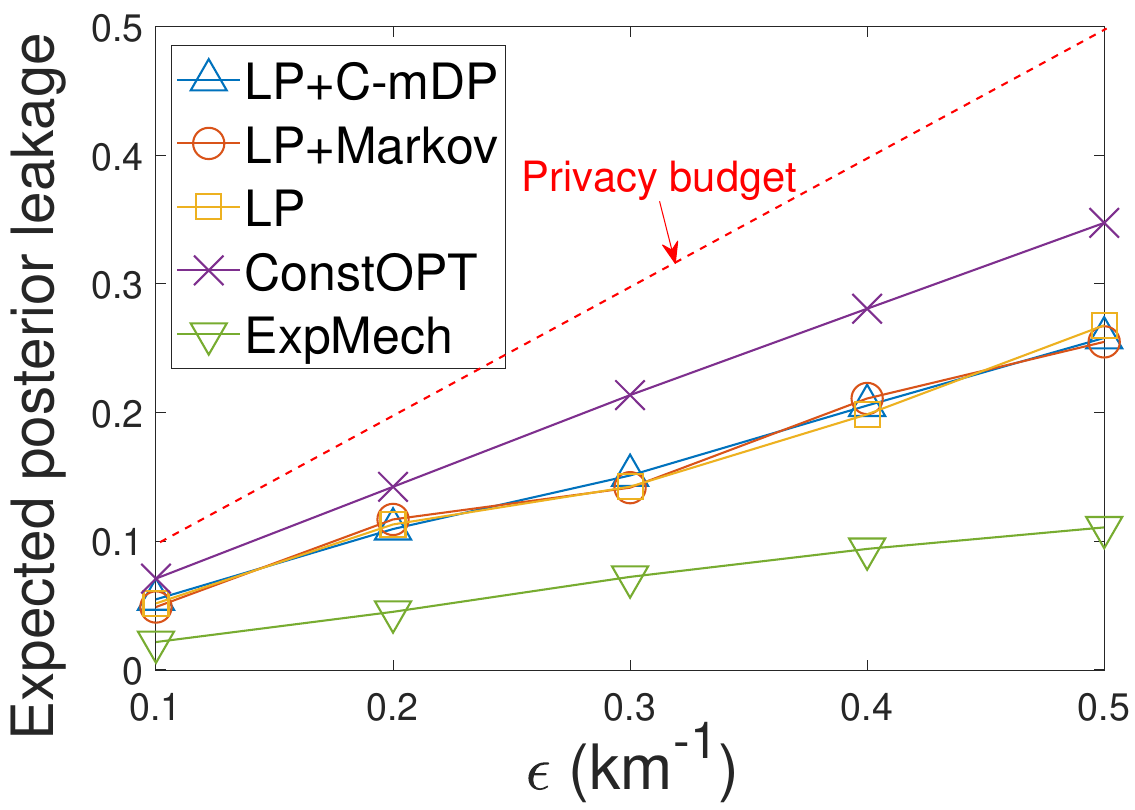}
}
\label{}
\end{minipage}
\vspace{-0.20in}
\caption{Expected posterior leakage vs. privacy budget $\epsilon$.}
\label{fig:exPLepsilon}
\vspace{-0.15in}
\end{figure*}

\begin{figure*}[t]
\centering
\begin{minipage}{1.0\textwidth}
\centering
\subfigure[\small Case I: Rome, Italy]{
\includegraphics[width=0.23\textwidth, height = 0.105\textheight]{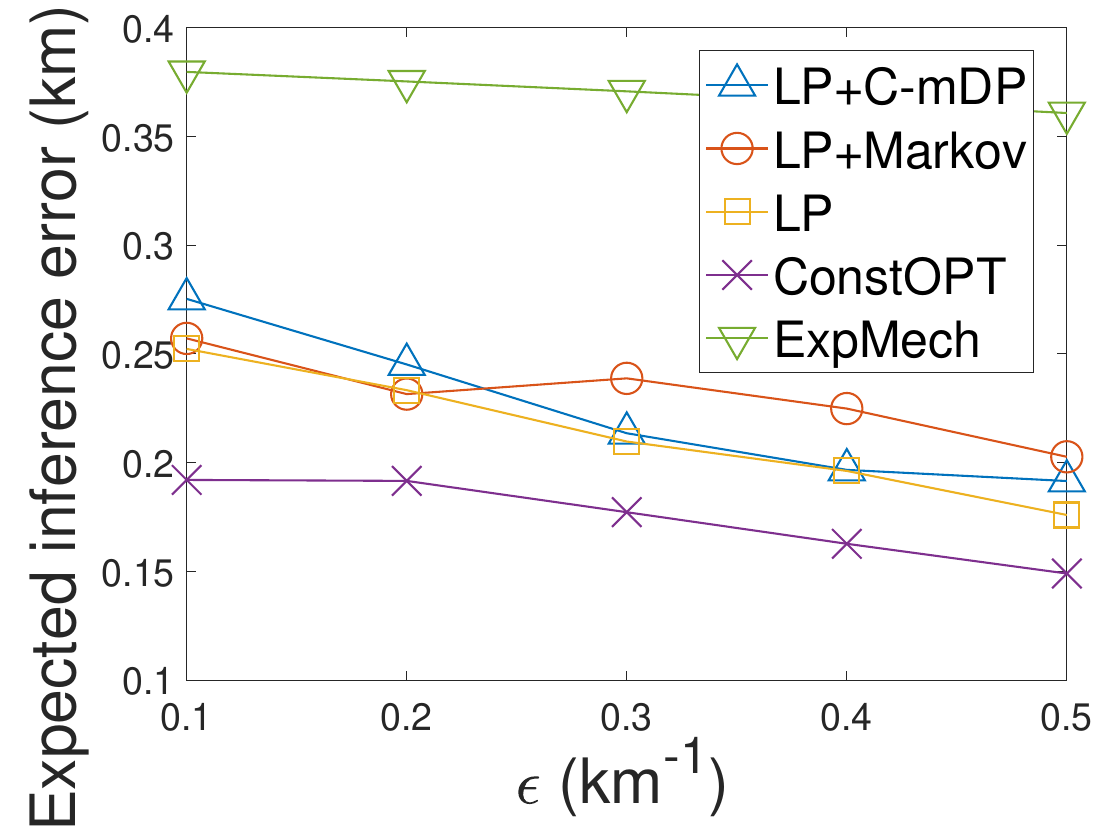}
}
\vspace{-0.00in}
\subfigure[\small Case I: Porto, Portugal]{
\includegraphics[width=0.23\textwidth, height = 0.105\textheight]{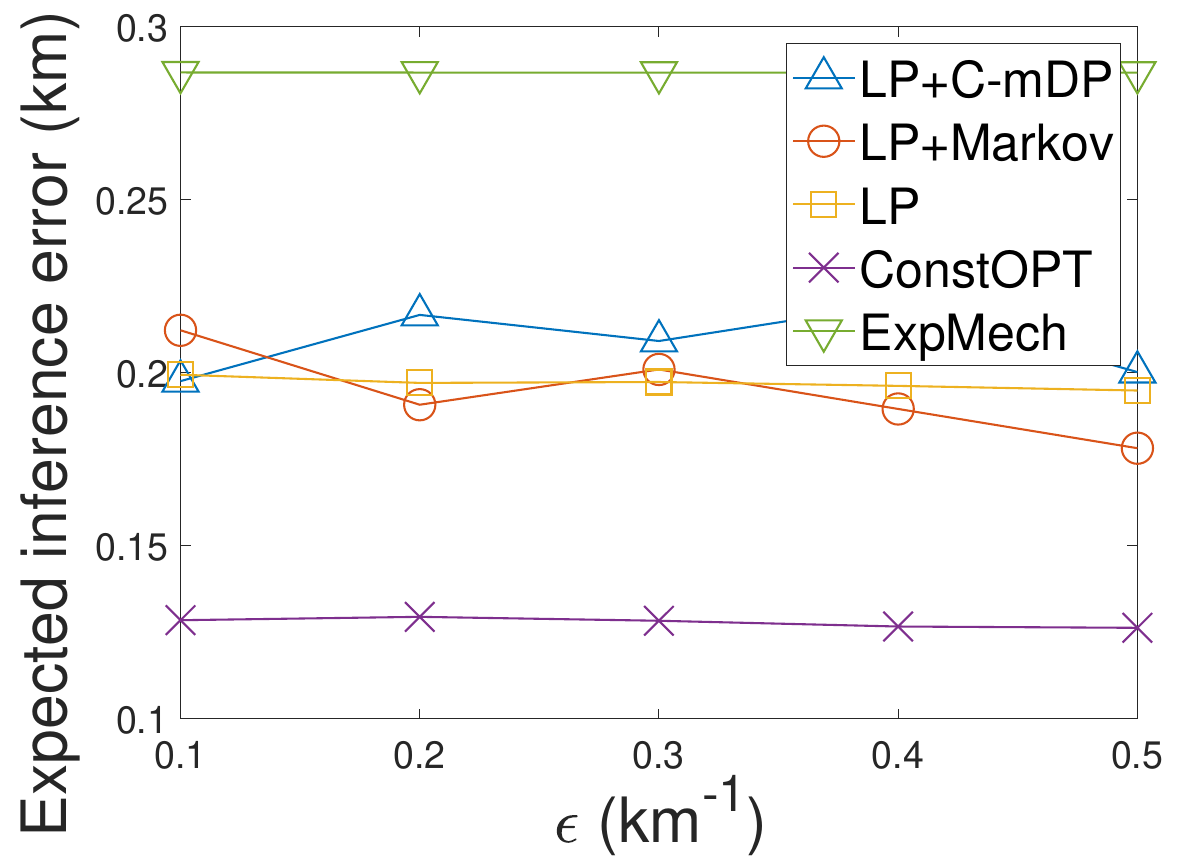}
}
\subfigure[\small Case II: Rome, Italy]{
\includegraphics[width=0.23\textwidth, height = 0.105\textheight]{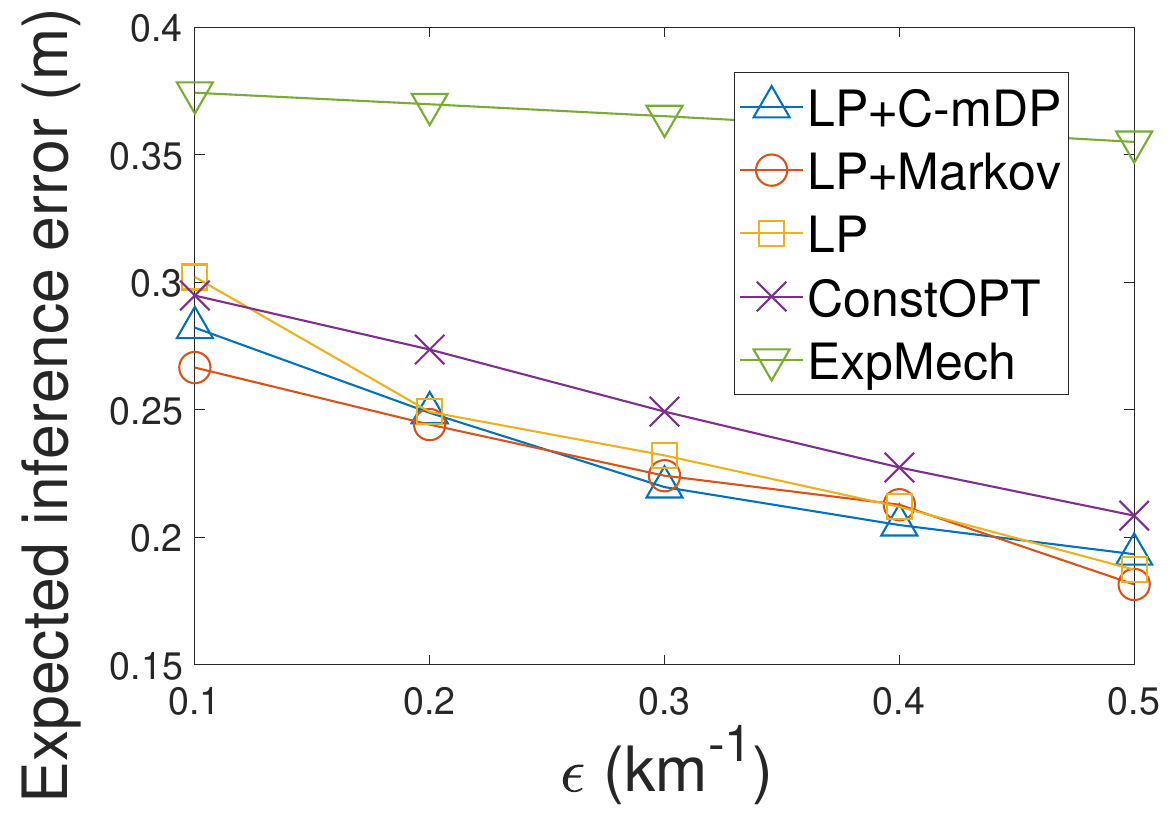}
}
\subfigure[\small Case II: Porto, Portugal]{
\includegraphics[width=0.23\textwidth, height = 0.105\textheight]{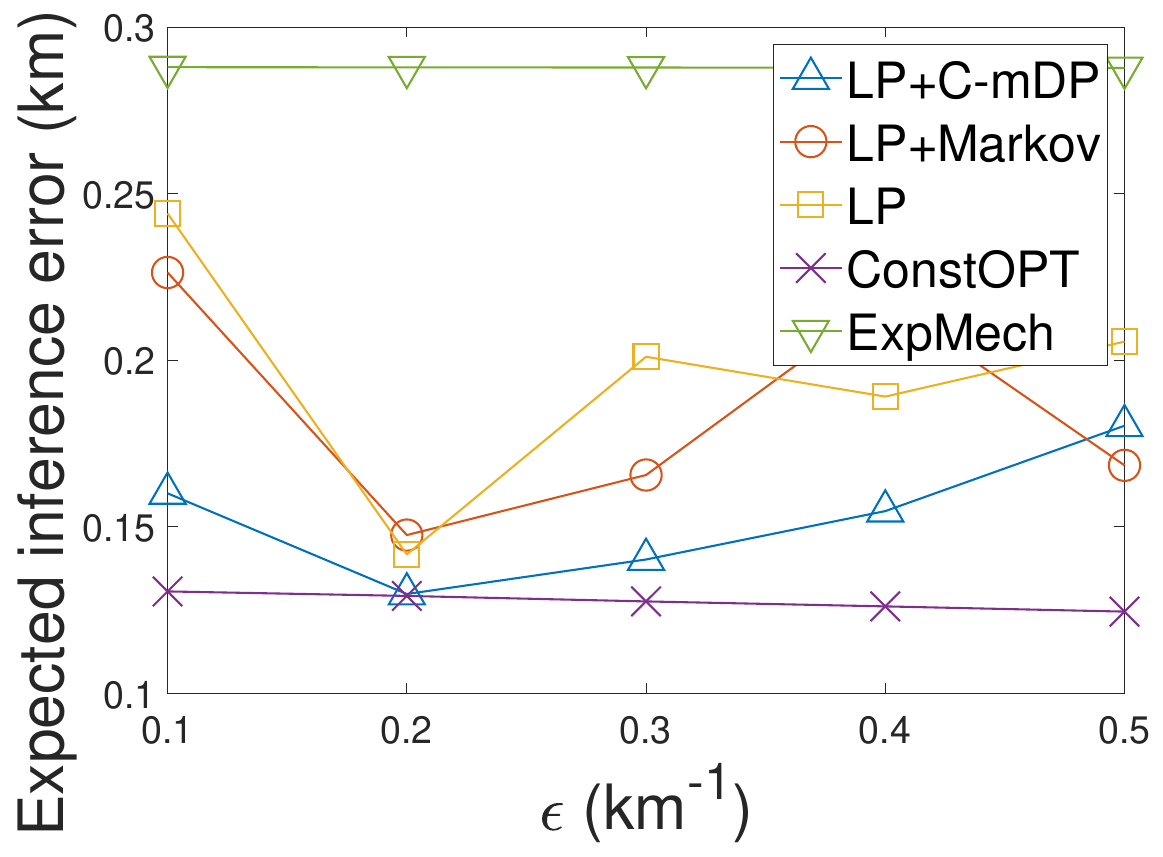}
}
\label{}
\end{minipage}
\vspace{-0.20in}
\caption{Expected inference error vs. privacy budget $\epsilon$.}
\label{fig:EIEepsilon}
\vspace{-0.15in}
\end{figure*}

\begin{figure*}[t]
\centering
\begin{minipage}{1.0\textwidth}
\centering
\subfigure[\small Case I: Rome, Italy]{
\includegraphics[width=0.23\textwidth, height = 0.105\textheight]{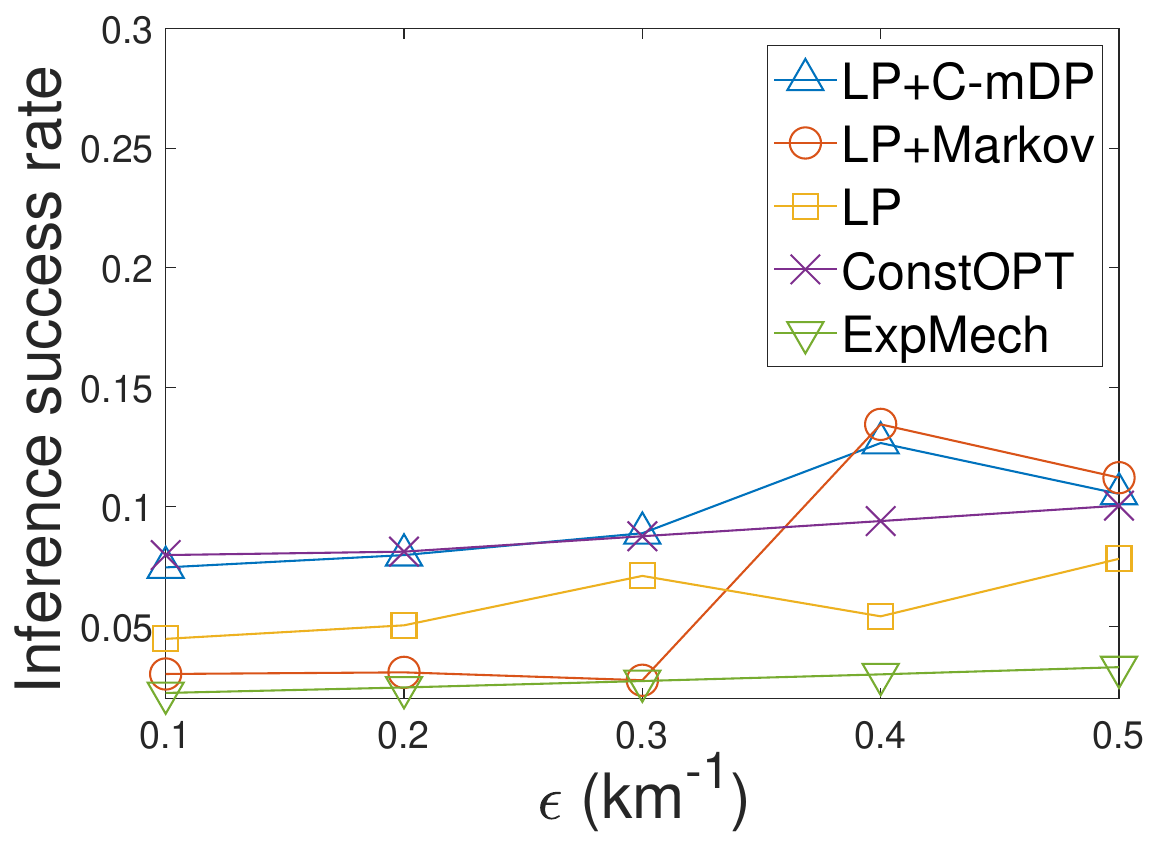}
}
\vspace{-0.00in}
\subfigure[\small Case I: Porto, Portugal]{
\includegraphics[width=0.23\textwidth, height = 0.105\textheight]{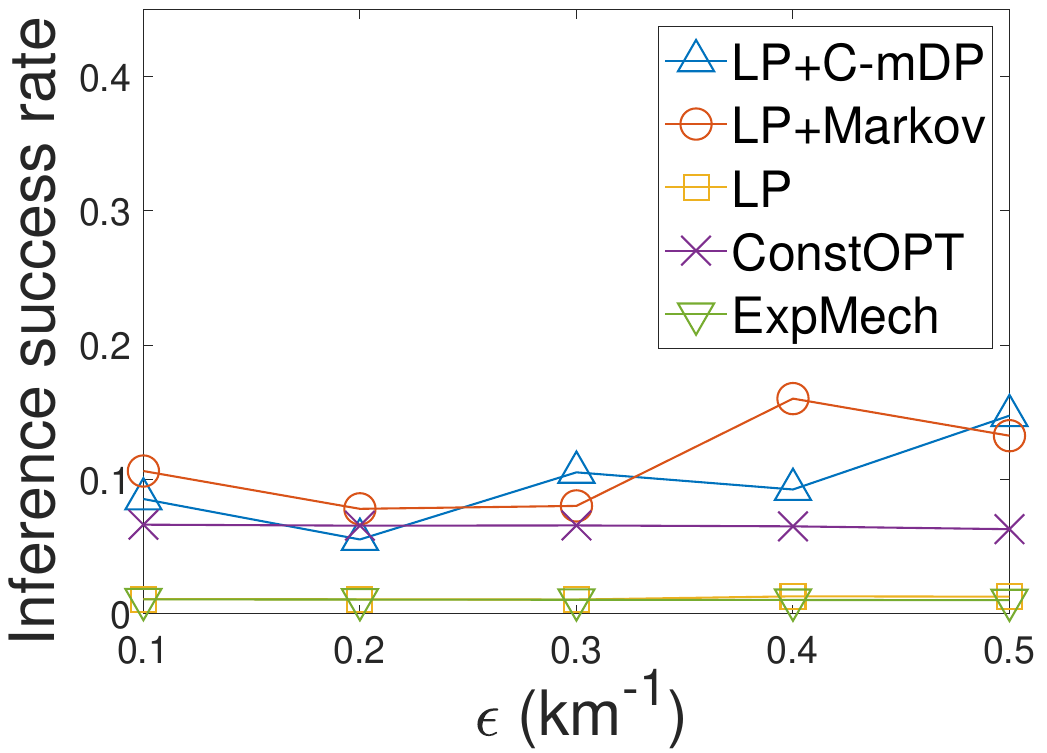}
}
\subfigure[\small Case II: Rome, Italy]{
\includegraphics[width=0.23\textwidth, height = 0.105\textheight]{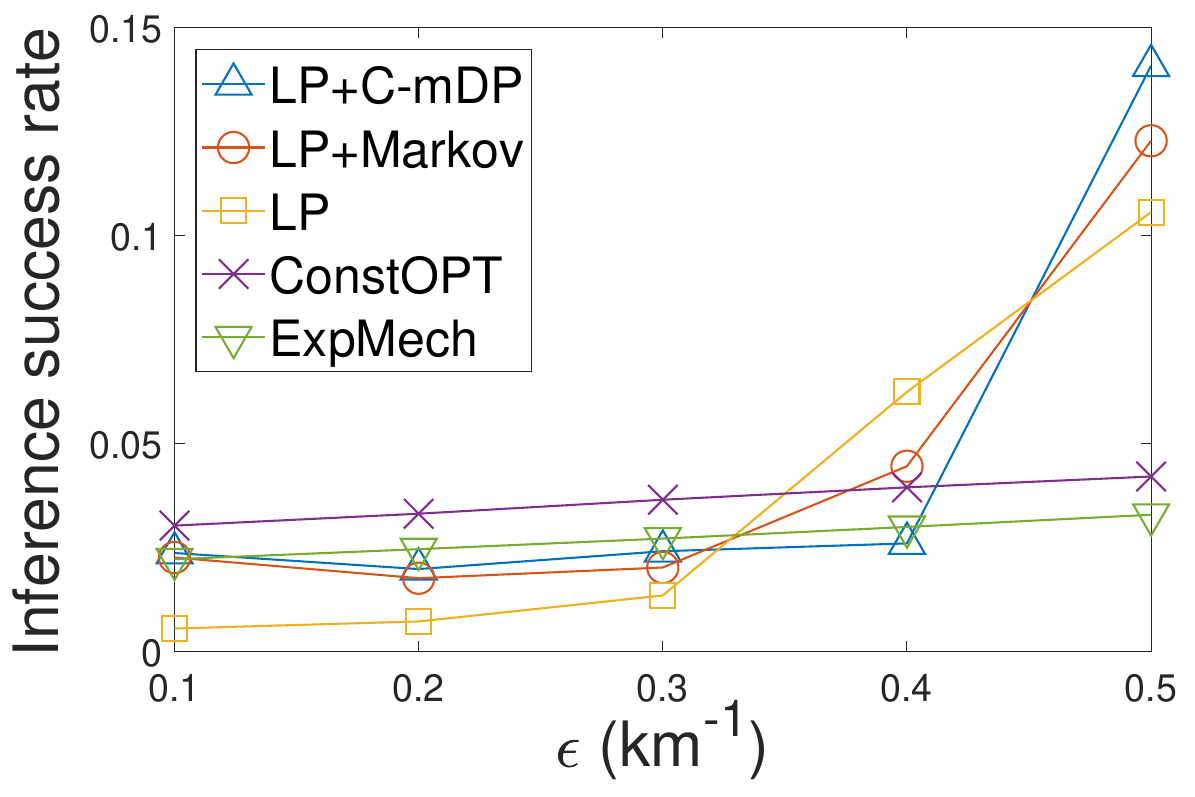}
}
\subfigure[\small Case II: Porto, Portugal]{
\includegraphics[width=0.23\textwidth, height = 0.105\textheight]{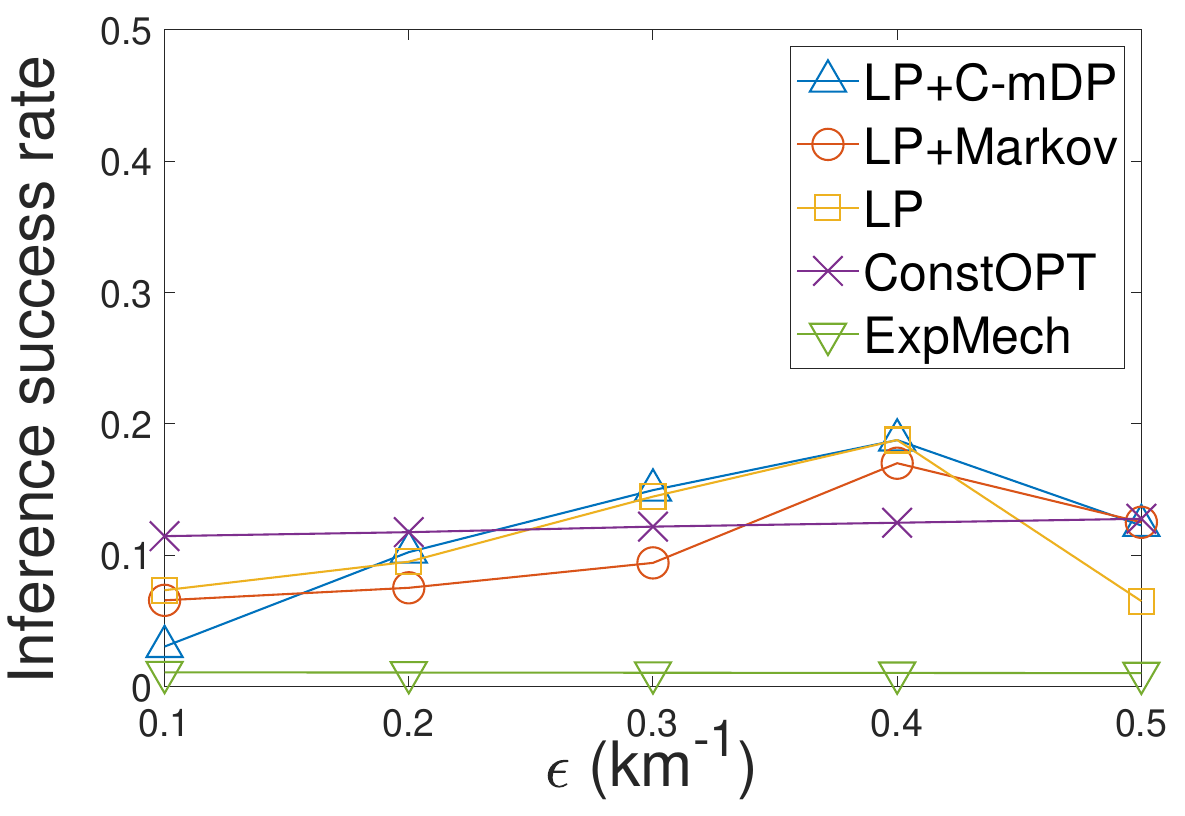}
}
\label{}
\end{minipage}
\vspace{-0.15in}
\caption{Inference success rate vs. privacy budget $\epsilon$.}
\label{fig:ISRepsilon}
\vspace{-0.15in}
\end{figure*}

\begin{figure}[t]
\centering
\begin{minipage}{0.5\textwidth}
\centering
  \subfigure[\small Rome, Italy]{
\includegraphics[width=0.465\textwidth, height = 0.12\textheight]{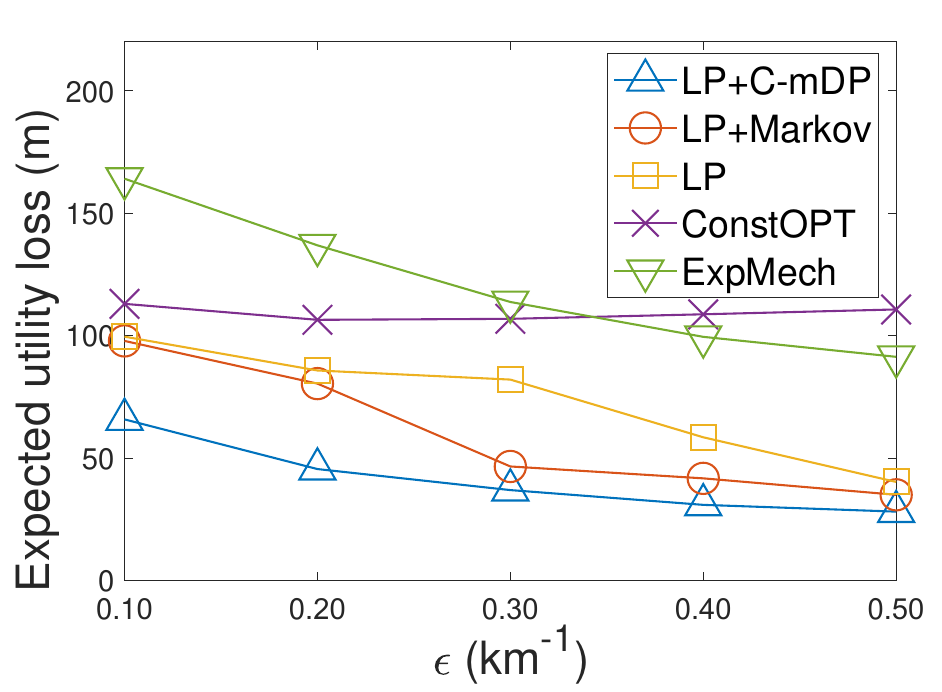}
}
\subfigure[\small Porto, Portugal]{
\includegraphics[width=0.465\textwidth, height = 0.12\textheight]{./fig/ULepsilon_Porto}
}
\end{minipage}
\vspace{-0.15in}
\caption{Expected utility loss vs. privacy budget $\epsilon$.}
\label{fig:ULepsilon}
\centering
\begin{minipage}{0.5\textwidth}
\centering
\subfigure[\small Rome, Italy]{
\includegraphics[width=0.465\textwidth, height = 0.12\textheight]{./fig/PLepsilon1}
}
\vspace{-0.00in}
\subfigure[\small Porto, Portugal]{
\includegraphics[width=0.465\textwidth, height = 0.12\textheight]{./fig/maxPLepsilon_Porto}
}
\end{minipage}
\vspace{-0.15in}
\caption{Posterior leakage vs. privacy budget $\epsilon$.}
\label{fig:PLepsilon}
\centering
\begin{minipage}{0.5\textwidth}
\centering
\subfigure[\small Rome, Italy]{
\includegraphics[width=0.465\textwidth, height = 0.12\textheight]{./fig/PLepsilon}
}
\vspace{-0.00in}
\subfigure[\small Porto, Portugal]{
\includegraphics[width=0.465\textwidth, height = 0.12\textheight]{./fig/meanPLepsilon_Porto}
}
\label{}
\end{minipage}
\vspace{-0.15in}
\caption{Expected posterior leakage vs. privacy budget $\epsilon$.}
\label{fig:exPLepsilon}
\centering
\begin{minipage}{0.5\textwidth}
\centering
\subfigure[\small Rome, Italy]{
\includegraphics[width=0.465\textwidth, height = 0.12\textheight]{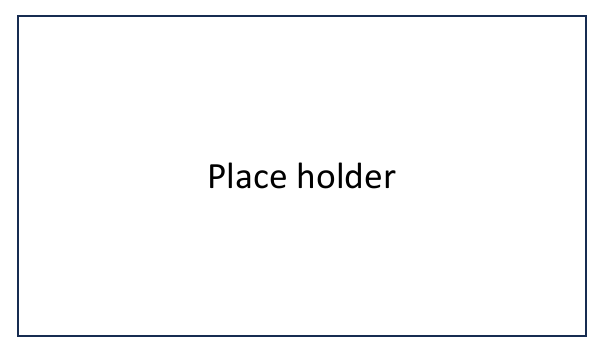}
}
\vspace{-0.00in}
\subfigure[\small Porto, Portugal]{
\includegraphics[width=0.465\textwidth, height = 0.12\textheight]{./fig/placeholder}
}
\label{}
\end{minipage}
\vspace{-0.15in}
\caption{Expected posterior leakage vs. privacy budget $\epsilon$.}
\label{fig:EIE}
\vspace{-0.00in}
\centering
\begin{minipage}{0.5\textwidth}
\centering
\subfigure[\small Rome, Italy]{
\includegraphics[width=0.465\textwidth, height = 0.12\textheight]{./fig/placeholder}
}
\vspace{-0.00in}
\subfigure[\small Porto, Portugal]{
\includegraphics[width=0.465\textwidth, height = 0.12\textheight]{./fig/placeholder}
}
\label{}
\end{minipage}
\vspace{-0.15in}
\caption{Expected posterior leakage vs. privacy budget $\epsilon$.}
\label{fig:EIE}
\vspace{-0.15in}
\end{figure}}

\DEL{
\vspace{-0.10in}
\subsection{Privacy Evaluation} 
\label{subsec:PL}
\vspace{-0.03in}
\subsubsection{Privacy measures} Finally, we measure the privacy levels achieved by different methods using the following two metrics 
\newline {\revision (1) \emph{Posterior leakage (PL)}, defined in Eq. (\ref{eq:PLcontext}), reflecting how much additional information attackers can gain from perturbed data compared to the prior knowledge of the locations. Specifically, for each observed perturbed record $y$ and the context information $\mathbf{b}$, we calculate the posterior of each pair $x_i, x_j \in \mathcal{X}$, i.e., $\Pr\left[X = x_i | Q(X, B_X) = y, {B}_X = \mathbf{b}\right]$ and $\Pr\left[X = x_j | Q(X, B_X) = y, {B}_X = \mathbf{b}\right]$ using the Bayes' formula, based on which we can then derive the PL by:   
\vspace{-0.20in}
\begin{eqnarray}
&& \mathrm{PL}\left((x_i, x_j);Q|\mathbf{b}\right) \\ \nonumber 
&=& {\small  \sup_{y} |\ln \left(  \frac{\scriptsize \Pr\left[X = x_i | Q(X, B_X) = y, {B}_X = \mathbf{b}\right]}{\scriptsize \Pr\left[X = x_j| Q(X, B_X) = y, {B}_X = \mathbf{b}\right]}\right.\left\slash\left.\frac{p_{(x_i,\mathbf{b})}}{p_{x_j|\mathbf{b}}}\right)\right.|}. 
\end{eqnarray}
Here, each prior $p_{x|\mathbf{b}} = P(X=x|B_X = \mathbf{b})$ ($x \in \mathcal{X}$) is the conditional probability of $X = x$ given the previous locations $\mathbf{b}$, which can be directly derived from the datasets.}
\newline (2) \emph{Expected PL}, defined as the expected absolute value of the logarithmic ratio of the posterior ratio and the prior ratio of all the neighboring records, or mathematically represented by  
$$\textstyle \mathbb{E}\left(|\ln \left(\frac{\Pr\left[X = x_i | Q(X) = y\right]}{\Pr\left[X = x_j| Q(X) = y\right]}\right.\left\slash\left.\frac{\Pr\left[X=x_i\right]}{\Pr\left[X=x_j\right]}\right)\right.|\right).$$
As opposed to measuring the maximum posterior leakage like PL, the expected PL evaluates the expected value.   \looseness = -1
{\revision \newline (3) \emph{Expected inference error (EIE)}, also known as the \emph{unconditional expected privacy}  \cite{Qiu-CIKM2020,Shokri-CCS2012}, is defined by
\vspace{-0.03in}
\begin{equation}
\label{eq:EIEsingle}
\small \mathrm{EIE}\left(y_k\right) = \sum_{\hat{x}\in \mathcal{X}} \sum_{y_k\in \mathcal{Y}} P(X=x_i| Y = y_k)h_{\hat{X}|Y = y_k}(\hat{x})d_{\hat{x}, x_i}, 
\vspace{-0.03in}
\end{equation}
where $h_{\hat{X}|Y = y_k}(\hat{x})$ represents the probability that the attacker estimates $\hat{x}$ as the vehicle's location given the vehicle's reported location is $y_k$. EIE represents the expected distortion between the estimated location, $\hat{x}$, and the actual location, $x_i$, with a higher EIE indicating a greater level of privacy achieved. In this context, we assume the attacker employs a \emph{Bayesian inference attack} \cite{Yu-NDSS2017}, which estimates the posterior probability of the true location and selects the location with the highest posterior as the estimated location.  \looseness = -1
\newline (4) \emph{Successful Inference Rate (SIR)} is defined as the probability that the location estimated by a Bayesian inference attack \cite{Yu-NDSS2017} matches the true location, $x_i$. Higher SIR indicates lower privacy level achieved. }

\vspace{-0.10in}
\begin{table}[h]
\caption{Privacy levels of different data perturbation methods (mean$\pm$1.96$\times$standard deviation).}
\vspace{-0.10in}
\label{Tb:exp:privacy}
\centering
\footnotesize 
\begin{tabular}{ c|c|c|c|c}
\hline
\multicolumn{1}{ c  }{Perturbation}& \multicolumn{4}{ c }{Privacy Metrics}\\
\cline{2-5}
\multicolumn{1}{ c|  }{algorithms}
& \multicolumn{1}{ c| }{PL}&\multicolumn{1}{ c |}{Expected PL} & \multicolumn{1}{ c| }{EIE}&\multicolumn{1}{ |c }{ISR} \\
\hline
\multicolumn{5}{ c }{Case I: Rome, Italy} 
\\
\hline
\multicolumn{1}{ c|  }{ \textbf{LP+C-mDP}} & 0.099$\pm$0.015 & 0.062$\pm$0.014 & 0.275$\pm$0.641 & 0.075$\pm$0.457  \\ 
\multicolumn{1}{ c|  }{ LP+Markov} & 0.098$\pm$0.013 & 0.062$\pm$0.013 & 0.257$\pm$0.601 & 0.112$\pm$0.576  \\ 
\multicolumn{1}{ c|  }{ LP} & 0.100$\pm$0.000 & 0.064$\pm$0.009 & 0.252$\pm$0.573 & 0.045$\pm$0.301   \\
\multicolumn{1}{ c|  }{ConstOPT} & 0.100$\pm$0.000 & 0.073$\pm$0.009 & 0.259$\pm$0.633 & 0.119$\pm$0.588  \\ 
\multicolumn{1}{ c|  }{ ExpMech} &  0.074$\pm$0.010 & 0.028$\pm$0.008 & 0.380$\pm$0.921 & 0.022$\pm$0.001   \\ 
\hline
\multicolumn{5}{ c }{Case I: Porto, Portugal} 
\\
\hline
\multicolumn{1}{ c|  }{ \textbf{LP+C-mDP}} & \textbf{0.098$\pm$0.050} & \textbf{0.056$\pm$0.028} & 0.198$\pm$0.372 & 0.171$\pm$0.728  \\ 
\multicolumn{1}{ c|  }{ LP+Markov} & 0.098$\pm$0.055 & 0.054$\pm$0.030 & 0.212$\pm$0.432 & 0.213$\pm$0.782  \\ 
\multicolumn{1}{ c|  }{ LP} & 0.100$\pm$0.000 & 0.051$\pm$0.017  & 0.199$\pm$0.328 & 0.022$\pm$0.212  \\
\multicolumn{1}{ c|  }{ConstOPT} & 0.1285$\pm$0.000 & 0.069$\pm$0.010 & 0.129$\pm$0.231 & 0.266$\pm$0.619  \\ 
\multicolumn{1}{ c|  }{ ExpMech} & 0.072$\pm$0.023 & 0.023$\pm$0.015  & 0.287$\pm$0.288 & 0.022$\pm$0.001  \\ 
\hline
\hline
\multicolumn{5}{ c }{Case II: Rome, Italy} 
\\
\hline
\multicolumn{1}{ c|  }{ \textbf{LP+C-mDP}} & \textbf{0.100$\pm$0.000} & \textbf{0.060$\pm$0.017} & 0.282$\pm$0.921 & 0.024$\pm$0.218  \\ 
\multicolumn{1}{ c|  }{ LP+Markov} & 0.096$\pm$0.039 & 0.057$\pm$0.027 & 0.267$\pm$0.933 & 0.027$\pm$0.170  \\ 
\multicolumn{1}{ c|  }{ LP} & 0.098$\pm$0.028 & 0.059$\pm$0.026 & 0.302$\pm$0.984 & 0.006$\pm$0.084\\ 
\multicolumn{1}{ c|  }{ConstOPT} & 0.100$\pm$0.000 & 0.076$\pm$0.020 & 0.295$\pm$0.761 & 0.042$\pm$0.706  \\  
\multicolumn{1}{ c|  }{ ExpMech} & 0.073$\pm$0.021 & 0.028$\pm$0.015 & 0.374$\pm$0.973 & 0.033$\pm$0.001
   \\  
\hline
\multicolumn{5}{ c }{Case II: Porto, Portugal} 
\\
\hline
\multicolumn{1}{ c|  }{ \textbf{LP+C-mDP}} & \textbf{0.094$\pm$0.047} & \textbf{0.055$\pm$0.031} & 0.193$\pm$0.333  & 0.030$\pm$0.234  \\ 
\multicolumn{1}{ c|  }{ LP+Markov} & 0.086$\pm$0.069 & 0.049$\pm$0.041 &  0.221$\pm$0.436 & 0.066$\pm$0.325  \\ 
\multicolumn{1}{ c|  }{ LP} & 0.090$\pm$0.060 & 0.052$\pm$0.038 & 0.241$\pm$0.357 & 0.073$\pm$0.338  \\ 
\multicolumn{1}{ c|  }{ConstOPT} & 0.100$\pm$0.000 & 0.071$\pm$ 0.008 & 0.135$\pm$0.220 & 0.128$\pm$0.286  \\ 
\multicolumn{1}{ c|  }{ ExpMech} & 0.071$\pm$0.022 & 0.022$\pm$0.014 & 0.287$\pm$0.288 & 0.010$\pm$0.000  \\  
\hline
\end{tabular}
\vspace{-0.18in}
\end{table}

\vspace{-0.03in}
\subsubsection{Experimental results} Table \ref{Tb:exp:privacy} shows the PL, the expected PL, EIE, and ISR of the five methods when the privacy budget $\epsilon = 0.1$km$^{-1}$. Fig. \ref{fig:PLepsilon}(a)(b)(c)(d) and Fig. \ref{fig:exPLepsilon}(a)(b)(c)(d) show the PL and the expected PL of these methods given different values of the privacy budget $\epsilon$. As $\epsilon$ increases from 0.1 km$^{-1}$ to 0.5 km$^{-1}$. We observe a corresponding increase in both PL and expected PL due to the relaxation of mDP constraints. {\revision  Fig. \ref{fig:EIEepsilon}(a)-(d) and Fig. \ref{fig:ISRepsilon}(a)-(d) present the EIE and ISR metrics for the five methods under varying values of the privacy budget, $\epsilon$.} 

From Table \ref{Tb:exp:privacy} and Fig. \ref{fig:PLepsilon}(a)(b)(c)(d), we observe that the PL of all methods remains within the privacy budget, $\epsilon$. This finding aligns with the theoretical result in \textbf{Proposition \ref{prop:ContextDP}} (for C-mDP) and prior work \cite{Koufogiannis2015GradualRO} (for mDP), both of which confirm that satisfying $\epsilon$-mDP is sufficient to ensure that the PL does not exceed $\epsilon$." 

{\revision According to Table \ref{Tb:exp:privacy}, Fig. \ref{fig:PLepsilon}(a)-(d), Fig. \ref{fig:exPLepsilon}(a)-(d), Fig. \ref{fig:EIEepsilon}(a)-(d), and Fig. \ref{fig:ISRepsilon}(a)-(d), ExpMech exhibits a lower PL, a lower expected PL, a higher EIE, and a lower ISR compared to the other four methods, indicating a higher level of privacy protection based on these metrics. However, this increased privacy comes at a significant cost to utility, as demonstrated in Table \ref{Tb:exp:QL} and Fig. \ref{fig:ULepsilon}(a)-(d). Importantly, within the mDP framework, data is generally considered adequately protected as long as the PL remains within a predetermined budget. All five methods meet this requirement, as shown in Table \ref{Tb:exp:privacy} and Fig. \ref{fig:PLepsilon}(a)-(d). Beyond this point, further maximizing indistinguishability is often unnecessary, with the primary focus shifting toward minimizing utility loss.

Additionally, we find that EIE decreases while ISR increases as $\epsilon$ becomes larger across all figures. This is because a higher $\epsilon$ corresponds to less deviation between obfuscated and real locations. Consequently, this results in a smaller expected distortion between the estimated location (as inferred by attackers) and the real location, and increases the likelihood of the estimated location matching the real location.}}

\vspace{-0.00in}
\subsection{Computation Efficiency Evaluation}
\label{subsec:efficiency}
\vspace{-0.00in}
Finally, we evaluate the computational efficiency of perturbation matrix optimization and perturbed record selection. 
We use the MATLAB optimization toolbox \texttt{linearprog} to solve the LP problem formulated in 
Eq.~(\ref{eq:C-mDPObjred})--(\ref{eq:CLPconstraint1red}). 

Fig.~\ref{fig:LPtime}(a) reports the computation time (in seconds) for solving the LP with \texttt{linearprog} as a function of the privacy budget $\epsilon$ (km$^{-1}$) on the Rome and Porto datasets. Overall, the LP solving time remains low and stable, ranging from $\approx 8.07$\,s to $\approx 11.53$\,s across all settings. Specifically, for Rome, the runtime across different privacy budget is approximately 8-12 seconds, with an average of $\approx 10.24$ seconds.
For Porto, the corresponding runtime is approximately 8-10 seconds, with an average of $\approx 9.34$ seconds.
These results indicate that the LP-based perturbation matrix optimization is computationally efficient (on the order of $\sim$10 seconds per solve) and does not exhibit a sharp runtime increase as $\epsilon$ varies.

Overall, the figure demonstrates that solving the perturbation matrix optimization problem via \texttt{linearprog} is computationally efficient and practically feasible, as the runtime remains stable and lightweight across all evaluated settings.

Note that, both perturbation matrix and utility loss can be precomputed, indicating such computation time is acceptable. Given the perturbation perturbation matrix, we then measure the computation time for perturbed record selection and depict the results in Fig. \ref{fig:LPtime}(b). From the figure, we observe that the computation time doesn't exceed  1.10 milliseconds, while the average computation time is 1.01 milliseconds, which is also acceptable for vehicle-based spatial crowdsourcing. 


\vspace{-0.00in}
\section{Related Works}
\vspace{-0.00in}
\label{sec:related}
\textbf{mDP}. mDP was originally introduced in the domain of location privacy, requiring ``geo-indistinguishability'' for each pair of nearby locations \cite{Andres-CCS2013}. In this context, ``neighboring records'' are defined based on their Euclidean distance, diverging from the original DP classification, which relies on the Hamming distance of databases. Currently, mDP has been explored using various metric choices, including Haversine distance and Euclidean distance in geo-location data \cite{Pappachan-EDBT2023}, as well as Hyperbolic distance \cite{Feyisetan-ICDM2019} and Levenshtein distance \cite{Jones-ACC2018} in text data. Considering 
the varied sensitivity of utility loss to perturbation in general distance metric spaces, a widely used paradigm of finding the optimal mDP is to formulate it as an LP problem \cite{Bordenabe-CCS2014}, which, however, may suffer from the polynomial explosion of variables and constraints \cite{ImolaUAI2022}. Recent efforts such as \cite{Qiu-CIKM2020, Qiu-TMC2020, Qiu-IJCAI2024, Liu-CCS2025} have enhanced the computational efficiency of LP-based methods by employing optimization decomposition. 
Another line of work applies exponential noise \cite{Carvalho2021TEMHU} or combines the predefined noise with LP \cite{ImolaUAI2022, qiu-interpolation2026}, which achieves higher computation efficiency but at the expense of compromising data utility. 

\vspace{0.02in}
\noindent \textbf{Context-aware data perturbation}. Many research works have extended DP to context-aware data perturbation. A notable such framework is called \emph{Pufferfish privacy} \cite{Kifer-PODS2012}, which allows data protectors to define privacy for their own data-sharing demands by considering datasets' background information. DP actually can be viewed as a specific instance of Pufferfish privacy when the correlation between secret records is not considered. Inspired by Pufferfish, He et al. \cite{He-SIGMOD2014} proposed the \emph{Blowfish privacy} framework by allowing users to specify the secret records to protect. Following these frameworks, several other studies have focused on describing the dataset background information using various statistical models such as parameter-based descriptions of correlation \cite{Chen-VLDB2014, Zhu-TIFS2015, Liu-NDSS2016},  Gaussian correlation models \cite{Yang-SIGMOD2015, Chen-TBD2021}, and Markov chain models \cite{Qiu-SIGSPATIAL2022}.

These related works differ from ours in two key aspects. Firstly, they primarily examine how context information (data correlation) affects DP, whereas \emph{our focus is on how context information influences utility loss caused by data perturbation}. Secondly, they rely on the Hamming distance to define neighboring databases, as is typical in traditional DP, whereas \emph{we consider mDP, where the distance between secret records is defined within a general metric space}. This extension to a metric space in mDP introduces unique challenges, as it requires more precise noise control due to the finer-grained indistinguishability between secret data and the varying utility loss caused by noise in different directions and magnitudes.

\vspace{0.02in}
\noindent \textbf{Spatiotemporal location privacy}. Some recent studies address context-aware location privacy by analyzing the spatiotemporal correlations in mobile users' reported locations, focusing either on data from a single user across multiple time points (e.g., trajectory data) \cite{LIAO2007311, Xu-WWW2017, Cao-ICDE2017, Emrich-ICDE2012, Li-Sigspatial2008, Arain-MTA2018} or on data from multiple users \cite{Cao-ICDE2019, Li-SigSpatial2017}. Many of these works assume that users' movements follow a first-order Markov process \cite{LIAO2007311, Emrich-ICDE2012}, where each user's current location depends on the previous location. For example, Liao et al. \cite{LIAO2007311} applied a hierarchical Markov model to predict a user’s trajectory based on their visited places and observed temporal patterns. 

Our framework is distinct from these works as we discard the Markovian assumption, proposing a more general model that captures correlations in protected data. Our empirical study further reveals that vehicle mobility patterns may not strictly adhere to a first-order Markov process, with Markov blankets for locations showing strong dependence on environmental factors such as regional differences and time.  

\vspace{-0.05in}
\section{Discussion and Conclusion}
\vspace{-0.00in}
\label{sec:conclude}
In this paper, we proposed a new mDP framework, C-mDP, by considering the impact of context information on data perturbation. Considering the high computation load of C-mDP when solving it as an LP, we reduce its complexity by identifying the CI relationship between secret records and context variables. To illustrate practical applicability, we conduct a case study by applying C-mDP to protect vehicles' location privacy. The experimental results demonstrate the superiority of C-mDP over the existing mDP methods. 



We identify several promising directions to further improve the proposed C-mDP framework. First, our current perturbation optimization assumes data utility depends on the distribution of secret data (Assumption \ref{assu:datautility}). To address potential correlations in downstream data processing, we plan to remove this assumption and further explore CI relationships among secret records. Section \ref{sec:discussion} of the Appendix outlines how to adapt the C-mDP framework when this assumption is relaxed. Second, we aim to develop context-aware threat models where attackers use context data to narrow the search space of secret data, improving inference accuracy. As a countermeasure, we will constrain perturbed data within a range (e.g., identified by deep generative models) where it becomes difficult for attackers using context information to distinguish perturbed data from real data.

{\revisiondone Finally, our current evaluation uses public, de-identified trajectory data to estimate population-level mobility statistics (e.g., priors/transition tendencies and Markov blanket structure) and to compute the offline perturbation matrix $\mathbf{Q}$. In practical deployments, however, such statistics (or auxiliary predictive models used to estimate them) may be learned from sensitive, non-public user data. An important direction for future work is therefore to \emph{privatize the offline mechanism-design pipeline} itself, for example, by learning the required priors/transition statistics (or training the corresponding neural models) under differential privacy (e.g., by releasing DP-sanitized aggregates), and then composing that privacy loss with the output-privacy budget that governs the released perturbed outputs. This extension would provide end-to-end privacy protection covering both offline model/parameter learning and online per-user perturbation.}

\section{Acknowledgements}
The authors used ChatGPT 5.2 to revise the text in Sections \ref{sec:intro}--\ref{sec:conclude} to correct typos, grammatical errors, and awkward phrasing. 



\DEL{
\vspace{-0.00in}
\begin{table}[t]
\caption{Performance of different data perturbation methods.}
\vspace{-0.10in}
\label{Tb:exp:QL}
\centering
\footnotesize 
\begin{tabular}{ c|c|c|c}
\hline
\hline
\multicolumn{1}{ c  }{Perturbation}& \multicolumn{3}{ c }{Metrics} \\
\cline{2-4}
\multicolumn{1}{ c|  }{algorithms}
 & \multicolumn{1}{ |c| }{Inf.  error}&\multicolumn{1}{ |c| }{Expected PL }
 & \multicolumn{1}{ |c }{Utility loss}
 \\ 
\hline
\hline
\multicolumn{1}{ c|  }{ \textbf{LP+C-mDP}} & \textbf{0.206$\pm$0.478} & \textbf{0.062$\pm$0.014} & \textbf{0.0065$\pm$0.016} \\ 
\multicolumn{1}{ c|  }{ LP+Markov} & 0.212$\pm$0.511 & x.xxx$\pm$x.xxx & 0.0071$\pm$0.020 \\ 
\multicolumn{1}{ c|  }{ LP} & 0.298$\pm$0.651 & x.xxx$\pm$x.xxx & 0.0135$\pm$0.036  \\
\multicolumn{1}{ c|  }{ Exp} & 0.394$\pm$0.634 & x.xxx$\pm$x.xxx & 0.016$\pm$0.035  \\ 
\multicolumn{1}{ c|  }{ TEM} & x.xxx$\pm$x.xxx & x.xxx$\pm$x.xxx & x.xxx$\pm$x.xxx  \\ 
\multicolumn{1}{ c|  }{ ExpMech} & x.xxx$\pm$x.xxx & x.xxx$\pm$x.xxx & x.xxx$\pm$x.xxx  \\ 
\hline
\end{tabular}
\vspace{-0.1in}
\end{table}}


\appendix

\vspace{-0.00in}
\section*{Appendix}

\section{Math Notations}
\label{app:sec:notations}
\vspace{-0.00in}
\begin{table}[h]
\caption{Main notations and their descriptions.}
\vspace{-0.10in}
\label{Tb:Notationmodel}
\centering
\small 
\begin{tabular}{l l}
\hline
\hline
Symbol                  & Description \\
\hline
$Q$                  & Perturbation function \\
$\mathcal{X}$                  & Secret data set (domain of the function $Q$)\\
$\mathcal{Y}$                  & Perturbed data set (range of the function $Q$) \\
$\mathcal{V}$                  & Contextual information space \\
$X$                  & Random variable of secret data  \\
${V}_{X}$                  & Context variables of secret data $X$ \\
${B}_{X}$                  & Markov blanket of secret data $X$ \\
$X_t$                  & (Case study) Random variable to represent the target \\ 
& vehicle's location at time slot $t$  \\
$x_i$                  & Secret record $i$ \\
$y_k$                  & Perturbed record $k$ \\
$\epsilon$                  & Privacy budget \\
$d_{x_i, x_j}$                  & Distance between secret records $x_i$ and $x_j$ \\
$\mathbf{Q}$ & Perturbation matrix \\
$\mathcal{L}\left(\mathbf{Q}\right)$                 & Loss function of the perturbation matrix $\mathbf{Q}$ \\
$q_{x_i,y_k}$                  & Probability of selecting $y_k$ as the perturbed data given the  \\ 
&  real record $x_i$ \\
$q_{(x_i,\mathbf{v}),y_k}$                  & Probability of selecting $y_k$ as the perturbed data given the \\ 
&   real record $x_i$ and the context data $\mathbf{v}$\\
$q_{(x_i,\mathbf{b}),y_k}$                  & Probability of selecting $y_k$ as the perturbed data given the \\ 
&    real record $x_i$ and the Markov blanket $\mathbf{b}$\\
$c_{x_i,y_k}$                  &  Data utility loss caused by the perturbed record $y_k$ when \\
&  the real record is $x_i$ \\
\hline
\end{tabular}
\normalsize
\end{table}
\vspace{-0.00in}

\section{Omitted Proofs}
\label{sec:proofs}
\subsection{Proof of Proposition \ref{prop:ContextDP}}
\begin{reproposition}
To satisfy the context-aware PL bound in Eq. (\ref{eq:PLcontextbound}), it is sufficient to enforce the C-mDP constraints
\vspace{-0.05in}
\begin{equation}
q_{(x_i,\mathbf{v}),y_k} - e^{\epsilon d_{(x_i,\mathbf{v}),(x_j,\mathbf{v}')}} \cdot q_{(x_j,\mathbf{v}'),y_k} \leq 0, ~\forall y_k \in \mathcal{Y}, 
\vspace{-0.00in}
\end{equation}
for each neighboring records  $x_i, x_j \in \mathcal{X}$ and all their possible context data $\mathbf{v}, \mathbf{v}' \in \mathcal{V}$. 
\end{reproposition}

\begin{proof}
Let $Y$ denote the released (perturbed) output, i.e., $Y = Q(X,V_X)$.
For any $(x_i,\mathbf{v}) \in \mathcal{X}\times\mathcal{V}$ and any $y\in\mathcal{Y}$, by Bayes' rule,
\begin{eqnarray}
\label{eq:posterior-joint}
&& \Pr\left[(X,V)=(x_i,\mathbf{v}) \mid Y=y\right] \\
&=&
\frac{\Pr\left[(X,V)=(x_i,\mathbf{v})\right]\Pr\left[Y=y \mid (X,V)=(x_i,\mathbf{v})\right]}{\Pr[Y=y]} \\
&=& 
\frac{p_{(x_i,\mathbf{v})} \cdot q_{(x_i,\mathbf{v}),y}}{\Pr[Y=y]}.
\end{eqnarray}
Hence, for any $(x_i,\mathbf{v})$ and $(x_j,\mathbf{v}')$ with $p_{(x_j,\mathbf{v}')}>0$ and any $y\in\mathcal{Y}$ such that
$q_{(x_j,\mathbf{v}'),y}>0$, we have
\begin{eqnarray}
\frac{\Pr\left[(X,V)=(x_i,\mathbf{v}) \mid Y=y\right]}{\Pr\left[(X,V)=(x_j,\mathbf{v}') \mid Y=y\right]}
=
\frac{p(x_i,\mathbf{v}) \cdot q_{(x_i,\mathbf{v}),y}}{p(x_j,\mathbf{v}') \cdot q_{(x_j,\mathbf{v}'),y}},
\label{eq:post-ratio}
\end{eqnarray}
and hence 
\begin{eqnarray}
&& \ln\left(
\frac{\Pr\left[(X,V)=(x_i,\mathbf{v}) \mid Y=y\right]}{\Pr\left[(X,V)=(x_j,\mathbf{v}') \mid Y=y\right]}
\right)
-\ln\left(\frac{p(x_i,\mathbf{v})}{p(x_j,\mathbf{v}')}\right)
\\
&=&
\ln\left(\frac{q_{(x_i,\mathbf{v}),y}}{q_{(x_j,\mathbf{v}'),y}}\right).
\label{eq:log-diff}
\end{eqnarray}

Now suppose the C-mDP constraints in Eq.~(\ref{eq:C-mDP}) hold, then whenever $q_{(x_j,\mathbf{v}'),y}>0$, Eq.~(\ref{eq:C-mDP}) implies
\begin{equation}
\ln\left(\frac{q_{(x_i,\mathbf{v}),y}}{q_{(x_j,\mathbf{v}'),y}}\right)
\le
\epsilon d_{(x_i,\mathbf{v}),(x_j,\mathbf{v}')}.
\label{eq:log-upper}
\end{equation}
Moreover, applying Eq.~(\ref{eq:C-mDP}) again but with the pair swapped gives
\begin{eqnarray}
&& q_{(x_j,\mathbf{v}'),y} \le e^{\epsilon d_{(x_i,\mathbf{v}),(x_j,\mathbf{v}')}}q_{(x_i,\mathbf{v}),y}
\\
&\Rightarrow&
\ln\left(\frac{q_{(x_i,\mathbf{v}),y}}{q_{(x_j,\mathbf{v}'),y}}\right)
\ge
-\epsilon d_{(x_i,\mathbf{v}),(x_j,\mathbf{v}')}.
\label{eq:log-lower}
\end{eqnarray}
Combining Eq.~(\ref{eq:log-upper}) and Eq.~(\ref{eq:log-lower}) yields, for all $y$ with
$q_{(x_i,\mathbf{v}),y}>0$ and $q_{(x_j,\mathbf{v}'),y}>0$,
\begin{equation}
\left|
\ln\left(\frac{q_{(x_i,\mathbf{v}),y}}{q_{(x_j,\mathbf{v}'),y}}\right)
\right|
\le
\epsilon d_{(x_i,\mathbf{v}),(x_j,\mathbf{v}')}.
\label{eq:q-bound}
\end{equation}
(If $q_{(x_j,\mathbf{v}'),y}=0$, then Eq.~(\ref{eq:C-mDP}) forces $q_{(x_i,\mathbf{v}),y}=0$ as well, so such $y$ does not affect the supremum.)

Finally, substituting Eq.~(\ref{eq:q-bound}) into Eq.~(\ref{eq:log-diff}) and taking the supremum over $y\in\mathcal{Y}$ gives
\[
\mathrm{PL}\big((x_i,\mathbf{v}),(x_j,\mathbf{v}')\big)
\le
\epsilon d_{(x_i,\mathbf{v}),(x_j,\mathbf{v}')},
\]
which is exactly the context-aware PL bound in Eq.~(\ref{eq:PLcontextbound}).
\end{proof}

\DEL{
\begin{proof}
Given the observation of $Q(X, V_X) = y$ and the context data $V_X = \mathbf{v}$, the posterior ratio of $X = x_i$ and $X = x_j$ is calculated by
\normalsize
\begin{eqnarray}
&& \frac{\Pr\left[X = x_i | Q(X, V_X) = y, {V}_X = \mathbf{v}\right]}{\Pr\left[X = x_j| Q(X, V_X) = y, {V}_X = \mathbf{v}\right]} 
\\ \nonumber 
&=& \frac{\Pr\left[Q(X, V_X) = y|X = x_i, {V}_X = \mathbf{v}\right]\Pr\left[X = x_i | {V}_X = \mathbf{v}\right]}{\Pr\left[Q(X, V_X) = y|X = x_j, {V}_X = \mathbf{v}\right]\Pr\left[X = x_j|{V}_X = \mathbf{v}\right]} \\ 
&\leq & e^{\epsilon d_{(x_i, \mathbf{v}), (x_j,\mathbf{v}')}} \frac{\Pr\left[X=x_i| {V}_X = \mathbf{v}\right]}{\Pr\left[X=x_j|{V}_X = \mathbf{v}\right]} \\
&=& e^{\epsilon d_{(x_i, \mathbf{v}), (x_j,\mathbf{v}')}} \frac{p_{x_i|\mathbf{v}}}{p_{x_j|\mathbf{v}}}
\end{eqnarray}
\normalsize
implying that $\forall x_i, x_j \in \mathcal{X}$, $\mathrm{PL}\left[(x_i, \mathbf{v}), (x_j, \mathbf{v}')\right] \leq \epsilon d_{(x_i, \mathbf{v}), (x_j,\mathbf{v}')}$. The proof is completed. 
\end{proof}}

\subsection{Proof of Proposition \ref{prop:CDcorrect}}
\label{subsec:proofPLorder}
\begin{reproposition}
(PL guarantee)
If a perturbation matrix $\mathbf{Q}$ follows the CD policy and satisfies the corresponding mDP constraints:  
\vspace{-0.07in}
\begin{equation}
q_{(x_i,\mathbf{b}),y_k} - e^{\epsilon d_{(x_i,\mathbf{b}), (x_j,\mathbf{b}')}} \cdot q_{(x_j,\mathbf{b}'),y_k} \leq 0, ~\forall x_i, y_k, \mathbf{b}, \mathbf{b}' 
\vspace{-0.00in}
\end{equation}
where $d_{(x_i,\mathbf{b}), (x_j,\mathbf{b}')} \leq d_{(x_i, \mathbf{v}), (x_j,\mathbf{v}')}$, 
it is \textbf{sufficient} for $\mathbf{Q}$ to achieve the bounded context-aware PL as defined in Eq. (\ref{eq:PLcontextbound}). 
\end{reproposition}
\begin{proof}
Let $Y$ denote the released output, i.e., $Y = Q(X,V_X)$.
For any $(x,\mathbf{v})\in\mathcal{X}\times\mathcal{V}$ and any $y\in\mathcal{Y}$, Bayes' rule gives
\begin{equation}
\Pr \left[(X,V)=(x,\mathbf{v}) \mid Y=y\right]
=
\frac{p_{(x,\mathbf{v})}\, \cdot q_{(x,\mathbf{v}),y}}{\Pr[Y=y]}.
\label{eq:bayes-joint-red}
\end{equation}
Hence, for any $(x_i,\mathbf{v})$ and $(x_j,\mathbf{v}')$ with $p(x_j,\mathbf{v}')>0$ and any $y$ with
$q_{(x_j,\mathbf{v}'),y}>0$,
\begin{eqnarray}
&& \ln \left(
\frac{\Pr[(X,V)=(x_i,\mathbf{v})\mid Y=y]}{\Pr[(X,V)=(x_j,\mathbf{v}')\mid Y=y]}
\right)
-\ln \left(\frac{p(x_i,\mathbf{v})}{p(x_j,\mathbf{v}')}\right)
\\
&=&
\ln \left(\frac{q_{(x_i,\mathbf{v}),y}}{q_{(x_j,\mathbf{v}'),y}}\right).
\label{eq:logdiff-red}
\end{eqnarray}

Now assume $\mathbf{Q}$ follows the CD policy. Let $\mathbf{b}$ (resp., $\mathbf{b}'$) denote the Markov blanket
value associated with $\mathbf{v}$ (resp., $\mathbf{v}'$), i.e., $\mathbf{b}=\pi(\mathbf{v})$ and
$\mathbf{b}'=\pi(\mathbf{v}')$ for the corresponding projection $\pi:\mathcal{V}\to\mathcal{B}$.
By the CD policy, conditioning on $(X,\mathbf{b})$ makes the output independent of the remaining context, hence
\begin{equation}
q_{(x,\mathbf{v}),y} = q_{(x,\mathbf{b}),y},\quad \forall x\in\mathcal{X},\ \forall \mathbf{v}\in\mathcal{V},\ \forall y\in\mathcal{Y},
\label{eq:cd-eq}
\end{equation}
where $\mathbf{b}=\pi(\mathbf{v})$.

Since $\mathbf{Q}$ satisfies the reduced mDP constraints in Eq.~(\ref{eq:C-mDPred}), for all $y\in\mathcal{Y}$,
\begin{equation}
q_{(x_i,\mathbf{b}),y}
\le
\exp \big(\epsilon d_{(x_i,\mathbf{b}),(x_j,\mathbf{b}')}\big)\, \cdot q_{(x_j,\mathbf{b}'),y}.
\label{eq:red-upper}
\end{equation}
Moreover, applying Eq.~(\ref{eq:red-upper}) to the swapped pair $(x_j,\mathbf{b}')$ and $(x_i,\mathbf{b})$ yields
\begin{equation}
q_{(x_j,\mathbf{b}'),y}
\le
\exp \big(\epsilon d_{(x_i,\mathbf{b}),(x_j,\mathbf{b}')}\big)\, \cdot q_{(x_i,\mathbf{b}),y},
\label{eq:red-lower}
\end{equation}
and thus, for all $y$ with $q_{(x_i,\mathbf{b}),y}>0$ and $q_{(x_j,\mathbf{b}'),y}>0$,
\begin{equation}
\left|
\ln \left(\frac{q_{(x_i,\mathbf{b}),y}}{q_{(x_j,\mathbf{b}'),y}}\right)
\right|
\le
\epsilon d_{(x_i,\mathbf{b}),(x_j,\mathbf{b}')}.
\label{eq:red-logbound}
\end{equation}
(If $q_{(x_j,\mathbf{b}'),y}=0$, Eq.~(\ref{eq:red-upper}) forces $q_{(x_i,\mathbf{b}),y}=0$ as well, so such $y$
does not affect the supremum in the PL definition.)

Combining Eq.~(\ref{eq:cd-eq}) and Eq.~(\ref{eq:red-logbound}) gives
\begin{equation}
\left|
\ln \left(\frac{q_{(x_i,\mathbf{v}),y}}{q_{(x_j,\mathbf{v}'),y}}\right)
\right|
=
\left|
\ln \left(\frac{q_{(x_i,\mathbf{b}),y}}{q_{(x_j,\mathbf{b}'),y}}\right)
\right|
\le
\epsilon d_{(x_i,\mathbf{b}),(x_j,\mathbf{b}')}.
\label{eq:logbound-v}
\end{equation}
By the assumption $d_{(x_i,\mathbf{b}),(x_j,\mathbf{b}')} \le d_{(x_i,\mathbf{v}),(x_j,\mathbf{v}')}$, we further have
\begin{equation}
\left|
\ln \left(\frac{q_{(x_i,\mathbf{v}),y}}{q_{(x_j,\mathbf{v}'),y}}\right)
\right|
\le
\epsilon d_{(x_i,\mathbf{v}),(x_j,\mathbf{v}')}.
\label{eq:logbound-v2}
\end{equation}

Finally, substituting Eq.~(\ref{eq:logbound-v2}) into Eq.~(\ref{eq:logdiff-red}) and taking the supremum over
$y\in\mathcal{Y}$ yields the bounded context-aware posterior leakage in Eq.~(\ref{eq:PLcontextbound}).
Therefore, enforcing the reduced constraints in Eq.~(\ref{eq:C-mDPred}) under the CD policy is sufficient for
$\mathbf{Q}$ to achieve the context-aware PL bound.
\end{proof}

\DEL{
First, we introduce Lemma \ref{lem:obfindependent} as a preparation of the proof: 
\begin{lemma} 
\label{lem:obfindependent}
If $X  \perp \!\!\!\perp {B}_{X}^\mathrm{c} |{B}_{X}$, then $Q(X)  \perp \!\!\!\perp {B}_{X}^\mathrm{c} |{B}_{X}$. 
\end{lemma}}
\DEL{
\begin{proof} 
We first derive the posterior of $X = x$ given the observation of $Q(X, V_X) = y$ and the context data $V_X = \mathbf{v}$: 
\normalsize
\begin{eqnarray}
&& \Pr\left[X = x| Q(X, V_X) = y, {V}_X = \mathbf{v}\right] \\
&=& \frac{\Pr\left[X = x, {V}_X = \mathbf{v}\right]\Pr\left[Q(X, V_X) = y| X = x, {V}_X= \mathbf{v}\right]}{\Pr\left[Q(X, V_X) = y, {V}_X= \mathbf{v}\right]}\\  \nonumber
 &=& \frac{\Pr\left[X = x, {V}_X= \mathbf{v}\right]\overbrace{\Pr\left[Q(X, V_X) = y| X = x, {B}_{X}= \mathbf{b}\right]}^{\mbox{CD policy (Eq. (\ref{eq:CDpolicy}))}}}{\Pr\left[Q(X, V_X) = y, {V}_X= \mathbf{v}\right]}
\end{eqnarray}
\normalsize
Then, $\forall x_i, x_j \in \mathcal{X}$, we can find the posterior ratio of $X = x_i$ and $X = x_j$ given the observation of $Q(X, V_X) = y$ and the context data $V_X = \mathbf{v}$ is upper bounded by
\normalsize
\begin{eqnarray}
\label{eq:proofposteriorratio1}
&& \frac{\Pr\left[X = x_i | Q(X, V_X) = y, {V}_X = \mathbf{v}\right]}{\Pr\left[X = x_j| Q(X, V_X) = y, {V}_X = \mathbf{v}\right]} \\ \nonumber 
&=& \frac{\Pr\left[X=x_i| {V}_X = \mathbf{v}\right]\Pr\left[Q(X, V_X) = y| X =x_i, {B}_{X}= \mathbf{b}\right]}{\Pr\left[X=x_j|  {V}_X = \mathbf{v}\right]\Pr\left[Q(X, V_X) = y| X = x_j, {B}_{X}= \mathbf{b}\right]} \\ 
&\leq & e^{\epsilon d_{(x_i,\mathbf{b}), (x_j,\mathbf{b}')}} \frac{\Pr\left[X=x_i| {V}_X = \mathbf{v}\right]}{\Pr\left[X=x_j|{V}_X = \mathbf{v}\right]} \\
&\leq & 
e^{\epsilon d_{(x_i, \mathbf{v}), (x_j,\mathbf{v}')}} \frac{p_{x_i|\mathbf{v}}}{p_{x_j|\mathbf{v}}}
\end{eqnarray}
\normalsize
Similarly, we can prove that $\forall  x_i, x_j \in \mathcal{X}$
\normalsize
\begin{eqnarray}
\label{eq:proofposteriorratio2}
e^{-\epsilon d_{(x_i, \mathbf{v}), (x_j,\mathbf{v}')}} \frac{p_{(x_i,\mathbf{b})}}{p_{x_j|\mathbf{b}}}  \leq \frac{\Pr\left[X = x_i | Q(X, V_X) = y, {V} = \mathbf{v}\right]}{\Pr\left[X = x_j| Q(X, V_X) = y, {V} = \mathbf{v}\right]}. 
\end{eqnarray}
\normalsize
According to Eq. (\ref{eq:proofposteriorratio1})--(\ref{eq:proofposteriorratio2}), we can also obtain that 
\normalsize
\begin{eqnarray}
\nonumber  \frac{\Pr\left[X = x_i | Q(X, V_X) = y, {V} = \mathbf{v}\right]}{\Pr\left[X = x_j| Q(X, V_X) = y, {V}= \mathbf{v}\right]}\frac{p_{x_i|\mathbf{v}}}{p_{x_j|\mathbf{v}}} \in [e^{-\epsilon d_{(x_i, \mathbf{v}), (x_j,\mathbf{v}')}}, e^{\epsilon d_{(x_i, \mathbf{v}), (x_j,\mathbf{v}')}}]
\end{eqnarray}
\normalsize
implying that $\forall x_i, x_j \in \mathcal{X}$, $\mathrm{PL}\left[(x_i, \mathbf{v}), (x_j, \mathbf{v}')\right] \leq \epsilon d_{(x_i, \mathbf{v}), (x_j,\mathbf{v}')}$. The proof is completed. 
\end{proof}}

\subsection{Proof of Proposition \ref{propo:CDutility}}

\begin{reproposition}
Given that Assumption \ref{assu:datautility} holds, the loss function $\mathcal{L}(\mathbf{Q})$ in Eq. (\ref{eq:CA-UL}) can be rewritten as the following reduced form:
\begin{equation}
\mathcal{L}(\mathbf{Q}) = \sum_{(x_i, \mathbf{b}), y_k}  p_{(x_i, \mathbf{b})} \cdot c_{(x_i,\mathbf{b}),y_k} \cdot q_{(x_i,\mathbf{b}),y_k}. 
\end{equation}
\end{reproposition}

\begin{proof}
We let $\mathrm{sub}_{\mathcal{B}}(\mathbf{v})$ represent the sub-vector of $\mathbf{v}$ consisting of entries in the Markov blanket space $\mathcal{B}$. Based on \emph{marginalisation of probability} \cite{probability}, the relationship between $q_{(x_i,\mathbf{b}),y_k}$ and  $q_{(x_i,\mathbf{v}),y_k}$ can be expressed through the following equation:  
\begin{equation}
\label{eq:margin1}
\sum_{\mathrm{sub}_{\mathcal{B}}(\mathbf{v}) = \mathbf{b}} q_{(x_i,\mathbf{v}),y_k} = q_{(x_i,\mathbf{b}),y_k}, \forall \mathbf{b}, \mathbf{v}. 
\end{equation}
Given that $p_{x_i|\mathbf{c}} = p_{(x_i,\mathbf{b})}$ and $q_{x_i,y_k|\mathbf{c}} = q_{(x_i,\mathbf{b}),y_k}$, we can obtain that 
\normalsize
\begin{eqnarray}
\textstyle \mathcal{L}(\mathbf{Q}) &=& \sum_{\mathbf{v}}\sum_{x_i, y_k}  p_{(x_i,\mathbf{v})} \cdot c_{(x_i,\mathbf{v}),y_k} \cdot q_{(x_i,\mathbf{v}),y_k} \\ 
&=& \sum_{\mathbf{b}}\sum_{x_i, y_k} \sum_{\mathrm{sub}_{\mathcal{B}}(\mathbf{v}) = \mathbf{b}} p_{(x_i,\mathbf{b})} \cdot c_{(x_i,\mathbf{b}),y_k} \cdot q_{(x_i,\mathbf{v}),y_k} \\ 
&=& \sum_{\mathbf{b}}\sum_{x_i, y_k} p_{(x_i,\mathbf{b})} \cdot c_{(x_i,\mathbf{b}),y_k} \cdot q_{(x_i,\mathbf{b}),y_k}. 
\end{eqnarray}
\normalsize
The proof is completed. 
\end{proof}

\DEL{
\subsection{Proof of Theorem \ref{thm:seqcomposition}}
\label{subsec:proofseqcomposition}
\begin{proof} S1. For each pair of neighboring record $x_i, x_j \in \mathcal{X}$, since $Q_1, ..., Q_M$ are all conditionally independent given ${B}_{X} = \mathbf{b}$, then 
\normalsize
\begin{eqnarray}
\nonumber && \frac{P(Q_1(x_i, \mathbf{b}) = y_{k_1}, ...,Q_M(x_i, \mathbf{b}) = y_{k_M}|{B}_{X}= \mathbf{b})}{P(Q_1(x_j, \mathbf{b}) = y_{k_1}, ..., Q_M(x_j, \mathbf{b}) = y_{k_M}|{B}_{X} = \mathbf{b})} \\
&=& \prod_{m=1}^{M} \underbrace{\left(\frac{P(Q_m(x_i, \mathbf{b}) = y_{k_m}|{B}_{X}= \mathbf{b})}{P(Q_m(x_j, \mathbf{b}) = y_{k_m}|{B}_{X}= \mathbf{b})}\right)}_{\leq e^{\epsilon_m d_{x_i,x_j|\mathbf{b}}}}\\ 
&\leq& e^{(\sum_{m=1}^{M}\epsilon_m) d_{x_i,x_j}|\mathbf{b}} \\
&\leq& e^{(\sum_{m=1}^{M}\epsilon_m) d_{x_i,x_j}|\mathbf{v}}, ~\forall y_{k_1}, ..., y_{k_M} \in \mathcal{Y}
\end{eqnarray}
\normalsize
S2. For each pair of neighboring record $x_i, x_j \in \mathcal{X}$, as $Q_1, ..., Q_M$ are all independent, then 
\normalsize
\begin{eqnarray}
&& \frac{P(Q_1(x_i, \mathbf{b}) = y_{k_1}, ...,Q_M(x_i, \mathbf{b}) = y_{k_M})}{P(Q_1(x_j, \mathbf{b}) = y_{k_1}, ..., Q_M(x_j, \mathbf{b}) = y_{k_M})} \\
&=& \prod_{m=1}^{M} \underbrace{\left(\frac{P(Q_m(x_i, \mathbf{b}) = y_{k_m})}{P(Q_m(x_j, \mathbf{b}) = y_{k_m})}\right)}_{\leq e^{\epsilon_m d_{x_i,x_j|\mathbf{b}}}}\\ 
&\leq& e^{(\sum_{m=1}^{M}\epsilon_m) d_{x_i,x_j|\mathbf{b}}} \\
&\leq & e^{(\sum_{m=1}^{M}\epsilon_m) d_{x_i,x_j|\mathbf{v}}}, ~\forall y_{k_1}, ..., y_{k_M} \in \mathcal{Y}
\end{eqnarray}
\normalsize
The proof is completed. 
\end{proof}
\subsection*{Proof of Lemma \ref{lem:obfindependent}}
\begin{proof}
According to Bayes' formula
\begin{eqnarray}
&& \Pr\left[Q(X, V_X) = y| {V}\right] \\
\nonumber &=& \sum_{x}\Pr\left[Q(X, V_X) = y|X = x, {V}\right]\Pr\left[X = x |{V}\right] \\
&=& \sum_{x}\underbrace{\Pr\left[Q(X, V_X) = y|X = x, {B}_{X}\right]}_{\mbox{since $Q(X, V_X)  \perp \!\!\!\perp {B}_{X}^\mathrm{c} | \left\{X, {B}_{X}\right\}$}}\underbrace{\Pr\left[X = x|{B}_{X}\right]}_{\mbox{since $X  \perp \!\!\!\perp {B}_{X}^\mathrm{c} |{B}_{X}$}} \\
&=& \sum_{x}\Pr\left[Q(X, V_X) = y, X = x| {B}_{X}\right] \\
&=& \Pr\left[Q(X, V_X) = y| {B}_{X}\right]
\end{eqnarray}
indicating that $Q(X, V_X)  \perp \!\!\!\perp {B}_{X}^\mathrm{c} |{B}_{X}$. The proof is completed. 
\end{proof}

Notes: 
\begin{eqnarray}
\frac{\Pr\left[X = x_i | Q(X, V_X),  {B}_{X}\right]}{\Pr\left[X = x_j| Q(X, V_X), {B}_{X}\right]} 
\leq e^{\epsilon d_{x_i, x_j}} \frac{\Pr\left[X = x_i | {B}_{X}\right]}{\Pr\left[X = x_j| {B}_{X}\right]}
\end{eqnarray}}

\DEL{
\subsection{Proof of Theorem \ref{thm:parrallcomposition}}
\label{subsec:proofparrallcomposition}
\begin{proof} 
For each pair of neighboring record $x_i, x_j \in \mathcal{X}_m$, we use $Q_m$ to perturb them. Then, $\forall y_{k} \in \mathcal{Y}$ 
\normalsize
\vspace{-0.00in}
\begin{eqnarray}
\frac{P(Q_m(x_i, \mathbf{b}) = y_{k}|{B}_{X}= \mathbf{b})}{P(Q_m(x_j, \mathbf{b}) = y_{k}|{B}_{X}= \mathbf{b})} 
&\leq& e^{\epsilon_m d_{x_i,x_j|\mathbf{b}}} \\ 
&\leq & e^{(\max_{m=1, ..., M}\epsilon_m) d_{x_i,x_j|\mathbf{b}}} \\ 
&\leq & e^{(\max_{m=1, ..., M}\epsilon_m) d_{x_i,x_j|\mathbf{v}}}. 
\end{eqnarray}
\normalsize
The proof is completed. 
\end{proof}}

\DEL{
\begin{theorem}
\label{thm:PLorder}
By following the CD policy, and satisfying the corresponding mDP constraints:  
\begin{equation}
q_{x_i|\mathbf{b},y_k_X} - e^{\epsilon d_{x_i, x_j}} q_{x_j,y_k|\mathbf{b}_X} \leq 0, ~\forall y_k \in \mathcal{Y}. 
\vspace{-0.00in}
\end{equation}

Given the context variable set ${V}$ and a subset ${V}' \subseteq {V}$, if we apply the following policy to the perturbation function $Q$
\begin{equation}
\label{eq:PLorder}
\Pr\left[Q(X, V_X) | X, {V}'\right] = \Pr\left[Q(X, V_X) | X, {V}\right],
\end{equation}
then enforcing $\Pr\left[Q(X, V_X) | X, {V}'\right]$ to satisfy mDP (as defined in Eq. (\ref{eq:C-mDPred})) is \textbf{sufficient} to achieve the bounded context-aware PL as described in Eq. (\ref{eq:PLcontextbound}). 
\end{theorem}
Intuitively, Theorem \ref{thm:PLorder} states that obfuscating secret data by considering less context information doesn't increase the posterior leakage. Based on Theorem \ref{thm:PLorder}, it is trivial to prove Corollary \ref{cor:CD}, which is a special case of Theorem \ref{thm:PLorder} (when ${V}' = {B}_{X}$).  
\begin{corollary}
\label{cor:CD}
(PL bound guarantee of the CD policy)
If an perturbation function $Q$ satisfies the CD policy (in Definition \ref{def:CD}) and the context-aware mDP constraint, then it can achieve the bounded context-aware PL (in Eq. (\ref{eq:PLcontextbound})). 
\end{corollary}

\begin{eqnarray}
\label{eq:}
\min && \sum_{x_i \in \mathcal{X}, y_k \in \mathcal{Y}, B_X\in \mathcal{V}}  p_{x_i|{B}_X} c_{(x_i, B_X), y_k} q_{x_i,y_k|{B}_X}\\
\mathrm{s.t.} && \mbox{mDP (Eq. (\ref{eq:C-mDPred})) is satisfied} \\
&& \mbox{Unit measure (Eq. (\ref{eq:})) is satisfied} \\ \label{eq:}
&& 0 \leq q_{i,k|{V}} \leq 1, \forall (x_i, y_k) \in \mathcal{X} \times \mathcal{Y}, \forall {V}. 
\end{eqnarray}
where $c_{(x_i, B_X), y_k}$

$\mathcal{L}(\mathbf{Q})$ can be rewritten as 
\begin{equation}
\mathcal{L}(\mathbf{Q}) = \sum_{V_X}\sum_{x_i, y_k}  p_{x_i|{B}_X} c_{(x_i, V_X), y_k} q_{x_i,y_k|{B}_X}
\end{equation}
Due to the limit of space, the detailed proof of Theorem \ref{thm:PLorder}--\ref{thm:parrallcomposition} and Corollary \ref{cor:CD} can be found in Section \ref{subsec:proofPLorder}--\ref{subsec:proofparrallcomposition} in the supplementary file. 
}

\section{Discussion: When Assumption \ref{assu:datautility} Is Removed}
\label{sec:discussion}
In the main part, we focus on the case where the data utility loss depends only on the (conditional) distributions of the secret and perturbed data (Assumption~\ref{assu:datautility}). In particular, letting $\mathbf{b}=\mathrm{sub}_{\mathcal{B}}(\mathbf{v})$ denote the restriction of $\mathbf{v}$ to the Markov-blanket space $\mathcal{B}$, Assumption~\ref{assu:datautility} implies that
\begin{eqnarray}
c_{(x_i,\mathbf{v}),y_k}
&=& h \big(p_{(x_i, \mathbf{v})},\, q_{(x_i,\mathbf{v}),y_k}; \boldsymbol{\theta}\big)
\\ 
&=& h \big(p_{(x_i, \mathbf{b})},\, q_{(x_i,\mathbf{b}),y_k}; \boldsymbol{\theta}\big)
\\ 
&\triangleq& c_{(x_i,\mathbf{b}),y_k}.
\end{eqnarray}
In this section, we discuss how our approach can be adapted when Assumption~\ref{assu:datautility} is removed, i.e., when $c_{(x_i,\mathbf{v}),y_k}$ is not necessarily equal to $c_{(x_i,\mathbf{b}),y_k}$.

\smallskip
We consider a two-step perturbation-optimization framework. In \textbf{Step 1}, we solve the reduced C-mDP under the CD policy and obtain an optimal mechanism
$\{q^*_{(x_i,\mathbf{b}),y_k}\}$ defined over $(x_i,\mathbf{b})$.
In \textbf{Step 2}, we refine $\{q^*_{(x_i,\mathbf{b}),y_k}\}$ to a context-specific mechanism
$\{q_{(x_i,\mathbf{v}),y_k}\}$ by preserving the induced marginals over $\mathbf{b}$:
\begin{equation}
\label{eq:margin}
\sum_{\mathbf{v}:\,\mathrm{sub}_{\mathcal{B}}(\mathbf{v})=\mathbf{b}}
p_{(x_i,\mathbf{v})}\, \cdot q_{(x_i,\mathbf{v}),y_k}
 = 
p_{(x_i,\mathbf{b})}\, \cdot q^*_{(x_i,\mathbf{b}),y_k},
\end{equation}
$\forall x_i\in\mathcal{X},\forall \mathbf{b}\in\mathcal{B}, \forall y_k\in\mathcal{Y}$, where $p_{(x_i,\mathbf{b})}\triangleq \sum_{\mathbf{v}:\,\mathrm{sub}_{\mathcal{B}}(\mathbf{v})=\mathbf{b}} p_{(x_i,\mathbf{v})}$.

\smallskip
\noindent\textbf{Refined C-mDP.}
Given $\{q^*_{(x_i,\mathbf{b}),y_k}\}$ from the reduced problem, we solve the following LP to refine it:
\begin{eqnarray}
\label{eq:C-mDPObjrefine}
\min_{\{q_{(x_i,\mathbf{v}),y_k}\}} && 
\sum_{(x_i,\mathbf{v}),y_k} p_{(x_i,\mathbf{v})}\, \cdot c_{(x_i,\mathbf{v}),y_k}\, \cdot q_{(x_i,\mathbf{v}),y_k} \\[2pt]
\mbox{s.t.} && 
\text{C-mDP constraints in Eq.~(\ref{eq:C-mDP}) are satisfied,} \\[-1pt]
\nonumber && \sum_{\mathbf{v}:\,\mathrm{sub}_{\mathcal{B}}(\mathbf{v})=\mathbf{b}}
p_{(x_i,\mathbf{v})}\, \cdot q_{(x_i,\mathbf{v}),y_k}
=
p_{(x_i,\mathbf{b})}\, \cdot q^*_{(x_i,\mathbf{b}),y_k},
\\
&& \forall x_i,\mathbf{b},y_k, \label{eq:CLPconstraint1refine}\\[-1pt]
&& \sum_{y_k\in\mathcal{Y}} q_{(x_i,\mathbf{v}),y_k}=1,\quad \forall (x_i,\mathbf{v}),\\
&& 0 \le q_{(x_i,\mathbf{v}),y_k}\le 1,\quad \forall (x_i,\mathbf{v}),y_k.
\end{eqnarray}

\begin{proposition}
\label{propo:refine}
Let $\mathcal{L}^*$ be the optimal objective value of the original C-mDP problem
(Eqs.~(\ref{eq:C-mDPObj})--(\ref{eq:CLPconstraint1})).
Let $\mathcal{L}^{\mathrm{ref}}$ be the optimal objective value of the refined problem
(Eqs.~(\ref{eq:C-mDPObjrefine})--(\ref{eq:CLPconstraint1refine})).
Then $\mathcal{L}^{\mathrm{ref}} \ge \mathcal{L}^*$, i.e., the refined problem yields an upper bound on the optimal loss of the original C-mDP.
\end{proposition}

\begin{proof}
The refined problem adds the marginal-preservation constraints in Eq.~(\ref{eq:CLPconstraint1refine}) to the original C-mDP constraints (and keeps the same simplex constraints). Therefore, any feasible solution to the refined problem is feasible for the original problem, i.e., the feasible set of the refined problem is a subset of that of the original problem. Since both problems minimize the same type of linear objective over their respective feasible sets, restricting the feasible set cannot decrease the optimal value. Hence $\mathcal{L}^{\mathrm{ref}} \ge \mathcal{L}^*$.
\end{proof}

\DEL{
\begin{figure}[t]
\centering
\begin{minipage}{0.5\textwidth}
\centering
  \subfigure[\footnotesize $|B_{X_t}| = 1$ and $m = 1$]{
\includegraphics[width=0.48\textwidth, height = 0.155\textheight]{./fig/cluster_rome_d_kmean_GPS}}
\label{}
\centering
  \subfigure[\footnotesize $|B_{X_t}| = 2$ and $m = 2$]{
\includegraphics[width=0.48\textwidth, height = 0.155\textheight]{./fig/cluster_rome_d_kmean_GPS}}
\label{}
\end{minipage}
\vspace{-0.05in}
\caption{xxx}
\label{fig:}
\vspace{-0.05in}
\end{figure}}

\DEL{
\begin{figure}[t]
\centering
\begin{minipage}{0.5\textwidth}
\centering
  \subfigure[\footnotesize Receiver operating characterisic]{
\includegraphics[width=0.48\textwidth, height = 0.155\textheight]{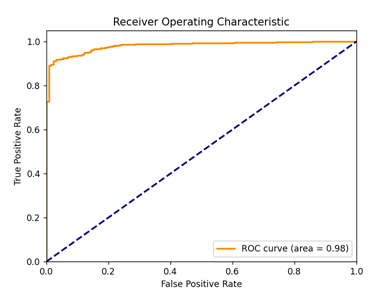}}
\label{}
\centering
  \subfigure[\footnotesize Precision recall curve]{
\includegraphics[width=0.48\textwidth, height = 0.155\textheight]{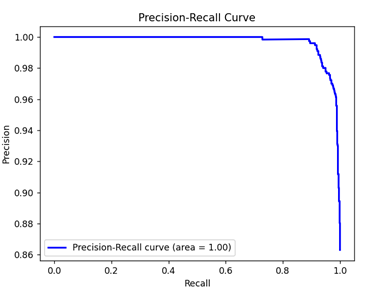}}
\label{}
\end{minipage}
\vspace{-0.05in}
\caption{xxx}
\label{fig:}
\vspace{-0.05in}
\end{figure}}

\DEL{
\section{Detailed Experiment Settings}
In the case study, since vehicles need to physically travel to designated locations to receive desired services, data utility loss can be quantified by the discrepancy between the estimated and actual travel costs to reach the designated locations. Note that our framework is also readily extended to other LBS applications with slight adjustments, provided that a clear relationship between data utility loss and perturbed data can be established. Specifically, we let $q_\ell$ denote the prior probability that the target location is located at the location $v_\ell$ ($l = 1,..., N$). Given a real location $x_n$, the utility loss caused by a perturbed location $v_j$ is defined as the expected error of the traveling costs to the target location, calculated by 
\begin{equation}
\label{eq:utilityloss}
\textstyle 
\Delta c_{x_n,v_j} = \sum_{l=1}^N q_\ell \left|c_{x_n,v_\ell} - c_{v_j,v_\ell}\right|. 
\end{equation}}

\section{Definition of Objective Function $\mathcal{L}(\mathbf{Q})$ in Benchmarks}
\label{sec:measureEUL}
By considering the delay in receiving the recommended destination from the server, when estimating the actual travel cost $c(x_i|\mathbf{v},x_{\mathrm{task}})$, we consider $x_i$ as the vehicle's location in the subsequent time slot, rather than the current location. This requires the vehicle to predict the distribution of the next location when applying the utility-preserving perturbation methods. Specifically, 
\begin{itemize} 
\item (1) ``\emph{LP}'' \cite{Qiu-TMC2020} and 
(3) ``\emph{ConstOPT}'' \cite{ImolaUAI2022}, as context-free perturbation methods, predict the next location using the prior distribution of the vehicles $p_{x_i}$. Hence, their estimated utility loss functions $\mathcal{L}_{\mathrm{LP}}(\mathbf{Q})$ and $\mathcal{L}_{\mathrm{ConstOPT}}(\mathbf{Q})$ are given by
\begin{equation}
\hat{\mathcal{L}}_{\mathrm{LP}}(\mathbf{Q}) = \hat{\mathcal{L}}_{\mathrm{ConstOPT}}(\mathbf{Q}) =\sum_{x_i, y_k} p_{x_i} \cdot c_{x_i,y_k} \cdot q_{x_i,y_k}
\end{equation}
\item (4) ``\emph{LP+Markov}'' assumes that the vehicle's mobility follows a first-order Markov process. Given its current location $x_t$, LP+Markov estimates the probability distribution of the next location as $p_{x_i|x_t}$, which is the estimated transition probability from $x_t$ to $x_i$ of the Markov chain. In this case, its utility loss function $\hat{\mathcal{L}}_{\mathrm{LP+Markov}}(\mathbf{Q})$ is estimated by
\begin{equation}
\hat{\mathcal{L}}_{\mathrm{LP+Markov}}(\mathbf{Q}) = \sum_{x_t} \sum_{x_i, y_k} p_{(x_i,x_t)} \cdot c_{(x_i,x_t),y_k} \cdot q_{(x_i,x_t),y_k}
\end{equation}
\item Our method ``\emph{LP+C-mDP}'' uses the Markov Blanket $\mathbf{b}$ and the corresponding transition probability $p_{(x_i,\mathbf{b})}$ to predict the next location. Accordingly, the utility loss function is estimated by 
\begin{equation}
\hat{\mathcal{L}}_{\mathrm{LP+CmDP}}(\mathbf{Q}) = \sum_{\mathbf{b}}\sum_{x_i, y_k} p_{(x_i,\mathbf{b})} \cdot c_{(x_i,\mathbf{b}),y_k} \cdot q_{(x_i,\mathbf{b}),y_k}
\end{equation}
\end{itemize}

\DEL{
\section{Distance Measure}
\label{sec:additional_detail}
\subsection{Option A (recommended): secret is the current location, history is context}
We treat the secret as the \emph{current} location $X_t \in \mathcal{X}$. Past locations (or other contextual variables) are used only as \emph{context} for selecting/parameterizing the mechanism, e.g., via a Markov blanket $B_t$ extracted from history. The privacy constraints are enforced over pairs of current locations, optionally conditioned on context:
\begin{equation}
q(y \mid x_t, B_t=b) \le \exp\big(\epsilon  d_{\mathrm{loc}}(x_t,x_t')\big) q(y \mid x_t', B_t=b),
\end{equation}
$\forall x_t,x_t'\in\mathcal{X},\ \forall y\in\mathcal{Y},\ \forall b$. Here $d_{\mathrm{loc}}:\mathcal{X}\times\mathcal{X}\to\mathbb{R}_{\ge 0}$ is a metric on locations (e.g., road-network shortest-path distance or Euclidean distance). Since the secret domain remains $\mathcal{X}$, there is no ambiguity arising from different numbers of past locations used in $B_t$.

\subsection{Option B: secret is a trajectory tuple with a product metric.} 
If we instead define the secret as a fixed-length tuple (trajectory window)
\begin{equation}
S_t = (X_t, X_{t-1}, \ldots, X_{t-H}) \in \mathcal{X}^{H+1},
\end{equation}
then we define a metric on $\mathcal{X}^{H+1}$ from a base location metric $d_{\mathrm{loc}}$.
A simple choice is a weighted $\ell_1$ product metric:
\begin{equation}
d(S_t,S_t') = \sum_{h=0}^{H} w_h   d_{\mathrm{loc}}\big(X_{t-h}, X_{t-h}'\big),
\end{equation}
where weights $w_h \ge 0$ control the relative importance of each lag (e.g., $w_0=1$ and $w_h=\alpha^h$ for some $\alpha\in(0,1)$).
The corresponding mDP constraint becomes
\begin{equation}
q(y \mid S_t) \le \exp\big(\epsilon   d(S_t,S_t')\big)  q(y \mid S_t'),
\end{equation}
$\forall  S_t,S_t'\in\mathcal{X}^{H+1},\forall y\in\mathcal{Y}$.

\paragraph{Weighted $\ell_p$ product metric (optional).}
More generally, for $p\ge 1$ we can use a weighted $\ell_p$ product metric:
\[
d_p(S_t,S_t') = \left(\sum_{h=0}^{H} w_h  d_{\mathrm{loc}}\big(X_{t-h}, X_{t-h}'\big)^p\right)^{1/p}.
\]
Then the privacy constraint is
\[
q(y \mid S_t) \le \exp\big(\epsilon  d_p(S_t,S_t')\big) q(y \mid S_t'),
\quad \forall S_t,S_t'\in\mathcal{X}^{H+1},\ \forall y\in\mathcal{Y}.
\]

\paragraph{Handling varying numbers of past locations.}
Even if only a subset of past locations is considered sensitive, we keep a fixed-length secret
$S_t=(X_t,\ldots,X_{t-H})$ and set $w_h=0$ for lags that are not treated as part of the secret.
This yields a well-defined metric on a single space $\mathcal{X}^{H+1}$ while allowing different
effective sensitivities across lags.
}

\section{Supplementary Experimental Details and Analyses}
\label {sec:additional_detail}

\subsection{MBI Training Configuration and Platform}
We use Binary Cross-Entropy (BCE) loss with mean reduction and Adam with learning rate $10^{-5}$ and batch size 128. We apply early stopping if there is no validation improvement for 400 epochs.
Experiments are conducted on a machine with an Intel Core i9-13900F CPU (24 cores, 2.00 -5.60\,GHz), 32\,GB DDR5 memory (4800\,MHz), and an NVIDIA GeForce RTX 4090 GPU (24\,GB GDDR6X VRAM). The training durations on Porto and Rome are 2730.0\,s and 492.4\,s, respectively.

\subsection{Population-level Training vs. Per-vehicle Evaluation (with Overlap)}
\label{subsec: dataset_overlap}
In our experimental pipeline, we use a large city-scale trajectory corpus (e.g., about $330\text{k}$ trajectories in Porto) to train the Markov blanket identification (MBI) module, and we use a smaller subset ($500$ trajectories) to evaluate the proposed perturbation mechanism and report utility/privacy results. Since the evaluation subset is sampled from the same corpus, it may overlap with trajectories used during MBI training.

We view the $330\text{k}$ trajectories as capturing \emph{population-level} mobility regularities of a city. Specifically, the CI-testing outputs and the learned feature-based predictor model how dependency strength (and thus Markov-blanket size) varies with coarse features such as speed, region, and time-of-day. This component is best interpreted as estimating a global statistical structure (a city-level prior), rather than modeling any particular vehicle.

In contrast, the $500$ trajectories are used to evaluate the perturbation mechanism on \emph{specific vehicles} and measure the utility/privacy performance of the released outputs. Even with overlap, the MBI module is used only to provide a city-level prior driven by aggregate regularities, while our formal privacy budget $\epsilon$ is defined for the \emph{online release} $Y=Q(X,V_X)$. Intuitively, because the MBI module is trained to capture coarse, city-level patterns from large-scale data, it is not used as a vehicle-specific descriptor in our evaluation.

\subsection{Neighborhood-Density Analysis}
\label{app:subsec:neighborhood}
\begin{figure}[t]
\centering
\begin{minipage}{0.48\textwidth}
\centering
  \subfigure[Rome]{
\includegraphics[width=0.48\textwidth, height = 0.13\textheight]{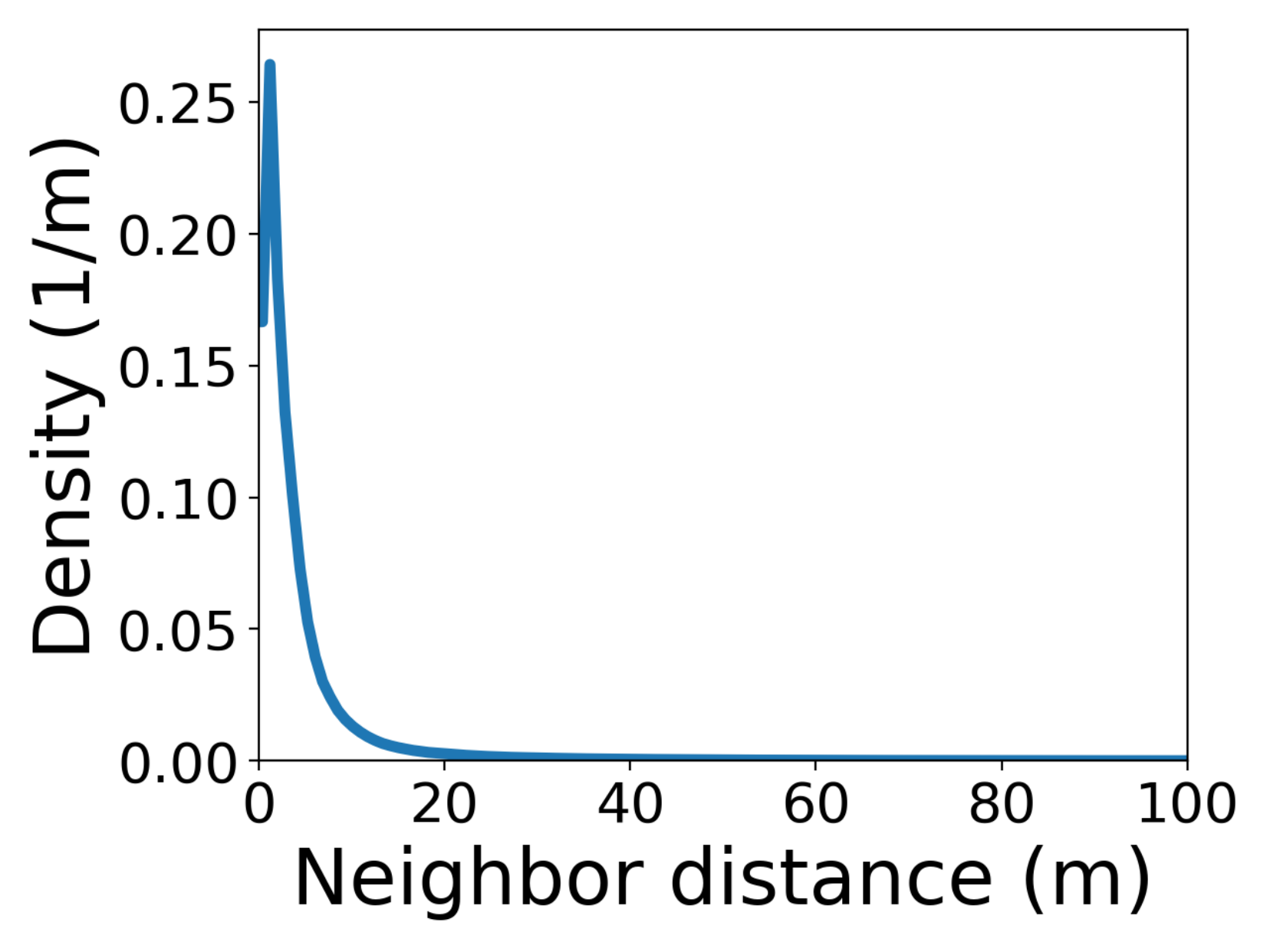}}
  \subfigure[Porto]{
\includegraphics[width=0.48\textwidth, height = 0.13\textheight]{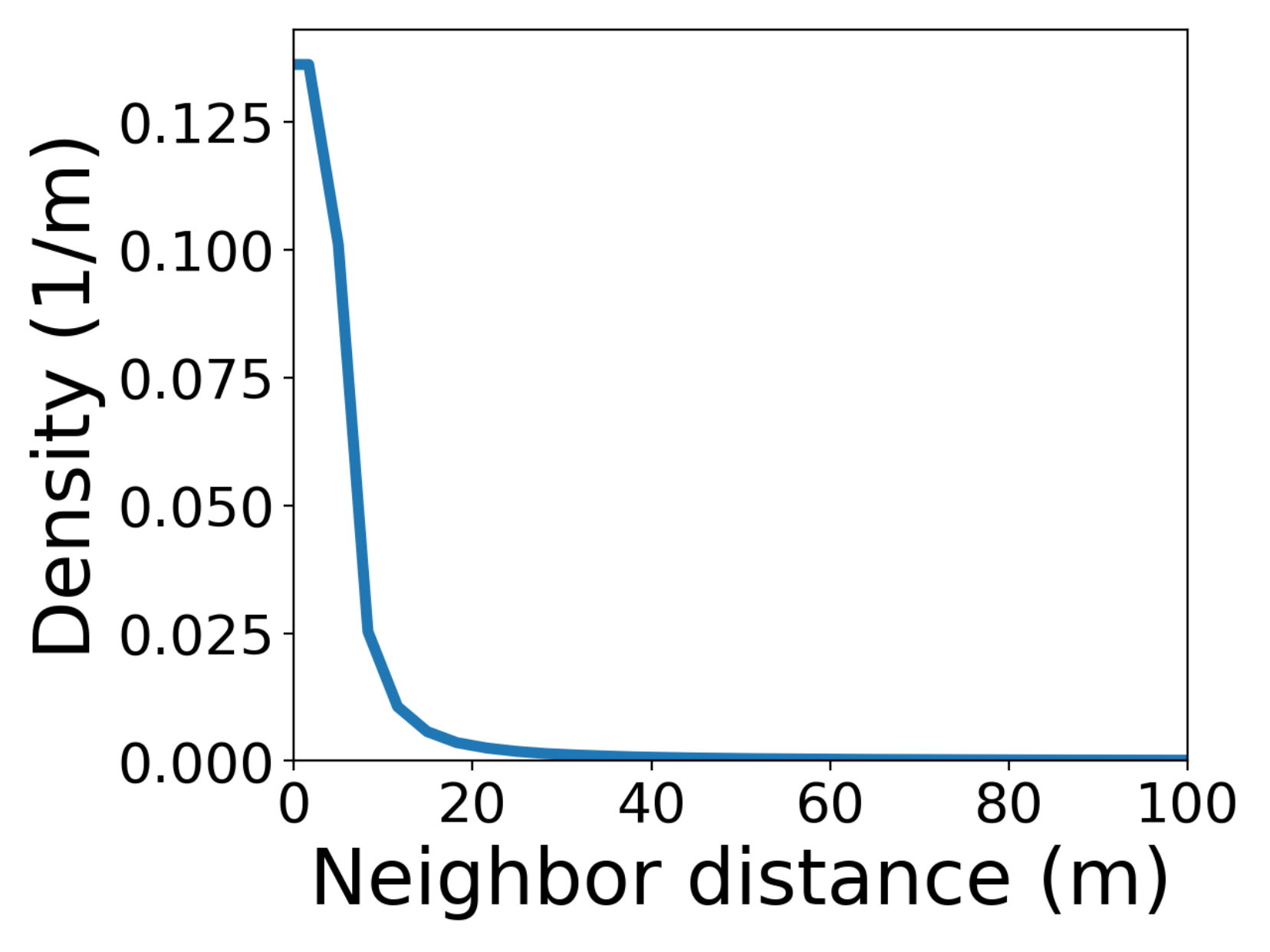}}
\label{}
\end{minipage}
\vspace{-0.15in}
\caption{PDF of neighbor distance.}
\label{fig:neighborDistanceDensity}
\vspace{-0.10in}
\end{figure}

\begin{figure}[t]
\centering
\begin{minipage}{0.48\textwidth}
\centering
  \subfigure[Rome]{
\includegraphics[width=0.48\textwidth, height = 0.13\textheight]{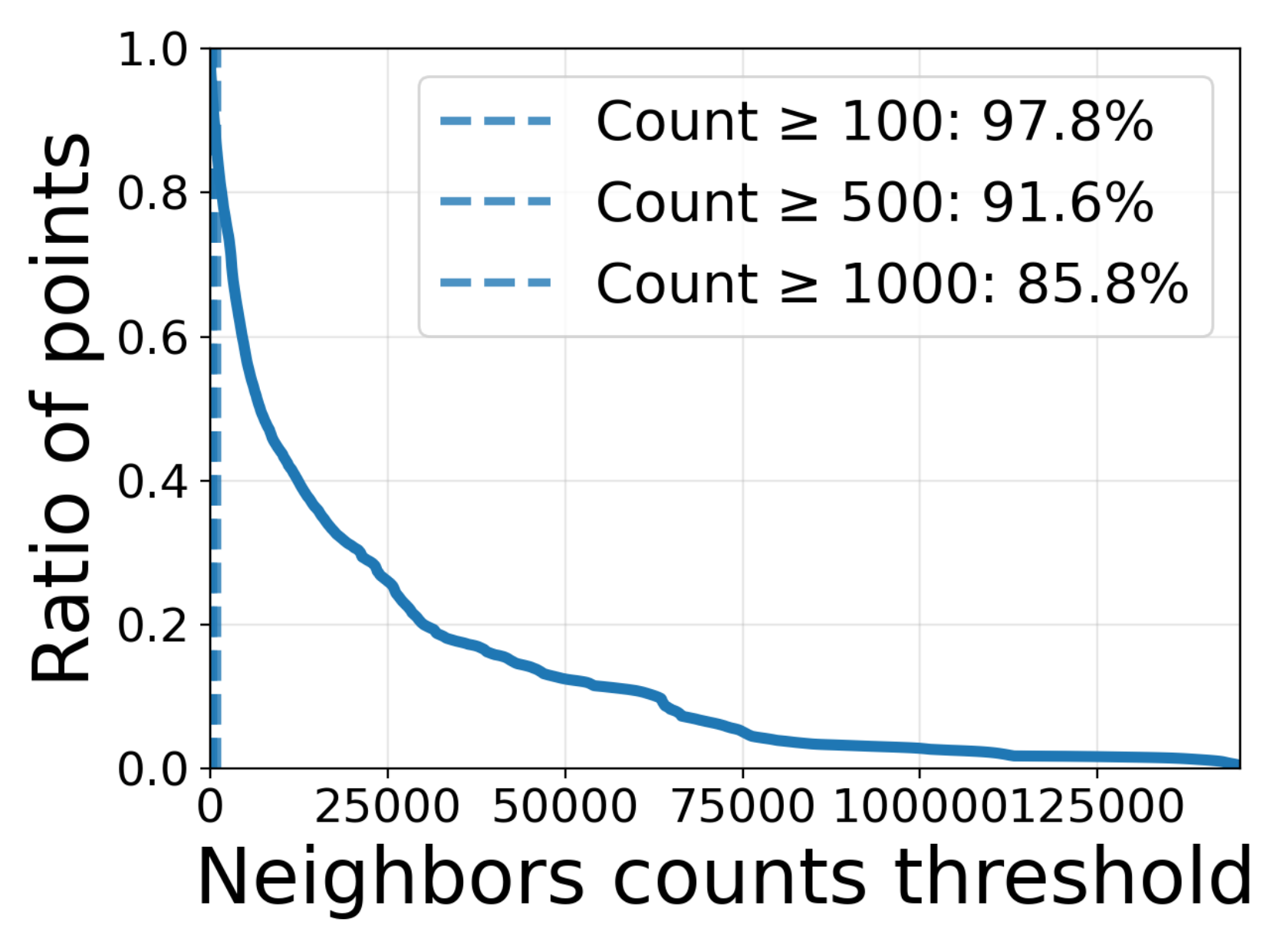}}
  \subfigure[Porto]{
\includegraphics[width=0.48\textwidth, height = 0.13\textheight]{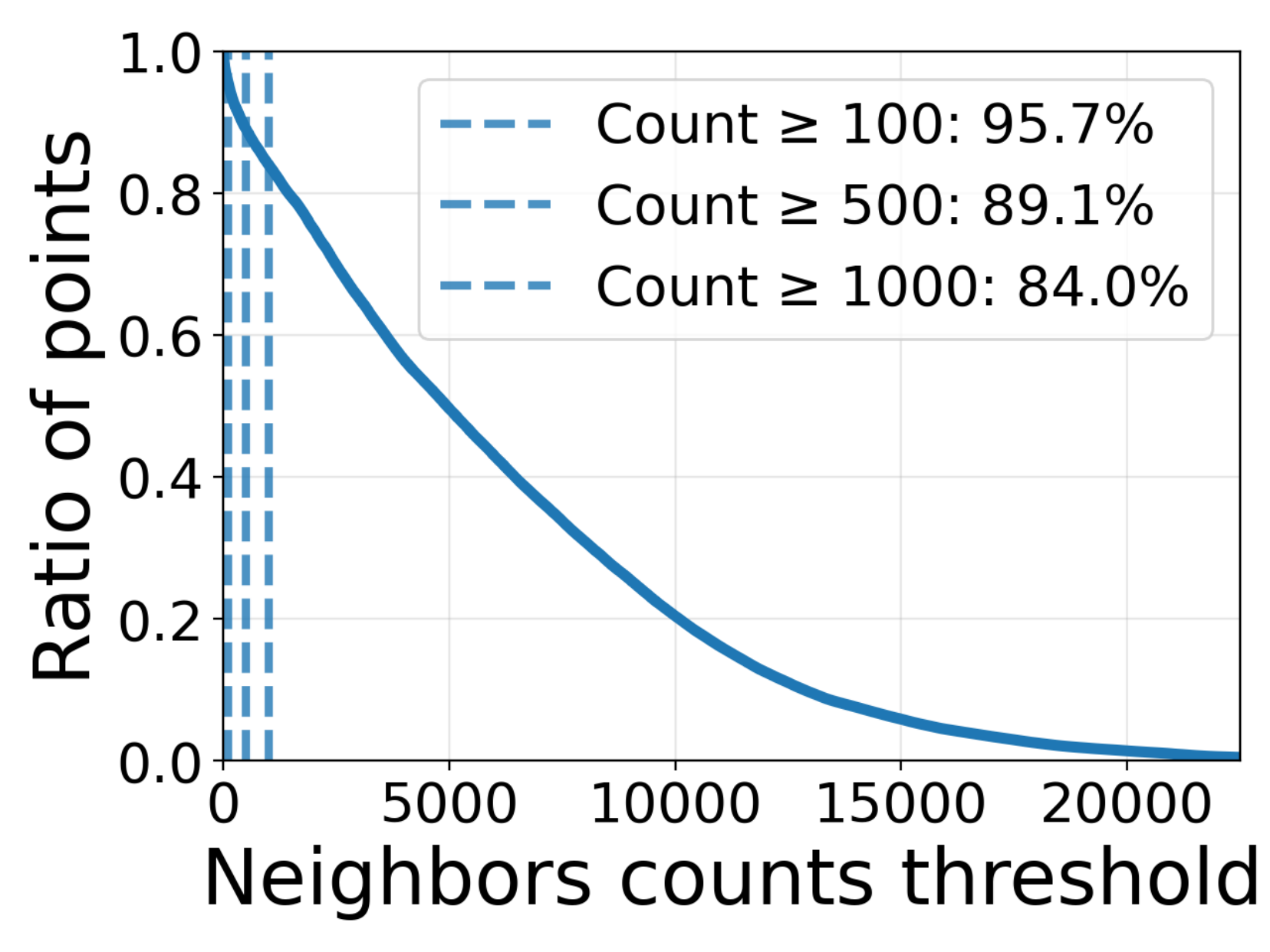}}
\label{}
\end{minipage}
\vspace{-0.15in}
\caption{Ratio of points with count$\geq$threshold vs. Neighbor counts (radius=100m).}
\label{fig:neighborCounts}
\vspace{-0.10in}
\end{figure}

\begin{figure}[t]
\centering
\begin{minipage}{0.48\textwidth}
\centering
  \subfigure[Rome]{
\includegraphics[width=0.48\textwidth, height = 0.13\textheight]{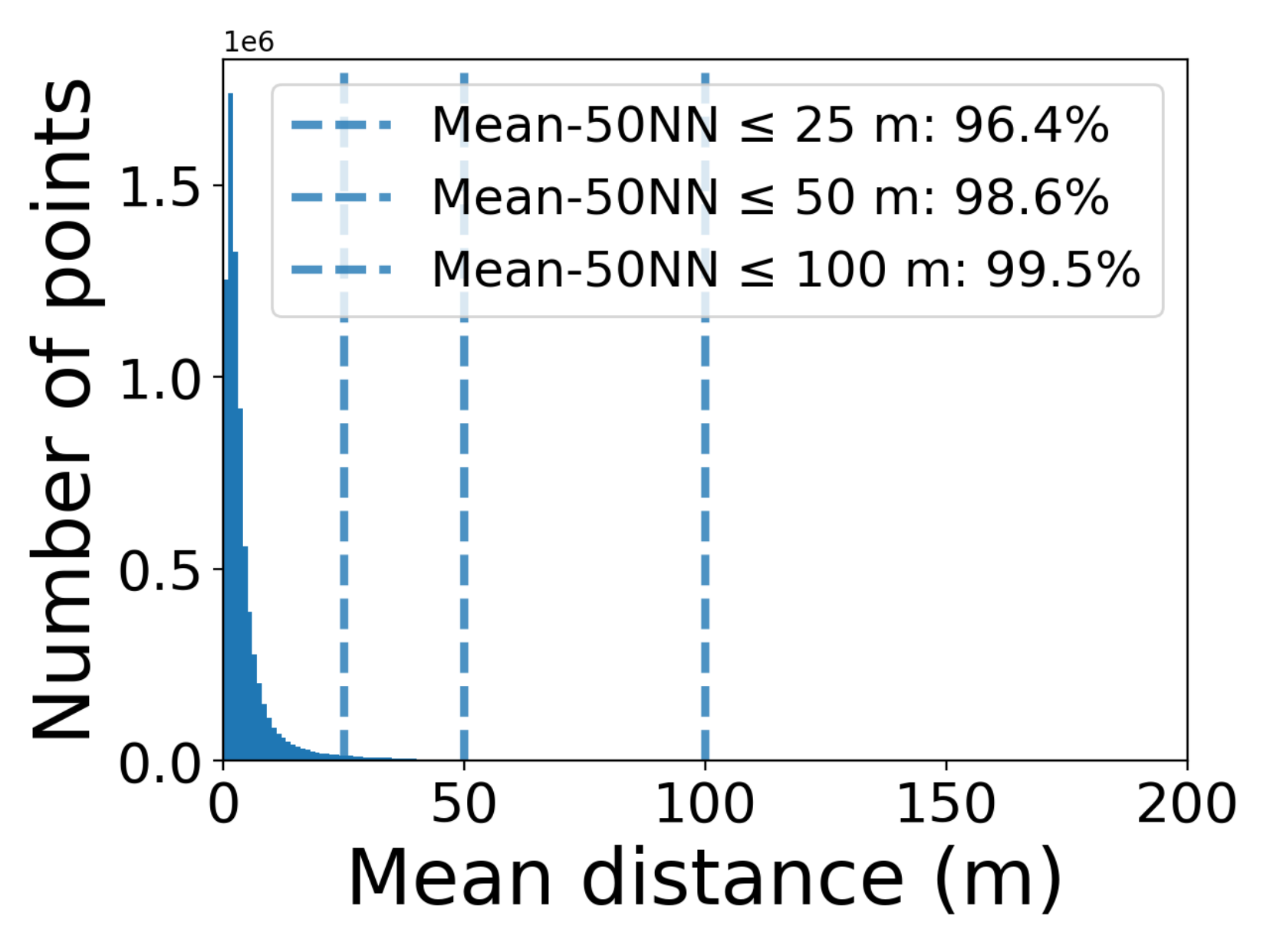}}
  \subfigure[Porto]{
\includegraphics[width=0.48\textwidth, height = 0.13\textheight]{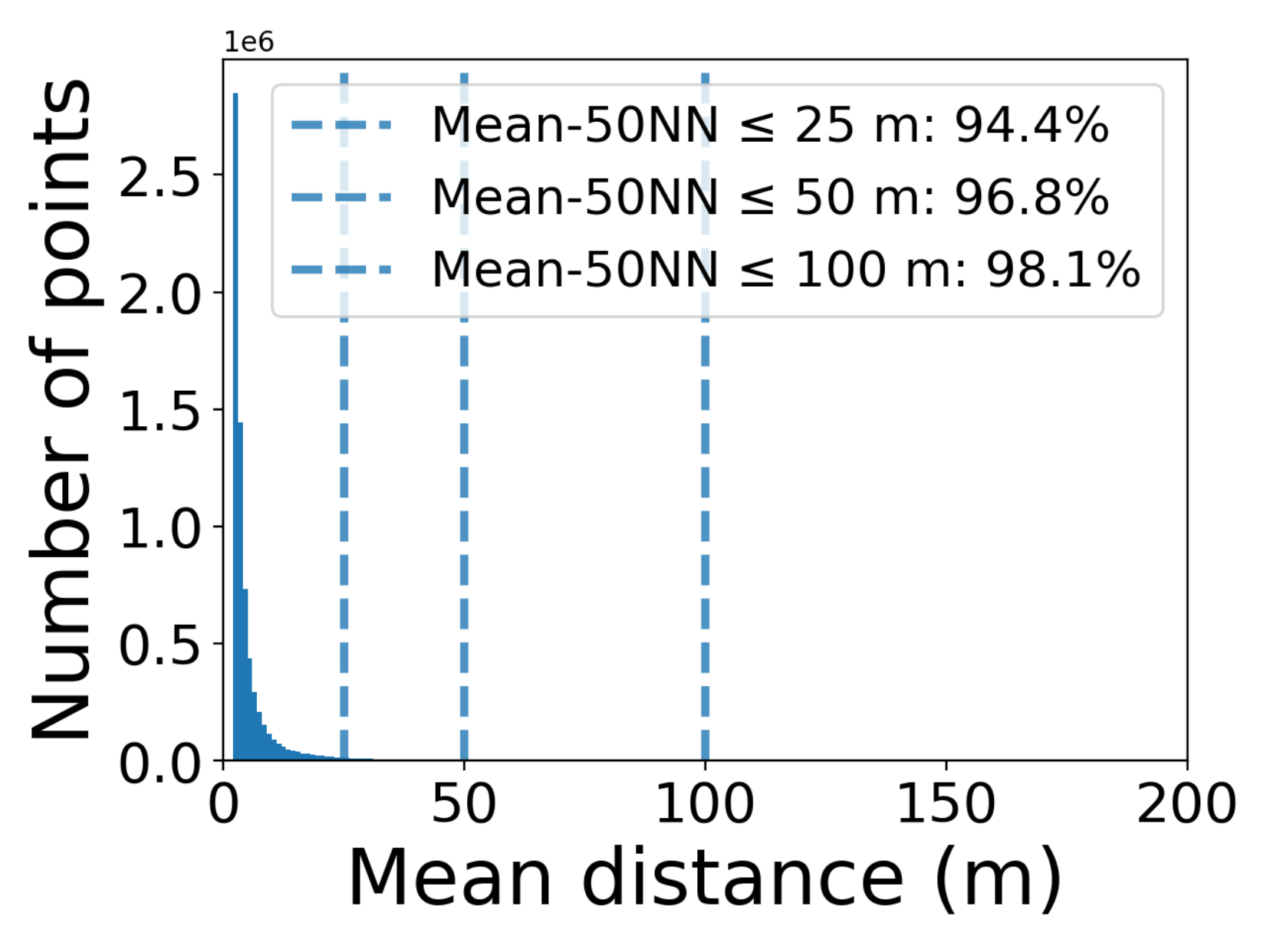}}
\label{}
\end{minipage}
\vspace{-0.15in}
\caption{Distribution of per-point mean-50NN distance.}
\label{fig:neighborDistanceMean}
\vspace{-0.10in}
\end{figure}

We also conduct a neighborhood-density analysis.
Fig. \ref{fig:neighborDistanceDensity}, \ref{fig:neighborCounts}, \ref{fig:neighborDistanceMean} indicate that both Rome and Porto are highly dense location domains. In Fig. \ref{fig:neighborDistanceDensity}, the PDF of neighbor distances is concentrated near zero and decays rapidly, approaching zero around 20 meters, showing that most neighbors lie within a short distance.

Fig. \ref{fig:neighborCounts} shows that each location point has a large number of nearby candidates within a 100\,m radius: $97.8\%$/$95.7\%$ of points (Rome/Porto) have at least 100 neighbors, $91.6\%$/$89.1\%$ have at least 500 neighbors, and $85.8\%$/$84.0\%$ have at least 1000 neighbors.

Fig. \ref{fig:neighborDistanceMean} further shows that these neighbors are spatially close: the mean distance to the 50 nearest neighbors is at most 100\,m for $99.5\%$ of points in Rome and $98.1\%$ in Porto (and at most 50\,m for $98.6\%$ and $96.8\%$, respectively).

Overall, the results confirm that most location points have sufficiently dense local neighborhoods, which ensures that the mDP-based location privacy mechanism can be applied with an adequate set of nearby candidate locations.

\end{document}